\documentclass[11pt,a4paper]{article}
\usepackage[T1]{fontenc}
\usepackage[utf8]{inputenc}
\usepackage{lmodern}

\usepackage[a4paper,margin=2.35cm]{geometry}
\usepackage{amsmath,amssymb,mathtools,bm}
\usepackage{graphicx}
\usepackage{float}
\usepackage{booktabs,tabularx,array}
\usepackage{enumitem}
\usepackage{xcolor}
\usepackage{hyperref}
\usepackage{microtype}
\usepackage{fancyhdr}
\usepackage{tikz-cd}
\usepackage{amsthm}

\newtheorem{theorem}{Theorem}[section]
\newtheorem{corollary}[theorem]{Corollary}
\newtheorem{lemma}{Lemma}[section]
\newtheorem{remark}{Remark}[section]

\definecolor{linkblue}{RGB}{27,79,114}
\definecolor{muted}{RGB}{90,96,102}
\hypersetup{
  colorlinks=true,
  linkcolor=linkblue,
  citecolor=linkblue,
  urlcolor=linkblue,
  pdftitle={A comparison between statistics and 't Hooft anomaly},
  pdfauthor={Hanyu Xue}
}

\newenvironment{gaugetetrahedron}[4]{%
  \begin{tikzpicture}[baseline=(current bounding box.center),
      every node/.style={font=\small,inner sep=2pt},>=latex,
      gauge edge/.style={->,shorten <=2pt,shorten >=2pt}]
    \node[anchor=east] (g0) at (-2.25,0) {$#1$};
    \node[anchor=south] (g1) at (0,1.35) {$#2$};
    \node[anchor=west] (g2) at (2.25,0) {$#3$};
    \node[anchor=north] (g3) at (0,-1.35) {$#4$};
}{%
  \end{tikzpicture}%
}

\newcommand{\gaugetetrahedronedges}[6]{%
    \draw[gauge edge] (g0.east) -- node[above left] {$#1$} (g1.south);
    \draw[gauge edge] (g1.south) -- node[above right] {$#2$} (g2.west);
    \draw[gauge edge] (g2.west) -- node[below right] {$#3$} (g3.north);
    \draw[gauge edge] (g0.east) -- node[below left] {$#6$} (g3.north);
    \draw[gauge edge] (g1.south) -- node[right,pos=0.23] {$#5$} (g3.north);
    \draw[gauge edge,preaction={draw=white,line width=3pt}]
      (g0.east) -- node[below,pos=0.45,anchor=north east] {$#4$} (g2.west);
}

\newcommand{\gaugetriangle}[6]{%
  \begin{tikzpicture}[baseline=(current bounding box.center),
      every node/.style={font=\small,inner sep=2pt},>=latex]
    \node[anchor=east] (g0) at (-2.25,0) {$#1$};
    \node[anchor=south] (g1) at (0,1.35) {$#2$};
    \node[anchor=west] (g2) at (2.25,0) {$#3$};
    \draw[->,shorten >=2pt] (g0.east) -- node[above left] {$#4$} (g1.south);
    \draw[->,shorten <=2pt,shorten >=2pt] (g1.south) -- node[above right] {$#5$} (g2.west);
    \draw[->,shorten >=2pt] (g0.east) -- node[below] {$#6$} (g2.west);
  \end{tikzpicture}%
}

\setlist{nosep,leftmargin=2.2em}

\newcommand{\ZZ}{\mathbb Z}
\newcommand{\RR}{\mathbb R}
\newcommand{\RZ}{\RR/\ZZ}
\newcommand{\U}{\mathrm U}
\newcommand{\ii}{\mathrm i}
\newcommand{\one}{\mathbf 1}
\newcommand{\supp}{\operatorname{supp}}

\newcommand{\Hom}{\operatorname{Hom}}

\newcommand{\ket}[1]{\lvert #1\rangle}
\newcommand{\bra}[1]{\langle #1\rvert}
\newcommand{\confstate}[1]{\ket{#1}^{\mathrm{conf}}}
\newcommand{\e}{\mathrm e}
\newcommand{\dd}{\mathrm d}
\newcommand{\Alt}{\operatorname{Alt}}
\newcommand{\Sg}{\mathcal S}
\newcommand{\conj}[2]{{}^{#1}\!#2}

\title{\bfseries A comparison between statistics and 't Hooft anomaly}
\author{Hanyu Xue\\[0.35em]\normalsize Massachusetts Institute of Technology\\[0.2em]\normalsize\texttt{xhy2002@mit.edu}}
\date{October 12, 2026}

\makeatletter
\newcommand{\embeddedTjunction}{%
}
\makeatother

\begin{document}

\maketitle

\begin{abstract}
Anomalies of generalized symmetries are closely related to the statistics of topological excitations. Truncating a symmetry transformation to a finite region produces an operator that creates symmetry defects on its boundary; we call this a symmetry patch operator. Such operators resemble hopping operators that create topological excitations at the boundary of their support, but this geometric similarity does not justify identifying the two. For finite Abelian groups in one spatial dimension, we prove that if the same family of local operators defines both valid hopping operators and compatible symmetry patches, the type-III component of the anomaly is $-2$ times that of the statistics. We establish a one-to-one correspondence between statistics and anomaly by treating hopping operators as symmetric operators that commute with the global symmetry, following the description of symmetry through the algebra of local symmetric operators. We explain this correspondence through gauging by coupling boundary matter to a Dijkgraaf--Witten gauge field in one higher spatial dimension. Hopping operators act on the gauge-invariant degrees of freedom. After fixing the background gauge field, the residual gauge transformations become the global symmetry, and the gauge invariance of hopping operators becomes the requirement that they commute with this symmetry.
\end{abstract}

\tableofcontents

\section{Introduction}
\label{sec:intro}

Symmetry and topological excitations are often discussed together in the study of topological order. Two central concepts associated with them are symmetry anomalies and excitation statistics. Classic field-theoretic examples of quantum anomalies include the 1969 studies of chiral-current conservation by Adler and by Bell and Jackiw~\cite{Adler1969AxialAnomaly,BellJackiw1969PCAC}. The 't Hooft anomaly considered here more specifically originates from the anomaly-matching principle proposed by 't Hooft in 1980: a global symmetry may hold in the original theory while admitting no consistent gauging, and this obstruction must match along the renormalization-group flow~\cite{tHooft1980AnomalyMatching}. This viewpoint was subsequently extended to higher-form symmetries~\cite{GaiottoEtAl2015GeneralizedSymmetries}.

Statistics, on the other hand, originates from bosonic, fermionic, and more general anyonic quasiparticle statistics~\cite{Wilczek1982FractionalSpin,Kitaev2003ToricCode,Kitaev2006Anyons}, and has been generalized to extended topological excitations in higher dimensions~\cite{WangLevin2014LoopBraiding,Fidkowski_2022,Kobayashi_2026,Xue2026Statistics,Feng_2026}. Recent developments in generalized symmetries and extended excitations have revealed a common classification and potential connections between 't Hooft anomalies and statistics~\cite{GaiottoEtAl2015GeneralizedSymmetries,FengEtAl2025HigherFormAnomalies}. However, establishing this connection involves pitfalls. Conclusions drawn from the simplest models may describe special cases rather than universal relations between the two concepts.

Consider the toric code~\cite{Kitaev2003ToricCode}. From the symmetry viewpoint, $X$-loop operators commute with every term in the Hamiltonian, and together form a $1$-form global symmetry. From the excitation viewpoint, an $X$-string operator supported on an open string is a hopping operator for the $m$ anyon.
This suggests the following correspondence: truncating a symmetry transformation to a patch produces a hopping operator for topological excitations; conversely, a hopping operator with closed support is a symmetry transformation.

This idea is not uncommon in the literature. For example, version v3 of Ref.~\cite{HsinChen2026Bockstein} states explicitly at the beginning of Section~3:
\begin{quote}
...placing the excitation (hopping) operator on a closed support turns it into a generalized-symmetry generator.
\end{quote}
Similarly, Section~V of Ref.~\cite{Kobayashi_2026} claims that truncating a global symmetry to a patch constructs hopping operators whose statistics precisely reflect the 't Hooft anomaly of that symmetry.

Using the same definitions as Refs.~\cite{Kobayashi_2026,HsinChen2026Bockstein}, we \textbf{rigorously disprove this claim}. Consider anomalies and domain-wall $F$-symbol statistics in one-dimensional systems, both classified by $H^3(BG,\RZ)$ for a finite Abelian group $G$. In Theorem~\ref{thm:typeiii-inverse-square}, we prove that \textbf{if the same formula constructs both symmetry transformations and hopping operators, the type-III part of its anomaly must be $-2$ times the type-III part of its statistics}. This contradicts the general identification of anomaly with statistics asserted in Ref.~\cite{Kobayashi_2026}. Type-III is defined in Section~\ref{subsec:patch-D}. Briefly, the type-III part of $H^3(BG,\RZ)$ is the mixed contribution from three independent generators of $G$. For example, the type-III part of $H^3(B\ZZ_2^3,\RZ)$ is $\ZZ_2$, generated by $\frac12a_1a_2a_3\in H^3(B\ZZ_2^3,\RZ)$. In particular, for a global symmetry with anomaly $\frac12a_1a_2a_3$, the type-III part is not divisible by two: none of its truncations can be valid hopping operators, and their statistics cannot be defined at all. Appendix~\ref{app:typeiii-constructions} gives an explicit lattice realization of this $\ZZ_2^3$ anomaly and exhibits the contradiction directly in the model.

This theorem does not deny a connection between statistics and anomaly. It shows that identifying symmetry patch operators with hopping operators is not the correct way to establish their one-to-one correspondence. In the toric code, both are $X$-strings that shift the matter field without an accompanying quantum phase. In the presence of a twist, however, symmetry transformations and hopping operators acquire different quantum phases, and these phases carry the anomaly and statistics. \textbf{Geometric similarity conceals a conceptual distinction that becomes apparent only when the quantum phases are taken into account}.

We first outline the definitions of the two theories; detailed discussions appear in Sections~\ref{sec:else-nayak} and~\ref{sec:statistics}. A global symmetry $G$ of a one-dimensional lattice system is a family of finite-depth local quantum circuits $\{\mathcal S(g)\mid g\in G\}$ satisfying $\mathcal{S}(g)\mathcal{S}(h)=\mathcal{S}(gh)$. Restricting the symmetry to an interval $I$ requires truncating these operators to $\{P_I(g)\mid g\in G\}$, with freedom in the truncation at either endpoint. We call these symmetry patch operators and say that they create symmetry defects at the endpoints. The truncated operators generally satisfy the group law only up to endpoint corrections:
\begin{equation}\label{eq: symmetry patch multiplication}
	P_I(g)P_I(h)=\Omega_L(g,h)\Omega_R(g,h)P_I(gh),
\end{equation}
Here $L,R$ are the endpoints of $I$, and $\Omega_L(g,h)$ and $\Omega_R(g,h)$ commute. We may regard $\Omega_L(g,h)$ as the fusion of symmetry defects at $L$. Its associator determines a cohomology class in $H^3(BG,\U(1))$, called the Else--Nayak index~\cite{ElseNayak2014}. This class is independent of the choice of truncation $P_I(g)$ and is thus an invariant of the global symmetry $\{\mathcal S(g)\mid g\in G\}$ alone. It is precisely the 't Hooft anomaly, or obstruction to gauging, of the global symmetry~\cite{GaiottoEtAl2015GeneralizedSymmetries}. This notion extends to $(q-1)$-form symmetries in $d$ spatial dimensions~\cite{GaiottoEtAl2015GeneralizedSymmetries}. For example, Ref.~\cite{FengEtAl2025HigherFormAnomalies} argues that lattice truncations of a global symmetry, together with coherence data for their composition, yield a 't Hooft anomaly classified by $H^{d+2}(B^qG,\RZ)$, where $B^qG=K(G,q)$ is an Eilenberg--MacLane space.

We next describe topological excitations and statistics. Geometrically, a hopping operator supported on a $(p+1)$-cell creates excitations on its $p$-dimensional boundary; a string operator creating quasiparticles at its endpoints is a familiar example. We will find it more convenient to use the dual-cell description, with $q=d-p$: a hopping operator associated with a $(q-1)$-cell creates excitations on its $q$-dimensional coboundary. Levin and Wen already characterized lattice-particle statistics through algebraic relations among hopping operators in Ref.~\cite{LevinWen2003}. In Refs.~\cite{KawagoeLevin2020,Fidkowski_2022}, hopping operators $U(s)$ supported on $s$ and their inverses $U(s)^{-1}$ are applied in a specified order to an excitation state $\confstate{a}$, returning to the same state and producing the Berry phase
\begin{equation}\label{eq: intro statistical phase}
	{}^{\mathrm{conf}}\!\bra{a}\,
	U(s_n)^{\varepsilon_n}\cdots U(s_1)^{\varepsilon_1}
	\confstate{a}\in \U(1),\qquad \varepsilon_i=\pm 1.
\end{equation}
This process has both geometric and quantum aspects: each hop changes the excitation configuration by $a\mapsto a+\dd(\varepsilon_i s_i)$ and attaches a phase to the quantum state. Even when the configuration returns to its starting point, a nontrivial Berry phase may remain. For suitably designed sequences, the total phase is invariant under specified local perturbations of $U(s)$ and therefore defines a statistical invariant.

In \textit{Statistics of Abelian Topological Excitations}~\cite{Xue2026Statistics}, we abstracted the conditions in Refs.~\cite{KawagoeLevin2020,Fidkowski_2022} that ensure well-defined statistical phases and stability under local perturbations into axioms, obtaining a mathematically rigorous definition of topological excitations. We specify configuration states and hopping operators directly, rather than solving for them from a Hamiltonian. For each allowed excitation configuration $a$, we choose the state line spanned by a normalized state $\confstate a$ and require each hopping operator to satisfy
\[
 U(s)\confstate a
 =\e^{2\pi\ii\theta(s,a)}\confstate{a+\dd s},
 \qquad \theta(s,a)\in\RZ.
\]
This is the \textbf{configuration axiom}: up to a phase, the action of a hopping operator on configuration states must agree with its geometric meaning, $a\mapsto a+\dd s$. This axiom is a prerequisite for Eq.~\eqref{eq: intro statistical phase} to make sense and define a statistical phase. A separate locality axiom characterizes locality through commutation relations among hopping operators at different positions. Together, these two axioms define the topological excitations considered here.

We place the lattice model on a triangulated sphere $X\simeq S^d$, with a finite Abelian fusion group $G$ and excitation codimension $q$. For fixed geometric data, the solutions of the excitation axioms form a topological moduli space $\mathcal{M}^q(X,G)$, and we define statistics to be its connected components $\pi_0\mathcal{M}^q(X,G)$. One can prove that, for any triangulation $X$ of $S^d$,
\begin{equation}\label{eq: classification of statistics}
	\pi_0\mathcal{M}^q(X,G)\simeq H^{d+2}(B^{q}G,\RZ)
\end{equation}
This approach gives both the notion of statistics and its classification by $H^{d+2}(B^{q}G,\RZ)$.\footnote{We have completed the full proof and verified it formally in Lean; the manuscript is in preparation. Refs.~\cite{Xue2026Statistics,XueWen2026Holographic} prove several special cases.} For $d=q=2$, the classification $H^4(B^2G,\RZ)$ agrees with that of braided fusion categories describing Abelian anyon braiding and half-braiding~\cite{EilenbergMacLane1954,JoyalStreet1993}.

Although both symmetry patch operators and hopping operators act on $(p+1)$-dimensional disks and create defects or excitations on their boundaries, they obey different requirements. A symmetry patch must implement the group multiplication in its interior, while its boundary truncation can be chosen freely. A hopping operator must act consistently with the geometric shift $a\mapsto a+\dd s$ on a family of configuration states, but may otherwise be chosen freely within the axioms. Neither set of requirements or choices contains the other: a symmetry patch obtained by truncation need not satisfy the configuration axiom. The conflict can be so strong that no special choice of symmetry patch resolves it. If a one-dimensional symmetry has a type-III anomaly that is not divisible by two, none of its truncations can coexist with the configuration axiom.

To prove this, we construct a one-dimensional operator sequence of the form $U(s_n)^{\varepsilon_n}\cdots U(s_1)^{\varepsilon_1}$; see Section~\ref{sec:typeiii-configuration}. When the $U(s)$ are patches of a $G$ symmetry, this expression gives the type-III part of the Else--Nayak index. When they are hopping operators with fusion group $G$, it gives the type-III part of the statistics with an additional factor of $-2$. Besides the central result, Theorem~\ref{thm:typeiii-inverse-square}, this proof yields new identities for both type-III anomaly and type-III statistics that are of independent interest. We present them in Sections~\ref{subsec:patch-D} and~\ref{subsec:typeiii-statistical-process}, respectively.

If identifying hopping operators with symmetry patches does not give a one-to-one correspondence between statistics and anomaly, how should that correspondence be established? We should \textbf{regard hopping operators as symmetric operators that commute with the global symmetry}. Symmetry constrains the motion of excitations, while hopping retains independent freedom in its interior action and quantum phase; see Section~\ref{sec:relation}. In their one-dimensional construction~\cite{KawagoeLevin2021Anomalies}, Kawagoe and Levin start with domain walls in a symmetry-breaking background, use movement and splitting operators that commute with the global symmetry, and show that the resulting $F$-symbol gives its anomaly. They clearly distinguish operators that move domain walls from symmetry patch operators. We extend a similar relation to $(q-1)$-form symmetries in $d$ spatial dimensions and analyze it through gauging. Section~\ref{sec:what-is-hopping} couples boundary matter to a DW gauge field in one higher dimension. Physical hopping acts on the gauge-invariant degrees of freedom of the coupled system. After fixing the background gauge field, the residual gauge transformations that preserve it become the global symmetry; the original gauge invariance of hopping becomes its commutation with this symmetry.

The relation among the Else--Nayak index, statistics, and anomaly reflects two viewpoints on symmetry.
Symmetry can be understood either through transformations satisfying a group law or through the algebra of local symmetric operators~\cite{ChatterjeeWen2023}. The first viewpoint studies symmetry transformations and their local truncations, leading in one dimension to the Else--Nayak index. The second studies local operations allowed by the symmetry, using hopping operators to describe excitation motion and extracting statistics from these processes. In our construction, the configuration axiom is precisely equivalent to the symmetric condition, linking the two approaches to the same anomaly. \textbf{These are parallel theories, each with its own operator requirements and freedom of choice.} Ignoring these distinctions and identifying their operators directly leads to the incorrect conclusions criticized here. The following table compares the two approaches.
\begin{center}
	\small
	\begin{tabularx}{0.96\textwidth}{@{}>{\bfseries\raggedright\arraybackslash}p{0.18\textwidth}
			>{\raggedright\arraybackslash}p{0.27\textwidth}
			>{\raggedright\arraybackslash}p{0.24\textwidth}
			>{\raggedright\arraybackslash}X@{}}
		\toprule
        Operator role & Relation to $\mathcal S$ & Invariant & Allowed perturbations\\
		\midrule
		symmetry patch operator
        & Implements $\mathcal S$ in the interior; need not commute with $\mathcal S$
		& Else--Nayak index
		& boundary perturbation\\
		 hopping operator
        & Commutes with $\mathcal S$; interior action is not fixed by $\mathcal S$
		& statistics
		& symmetric perturbation\\
		\bottomrule
	\end{tabularx}
\end{center}

We conclude the introduction with further observations on identifying symmetry patches with hopping operators. Besides geometric similarity, another source of confusion is the unchecked extrapolation of anyon statistics. In Section~\ref{sec:anyon-statistics-anomaly}, we construct a special family of anyon hopping operators that satisfies the group law on closed supports and thus defines a $1$-form symmetry. Under these conditions, anomaly and statistics correspond one-to-one, but with an opposite sign. This sign is physically meaningful, not a convention. In the preceding construction with commuting symmetry and hopping operators, we set both anomaly and statistics equal to $[\omega]$, fixing their relative sign. Directly closing the hopping operators instead gives anomaly $-[\omega]$. When $2[\omega]\ne0$, these are different anomaly classes. The factor of $-2$ between type-III anomaly and statistics discussed above has the same origin.

Beyond anyons, even this correspondence with an opposite sign need not survive. In the one-dimensional $\ZZ_2$ example of Section~\ref{sec:relation}, statistics and the Else--Nayak index can vary independently in $\ZZ_2$. In that example, imposing the additional requirement that hopping be both a symmetry patch and commute with the corresponding global symmetry does ensure equality of statistics and anomaly:
\begin{equation}
 P_I(g)\mathcal S(h)=\mathcal S(h)P_I(g),\qquad g,h\in G.
 \label{eq:symmetric truncation}
\end{equation}
However, the mechanism behind this repair is the symmetric condition. Requiring hopping to be a symmetry patch as well is \textbf{``adding legs to a snake''}---a superfluous addition. The factor of $-2$ for one-dimensional type-III and the factor of $-1$ for two-dimensional anyons already show that, under this identification, statistics $=$ anomaly generally has no solution. Even if we grant that the two are equal, the symmetric condition may still be impossible to realize; see Section~\ref{sec:double-role-operators}.

\paragraph{Notation and conventions.}

\begin{itemize}
    \item Unless stated otherwise, we consider spatial lattice models rather than spacetime models, and denote the spatial dimension by $d$. Section~\ref{sec:statistics} first gives a chain description on an arbitrary finite cell complex $Y$. Thereafter, and throughout the rest of the main text, space is a triangulated $d$-sphere denoted by $X$.
    \item Except where stated otherwise in Section~\ref{sec:else-nayak}, we consider symmetries or fusion rules described by a finite Abelian group $G$.
    \item Throughout, the group commutator is $[A,B]:=A^{-1}B^{-1}AB$.
    \item $C^k(X,G)$, $Z^k(X,G)$, and $B^k(X,G)$ denote simplicial cochains, cocycles, and coboundaries on $X$, respectively, and $\dd$ is the simplicial coboundary. We use $s,v$ for cochains and $z$ for cocycles.
\end{itemize}

\section{Symmetry transformations and the Else--Nayak index}
\label{sec:else-nayak}

In this section, we consider a symmetry $\mathcal{S}$ described by a finite group $G$ on a one-dimensional lattice. Truncating it to a finite interval gives a symmetry patch operator $P$. A nontrivial anomaly causes the group multiplication law to fail, $P(g)P(h)\ne P(gh)$, with local operators at the endpoints providing correction factors. The multiplication rules of these endpoint factors determine the Else--Nayak index, which characterizes the 't Hooft anomaly of $\mathcal{S}$.

When $G$ is Abelian, we can similarly analyze the failure of commutativity under truncation, $P(g)\Sg(h)\ne\Sg(h)P(g)$, introduce endpoint factors, and extract an index from their multiplication rules. We prove that this index precisely detects whether the symmetry has a type-III anomaly; see Theorem~\ref{thm:endpoint-exchange}. This theorem is of independent interest. In the structure of this paper, it constrains the configuration states required for a symmetry patch to be a hopping operator and yields counterexamples that cannot be repaired by locally redefining the endpoints; see Section~\ref{sec:typeiii-configuration}.

\subsection{Symmetry patches and endpoint multiplication factors}
\label{subsec:symmetry-patch-endpoint}

We consider a lattice system on a finite circle $S^1$.
Its Hilbert space $\mathcal H$ has a tensor-product decomposition over the lattice sites,
and the global symmetry of a finite group $G$ satisfies
\begin{equation}
 \Sg(g)\Sg(h)=\Sg(gh),\qquad g,h\in G.
 \label{eq:global-group-law}
\end{equation}
We assume that each $\Sg(g)$ is implemented by a finite-depth local quantum circuit.
In this subsection, $G$ may be non-Abelian.

Fix an interval $I$, with left and right endpoints denoted by $L,R$.
A symmetry patch operator $P(g)$ acts as $\Sg(g)$ in the interior of $I$
and acts trivially outside the interval away from its endpoints.
Here, having the same action means inducing the same conjugation on local operators in the region in question.
We use the notation
\[
 \conj{U}{O}:=UOU^{-1}.
\]
The endpoints are separated far enough that their supports remain disjoint after all the conjugations used below.
We suppress the interval subscript of $P$ and choose $P(e)=1$, where $e$ is the identity of $G$.

The group law in the interval interior implies that $P(g)P(h)$ and $P(gh)$ differ only at the two endpoints.
The local circuit structure further factorizes this difference into two endpoint operators:
\begin{equation}
 P(g)P(h)=\Omega_L(g,h)\Omega_R(g,h)P(gh).
 \label{eq:patch-fusion}
\end{equation}
The operators $\Omega_L$ and $\Omega_R$ are supported near $L,R$, respectively.
We may choose $\Omega_x(e,g)=\Omega_x(g,e)=1$, where $x=L,R$.

We now use Eq.~\eqref{eq:patch-fusion} to evaluate the product of three patches.
The two ways of inserting parentheses give
\begin{subequations}
 \label{eq:patch-associations}
 \begin{align}
 (P(g)P(h))P(k)
 &=\bigl[\Omega_L(g,h)\Omega_L(gh,k)\bigr]
   \bigl[\Omega_R(g,h)\Omega_R(gh,k)\bigr]P(ghk),
 \label{eq:patch-associate-left}\\
 P(g)(P(h)P(k))
 &=\bigl[\conj{P(g)}{\Omega_L(h,k)}\Omega_L(g,hk)\bigr]
 \notag\\[-.3em]
 &\quad\cdot
   \bigl[\conj{P(g)}{\Omega_R(h,k)}\Omega_R(g,hk)\bigr]P(ghk).
 \label{eq:patch-associate-right}
 \end{align}
\end{subequations}
The two bracketed factors in each equation are the corrections at $L,R$, respectively, both multiplying $P(ghk)$ from the left.
Associativity of the complete operators requires Eq.~\eqref{eq:patch-associate-left}
to equal Eq.~\eqref{eq:patch-associate-right}.
The left-end factors in the two equations can therefore differ only by a pure phase, and the right-end factors differ by its inverse.
This defines $\nu(g,h,k)\in\U(1)$:
\begin{subequations}
 \label{eq:EN-cocycle}
 \begin{align}
 \Omega_L(g,h)\Omega_L(gh,k)
 &=\nu(g,h,k)\,
   \conj{P(g)}{\Omega_L(h,k)}\Omega_L(g,hk),
 \label{eq:EN-left}\\
 \Omega_R(g,h)\Omega_R(gh,k)
 &=\frac{1}{\nu(g,h,k)}\,
   \conj{P(g)}{\Omega_R(h,k)}\Omega_R(g,hk).
 \label{eq:EN-right}
 \end{align}
\end{subequations}
The phases at the two endpoints are inverses, so associativity is unbroken for the complete operators.
To verify the cocycle condition for $\nu$, we compare the left-end factors of four patches.
In the first ordering, we successively apply the three instances of Eq.~\eqref{eq:EN-left}
corresponding to $(g,h,k)$, $(g,hk,l)$, and $(h,k,l)$, obtaining
\[
 \begin{aligned}
 &\Omega_L(g,h)\Omega_L(gh,k)\Omega_L(ghk,l)\\
 &=\nu(g,h,k)\nu(g,hk,l)\,
   \conj{P(g)}{\bigl(\Omega_L(h,k)\Omega_L(hk,l)\bigr)}
   \Omega_L(g,hkl)\\
 &=\nu(g,h,k)\nu(g,hk,l)\nu(h,k,l)\,
   \conj{P(g)P(h)}{\Omega_L(k,l)}\,
   \conj{P(g)}{\Omega_L(h,kl)}\Omega_L(g,hkl).
 \end{aligned}
\]
In the second ordering, we first apply the associativity relation for $(gh,k,l)$,
then use Eq.~\eqref{eq:patch-fusion} to rewrite the conjugation as conjugation by $P(g)P(h)$,
and finally apply the relation for $(g,h,kl)$, obtaining
\[
 \begin{aligned}
 &\Omega_L(g,h)\Omega_L(gh,k)\Omega_L(ghk,l)\\
 &=\nu(gh,k,l)\,
   \conj{P(g)P(h)}{\Omega_L(k,l)}\,
   \Omega_L(g,h)\Omega_L(gh,kl)\\
 &=\nu(gh,k,l)\nu(g,h,kl)\,
   \conj{P(g)P(h)}{\Omega_L(k,l)}\,
   \conj{P(g)}{\Omega_L(h,kl)}\Omega_L(g,hkl).
 \end{aligned}
\]
The last lines of the two calculations contain the same invertible operator factor.
Canceling this factor gives
\begin{equation}
 \nu(h,k,l)\nu(g,hk,l)\nu(g,h,k)
 =\nu(gh,k,l)\nu(g,h,kl).
 \label{eq:nu-cocycle}
\end{equation}
Equation~\eqref{eq:nu-cocycle} is precisely the three-cocycle condition, so $\nu\in Z^3(BG,\U(1))$.

There is some freedom in constructing $P,\Omega_L,\Omega_R$ from the global symmetry $\Sg$.
The effects of these choices on the associativity phase~\eqref{eq:EN-cocycle} fall into two classes.

The first transformation changes the local truncation of the patch.
Changes at the two endpoints can be made successively, so we describe only the left endpoint.
We modify $P(g)$ by a unitary operator $B(g)$ at $L$, defining
\begin{equation}
 \widehat P(g)=B(g)P(g),\qquad B(e)=1.
 \label{eq:patch-ambiguity}
\end{equation}
Equation~\eqref{eq:patch-fusion} immediately gives the product of the new truncations:
\[
\widehat P(g)\widehat P(h)
=B(g)\,\conj{P(g)}{B(h)}\,
\Omega_L(g,h)\Omega_R(g,h)B(gh)^{-1}\widehat P(gh).
\]
We keep the right-end factor fixed and choose
\begin{equation}
 \widehat\Omega_L(g,h)
 =B(g)\,\conj{P(g)}{B(h)}\,
 \Omega_L(g,h)B(gh)^{-1},\qquad
 \widehat\Omega_R(g,h)=\Omega_R(g,h).
 \label{eq:transported-endpoint-factor}
\end{equation}

All factors in the right-end associativity relation~\eqref{eq:EN-right} remain unchanged, so $\nu(g,h,k)$ is unchanged as well.\footnote{This can also be checked by substituting
Eq.~\eqref{eq:transported-endpoint-factor} directly into the left-end relation~\eqref{eq:EN-left}.
Equation~\eqref{eq:patch-fusion} gives the required identity
\begin{equation}
 \Omega_L(g,h)\,\conj{P(gh)}{B(k)}
 =\conj{P(g)}{\bigl(\conj{P(h)}{B(k)}\bigr)}\,\Omega_L(g,h).
 \label{eq:left-action-composition}
\end{equation}
This identity relates two successive conjugations to conjugation by the combined patch.
Applying it to the expansion of $\widehat\Omega_L(g,h)\widehat\Omega_L(gh,k)$,
and then using Eq.~\eqref{eq:EN-left}, gives the same associativity phase $\nu(g,h,k)$.}

The second transformation keeps all $P(g)$ fixed, so
Eq.~\eqref{eq:patch-fusion} uniquely determines $\Omega_L\Omega_R$.
A phase can nevertheless be redistributed between the two local operators:
\begin{equation}
 \Omega_L(g,h)\longmapsto\beta(g,h)^{-1}\Omega_L(g,h),\qquad
 \Omega_R(g,h)\longmapsto\beta(g,h)\Omega_R(g,h),
 \label{eq:omega-phase}
\end{equation}
where $\beta(g,h)\in\U(1)$ and $\beta(e,g)=\beta(g,e)=1$.
Substituting Eq.~\eqref{eq:omega-phase} into the left-end relation~\eqref{eq:EN-left} gives
\begin{equation}
 \nu\longmapsto\nu\,\delta\beta,\qquad
 (\delta\beta)(g,h,k)=
 \frac{\beta(h,k)\beta(g,hk)}{\beta(gh,k)\beta(g,h)}.
 \label{eq:truncation-coboundary}
\end{equation}

The quantity invariant under both transformations is the cohomology class
\begin{equation}
 [\nu]\in H^3(BG,\U(1)).
 \label{eq:EN-index}
\end{equation}
This class is the Else--Nayak index, which characterizes the anomaly of the global symmetry~\cite{ElseNayak2014}.

\subsection{Commutators on the same patch}
\label{subsec:patch-C}

From this subsection onward, we assume that $G$ is Abelian, so the global group law~\eqref{eq:global-group-law} ensures that the global symmetry transformations commute.
After truncating two global transformations to the same interval,
the two multiplication orders differ only at the endpoints.
We denote the corresponding endpoint factors by $C_L,C_R$:
\begin{equation}
 P(g)P(h)=C_L(g,h)C_R(g,h)P(h)P(g).
 \label{eq:C-exchange}
\end{equation}
This derivation parallels the original Else--Nayak derivation, except that we consider commutativity rather than associativity. Comparing both sides of Eq.~\eqref{eq:C-exchange} with $P(gh)$,
the multiplication formula~\eqref{eq:patch-fusion} directly allows the choice
\begin{equation}
 C_x(g,h)=\Omega_x(g,h)\Omega_x(h,g)^{-1},\qquad x=L,R.
 \label{eq:C-definition}
\end{equation}
We use Eq.~\eqref{eq:C-definition} as our specific choice of $C_L,C_R$ below.

We now compare two ways of rearranging $P(g)P(h)P(k)$.
Both move $P(g)$ to the far right and combine the other two patches into $P(hk)$.
The first moves $g$ past $h,k$ separately and then combines $h,k$;
the second first combines $h,k$ and then exchanges $g$ with $hk$.
We compute this relation for the shared right-end truncation in Fig.~\ref{fig:patch-comparison}.
At $R$, the first ordering gives the first line below,
and the second gives the operator factor in the last line.
Expanding $C_R$ using Eq.~\eqref{eq:C-definition}
and successively applying the three instances of the right-end relation~\eqref{eq:EN-right}
for $(h,k,g)$, $(h,g,k)$, and $(g,h,k)$, we obtain
\begin{equation}
 \begin{aligned}
 &C_R(g,h)\,\conj{P(h)}{C_R(g,k)}\,\Omega_R(h,k)\\
 &\quad=\Omega_R(g,h)\Omega_R(h,g)^{-1}\,
   \conj{P(h)}{\bigl(\Omega_R(g,k)\Omega_R(k,g)^{-1}\bigr)}
   \Omega_R(h,k)\\
 &\quad=\frac{1}{\nu(h,k,g)}
   \Omega_R(g,h)\Omega_R(h,g)^{-1}\,
   \conj{P(h)}{\Omega_R(g,k)}\,
   \Omega_R(h,gk)\Omega_R(hk,g)^{-1}\\
 &\quad=\frac{\nu(h,g,k)}{\nu(h,k,g)}
   \Omega_R(g,h)\Omega_R(gh,k)\Omega_R(hk,g)^{-1}\\
 &\quad=\frac{\nu(h,g,k)}
      {\nu(g,h,k)\nu(h,k,g)}
   \conj{P(g)}{\Omega_R(h,k)}\,C_R(g,hk).
 \end{aligned}
 \label{eq:three-associators}
\end{equation}
The ratio of the three $\nu$ factors explicitly gives the phase by which the two rearrangements differ at a single endpoint.

\subsection{Patch--global-symmetry commutators and type-III anomaly}
\label{subsec:patch-D}

In Section~\ref{subsec:patch-C},
Eq.~\eqref{eq:C-exchange} defines endpoint factors $C_x$ by comparing $P(g)P(h)$ with $P(h)P(g)$.
We now compare $P(g)\Sg(h)$ with $\Sg(h)P(g)$ and similarly define endpoint factors $D_x$:
\begin{equation}
 P(g)\Sg(h)=D_L(g,h)D_R(g,h)\Sg(h)P(g).
 \label{eq:end}
\end{equation}

We compare $P(g)\Sg(h)\Sg(k)$ with $\Sg(h)\Sg(k)P(g)$ through two rearrangements.
The first successively moves $P(g)$ to the far right;
the second first combines $\Sg(h)\Sg(k)$ using the group law~\eqref{eq:global-group-law}, and then performs the exchange.
Both give the same complete operator, so their factors at $L$ can differ at most by a pure phase:
\begin{equation}
 D_L(g,h)\,\conj{\Sg(h)}{D_L(g,k)}\ \propto\ D_L(g,hk).
 \label{eq:scalar}
\end{equation}
We wish to relate the phase in Eq.~\eqref{eq:scalar} to $\nu$. To do so, we introduce a geometric construction relating $D_L$ to $C_R$ and use the known identity~\eqref{eq:three-associators} for the latter.

Choose a truncation $Q$ slightly longer than $P$.
$Q$ shares the right endpoint $R$ with $P$, extends leftward beyond $L$,
and ends at another point far from both $L,R$.
Figure~\ref{fig:patch-comparison} shows the finite intervals used in the calculation.
\begin{figure}[H]
\centering
\setlength{\unitlength}{1mm}
\begin{picture}(150,48)
 \linethickness{.4pt}
 \put(24,41){\line(1,0){116}}
 \put(12,41){\makebox(0,0)[r]{$\Sg$}}
 \put(12,27){\makebox(0,0)[r]{$P$}}
 \put(12,13){\makebox(0,0)[r]{$Q$}}
 {\color{gray}
 \multiput(35,8)(0,2){17}{\line(0,1){.6}}
 \multiput(75,8)(0,2){17}{\line(0,1){.6}}
 \multiput(125,8)(0,2){17}{\line(0,1){.6}}
 }
 \linethickness{1pt}
 \put(75,27){\line(1,0){50}}
 \put(35,13){\line(1,0){90}}
 \put(75,27){\circle*{1}}
 \put(125,27){\circle*{1}}
 \put(35,13){\circle*{1}}
 \put(125,13){\circle*{1}}
 \put(35,3){\makebox(0,0){Other endpoint of $Q$}}
 \put(75,3){\makebox(0,0){$L$:$D_L(g,h)$}}
 \put(125,3){\makebox(0,0){$R$:$C_R(g,h)$}}
\end{picture}
\caption{The two truncations share their right endpoint $R$, while the left endpoint $L$ of $P$ lies in the interior of $Q$.}
\label{fig:patch-comparison}
\end{figure}
We construct $Q$ from the local circuit of the global symmetry
and require $Q$ to use the same endpoint truncation as $P$ near $R$.
In Fig.~\ref{fig:patch-comparison}, $L$ lies inside $Q$, whereas $R$ lies on its boundary.
Thus, when $P(g)$ and $Q(h)$ are exchanged,
the endpoint $R$ gives the same-patch exchange factor $C_R(g,h)$,
and $L$ gives the patch--global-symmetry exchange factor $D_L(g,h)$.
After fixing $\Omega_R$, we fix the phase of $D_L$ by
\begin{equation}
 P(g)Q(h)=C_R(g,h)D_L(g,h)Q(h)P(g).
 \label{eq:mixed}
\end{equation}

The common endpoint truncation also allows us to choose $\Omega_R(h,k)$ as the multiplication factor of $Q(h)Q(k)$ at $R$;
any possible phase difference is absorbed into the factor at the other endpoint.
We now compare two rearrangements of $P(g)Q(h)Q(k)$.
Exchanging successively and then combining gives, at $L$,
$D_L(g,h)\conj{\Sg(h)}{D_L(g,k)}$,
whereas combining first and then exchanging gives $D_L(g,hk)$.
The corresponding factors at $R$ are expressed in terms of $C_R,\Omega_R$.
The multiplication factor at the other endpoint of $Q$ is identical in the two rearrangements,
and its support is disjoint from $P$ and from the operators near $L,R$.
Canceling the common operator on the right therefore gives
\begin{equation}
 \begin{aligned}
 &\Bigl[C_R(g,h)\,\conj{P(h)}{C_R(g,k)}\,\Omega_R(h,k)\Bigr]
 D_L(g,h)\,\conj{\Sg(h)}{D_L(g,k)}\\
 &\qquad=
 \Bigl[\conj{P(g)}{\Omega_R(h,k)}\,C_R(g,hk)\Bigr]D_L(g,hk).
 \end{aligned}
 \label{eq:two-orders}
\end{equation}
The operators in square brackets are located at $R$, and $D_L$ is located at $L$.
The geometry of Fig.~\ref{fig:patch-comparison} makes conjugation by $Q(h)$ equal to conjugation by $P(h)$ at $R$ and by $\Sg(h)$ at $L$.
Equation~\eqref{eq:three-associators} relates the two bracketed factors in Eq.~\eqref{eq:two-orders},
so together the two equations give
\begin{equation}
 D_L(g,h)\,\conj{\Sg(h)}{D_L(g,k)}
 =\frac{\nu(g,h,k)\nu(h,k,g)}{\nu(h,g,k)}D_L(g,hk).
 \label{eq:product}
\end{equation}

We interchange $h,k$ in Eq.~\eqref{eq:product}.
The right-hand sides contain the same $D_L(g,hk)$, so the two orders differ by the six-term ratio
\begin{equation}
 \Alt\nu(g,h,k)=
 \frac{\nu(g,h,k)\nu(h,k,g)\nu(k,g,h)}
 {\nu(g,k,h)\nu(k,h,g)\nu(h,g,k)}.
 \label{eq:alt}
\end{equation}
Coboundary changes of $\nu$ cancel in this ratio:
\begin{equation}
 \Alt(\nu\,\delta\beta)=\Alt\nu.
 \label{eq:alt-coboundary}
\end{equation}
Thus, $\Alt\nu(g,h,k)$ depends only on the anomaly class~\eqref{eq:EN-index} of $\Sg$. We obtain the following theorem.

\begin{theorem}
 \label{thm:endpoint-exchange}
 In the setting of Section~\ref{subsec:symmetry-patch-endpoint}, let $G$ be a finite Abelian group,
 and let the anomaly of $\Sg$ be $[\nu]\in H^3(BG,\U(1))$.
 For any symmetry patch operators $\{P(g)\}$
 and local endpoint operators $\{D_L(g,h),D_R(g,h)\}$ satisfying Eq.~\eqref{eq:end},
 we have
 \begin{equation}
  D_L(g,h)\,\conj{\Sg(h)}{D_L(g,k)}
  =\Alt\nu(g,h,k)\,
   D_L(g,k)\,\conj{\Sg(k)}{D_L(g,h)}.
  \label{eq:exchange}
 \end{equation}
 where $\Alt\nu$ is defined by Eq.~\eqref{eq:alt}.
\end{theorem}

We say that $\Sg$ has a nontrivial \textbf{type-III anomaly} if there exist $g,h,k$ such that $\Alt\nu(g,h,k)\ne1$.

Theorem~\ref{thm:endpoint-exchange} immediately constrains symmetric truncations.
For a global symmetry $\Sg$, a family of symmetry patch operators $\{P(g)\}$
is called a symmetric truncation if
$P(g)\Sg(h)=\Sg(h)P(g)$ for all $g,h$.
In this case, Eq.~\eqref{eq:end} allows $D_L(g,h)=D_R(g,h)=1$.
Equation~\eqref{eq:exchange} in Theorem~\ref{thm:endpoint-exchange}
then requires $\Alt\nu(g,h,k)=1$, giving
\begin{corollary}
 \label{cor:no-symmetric-truncation}
 A one-dimensional symmetry with a nontrivial type-III anomaly admits no symmetric truncation.
\end{corollary}
By contrast, an anomalous $\mathbb Z_2$ symmetry may admit a symmetric truncation.
Corollary~\ref{cor:no-symmetric-truncation} shows that this property does not extend to all anomalous symmetries.

For an Abelian group, $\Alt\nu$ satisfies
\[
 \begin{aligned}
 \Alt\nu(g,h,k)&=\Alt\nu(h,g,k)^{-1}=\Alt\nu(g,k,h)^{-1},\\
 \Alt\nu(g,g,h)&=1,\\
 \Alt\nu(g_1g_2,h,k)&=\Alt\nu(g_1,h,k)\Alt\nu(g_2,h,k).
 \end{aligned}
\]
We call functions $G^3\to\U(1)$ satisfying these three relations \textbf{alternating functions}.
They form a group under pointwise multiplication.
Together with Eq.~\eqref{eq:alt-coboundary}, this makes $\Alt$ a homomorphism from $H^3(BG,\U(1))$ to this group:
\begin{equation}
 \Alt:[\nu]\longmapsto\Alt\nu.
 \label{eq:alt-typeiii-map}
\end{equation}
We call $\Alt\nu$ the \textbf{type-III part} of $[\nu]$.

For example, take $G=\ZZ_2^3$ and use additive notation through $\U(1)\simeq\RZ$.
Anomalies are classified by $H^3(B\ZZ_2^3,\RZ)\simeq\ZZ_2^7$.
Let $a_1,a_2,a_3$ be the coordinate $1$-cocycles of the three cyclic factors. The seven generators may be chosen as
\[
 \underbrace{\frac12a_1^3,\ \frac12a_2^3,\ \frac12a_3^3}_{\text{type-I}},\qquad
 \underbrace{\frac12a_1a_2^2,\ \frac12a_1a_3^2,\ \frac12a_2a_3^2}_{\text{type-II}},\qquad
 \underbrace{\frac12a_1a_2a_3}_{\text{type-III}}.
\]
The alternating ratios of the first six terms are identically $1$, while that of the last term is $-1$ on the three coordinate generators.

More generally, write $G=\prod_{i=1}^r\ZZ_{N_i}$ as a product of cyclic groups. A type-III cocycle may be chosen as
\begin{equation}
 \nu_{\mathrm{III}}(g,h,k)
 =\prod_{i<j<\ell}
 \exp\!\left(
 \frac{2\pi\ii\,p_{ij\ell}}{\gcd(N_i,N_j,N_\ell)}
 g_i h_j k_\ell
 \right).
 \label{eq:abelian-cocycle-representative}
\end{equation}
Here $g_i,h_i,k_i\in\ZZ_{N_i}$ denote the components,
and each $p_{ij\ell}$ independently takes values in $\ZZ_{\gcd(N_i,N_j,N_\ell)}$.
Writing $e_i$ for the standard generator of each cyclic factor, for $i<j<\ell$ we have
\begin{equation}
 \Alt\nu_{\mathrm{III}}(e_i,e_j,e_\ell)
 =\exp\!\left(\frac{2\pi\ii\,p_{ij\ell}}
 {\gcd(N_i,N_j,N_\ell)}\right).
 \label{eq:typeiii-coefficient}
\end{equation}
These values exhaust all possible values of alternating functions on the generators and uniquely determine such a function; hence $\Alt$ is surjective.
The alternating ratios of type-I and type-II classes are both $1$, and these classes generate exactly $\ker\Alt$.

Theorem~\ref{thm:endpoint-exchange} resembles Eq.~\eqref{eq:EN-left}, but has a stronger invariance property. Since $D_L(g,h)$ and $D_L(g,k)$ appear on both sides, redistributing phases between $D_L,D_R$ contributes nothing. Accordingly, the total phase coefficient $\Alt\nu(g,h,k)$ is a constant unaffected by $\nu\mapsto \nu\delta\beta$.

\begin{remark}
Bols et al.\ have proved a result similar to Theorem~\ref{thm:endpoint-exchange} in Ref.~\cite{BolsEtAl2026}. However, they relate the phase factor in Eq.~\eqref{eq:scalar} to $\nu$ indirectly: they first invoke the classification of anomalies by $H^3(BG,\U(1))$, then compute the phase factor in a representative model of each anomaly class. The construction in Fig.~\ref{fig:patch-comparison} gives a short, direct proof.
\end{remark}

\section{Topological excitations and statistics}
\label{sec:statistics}

A discussion of statistics requires a precise notion of topological excitation. Broadly, we model a spatial lattice system by a cell complex $Y$ and consider closed $p$-dimensional topological excitations whose fusion group is a finite Abelian group $G$. The data have two parts.

The first consists of excitation states. Their geometric configurations are described by
\[
 a\in B_p(Y,G)
\]
where $B_p(Y,G)$ is the group of boundary chains; each $ a\in B_p(Y,G)$ labels a normalized configuration state $\confstate{a}$. For example, consider the $e$ anyons of the toric code on a closed surface. Each vertex can contain $0$ or $1$ anyon. We restrict to the vacuum sector, so the total number of anyons is even. Each allowed distribution corresponds to $a\in B_0(Y,G)$, and $\confstate{a}$ is a ground state if $a=0$ and an excited state otherwise.

The second consists of hopping operators. Choose a generating set $G_0$ for $G$. To every $s=g\sigma$, where $g\in G_0$ and $\sigma$ is an oriented $(p+1)$-cell, associate a unitary hopping operator $U(s)$. Its geometric support $\supp(s)$ is the closure of the cell, and $\partial s$ records the excitations it creates at the boundary. In the toric code, the $Z$ operator on each edge can serve as an $e$-anyon hopping operator. We call these hopping operators elementary because their parameters $g\sigma$ are elementary cells. String and loop hopping will be defined in the next subsection as extended and closed hopping operators. An elementary cell need not be the smallest unit of the microscopic lattice; it may be a patch larger than the interaction length, so that the locality axiom below holds.

Following Ref.~\cite{Xue2026Statistics}, we call the data $\{\confstate{a},U(s)\}$ topological excitations if they satisfy two axioms:
\begin{enumerate}[label=\textbf{(\arabic*)}]
 \item \textbf{Configuration axiom.}
 For every $a\in B_p(Y,G)$ and elementary parameter $s$, there is a $\theta(s,a)\in\RZ$ such that
 \begin{equation}
  U(s)\confstate{a}
  =\e^{2\pi\ii\theta(s,a)}\confstate{a+\partial s}.
  \label{eq:configuration-axiom}
 \end{equation}
 \item \textbf{Locality axiom.}
 If $k\ge2$ elementary parameters satisfy $\bigcap_{i=1}^k\supp(s_i)=\varnothing$, then
 \begin{equation}
  [U(s_k),[\cdots,[U(s_2),U(s_1)]]]=1,
  \qquad [A,B]:=A^{-1}B^{-1}AB.
  \label{eq:locality-axiom}
 \end{equation}
\end{enumerate}
We require these operator identities only on the subspace spanned by the configuration states. The first axiom relates the geometric boundary to the action on states; the second expresses locality as commutation relations among hopping operators. First, hopping operators with disjoint supports must commute. Second, the cell complex $Y$ need not be the underlying microscopic lattice of the quantum many-body system, but may be a coarse-grained mesoscopic lattice. Each $U(s)$ is thus viewed as an FDQC on a submanifold of a certain scale, and $[U(s_2),U(s_1)]$ must be supported near the intersection of the supports of $s_1,s_2$. This yields further commutation relations.

Let $X$ be a $d$-dimensional combinatorial sphere, namely a triangulation of $S^d$. In the preceding definition of topological excitations, identify $Y$ with the dual cell complex $Y=X^\vee$ of $X$. This identifies chains on $Y$ with cochains on $X$, facilitating comparison with the field-theoretic constructions below. If the excitations have codimension $1\le q=d-p\le d$, the correspondence is
\begin{equation}
 B^q(X,G)\simeq B_{d-q}(X^\vee,G),
 \qquad
 C^{q-1}(X,G)\simeq C_{d-q+1}(X^\vee,G).
 \label{eq:chain-cochain-dictionary}
\end{equation}
We will accordingly regard $a,s$ as either cochains on $X$ or chains on $X^\vee$ as convenient. An elementary cochain is labeled by $s\in G_0\times X_{q-1}$, where $(g,\sigma)$ denotes the cochain that takes value $g$ on the $(q-1)$-simplex $\sigma$ and $0$ on all other simplices.

The boundary of $s$ on $X^\vee$ becomes its coboundary $\dd s$ on $X$. Moreover, dual cells have a common intersection precisely when their original simplices lie in a common $d$-simplex. We can therefore rewrite the two axioms in cochain language:
\begin{enumerate}[label=\textbf{(\arabic*)}]
 \item \textbf{Configuration axiom.}
 For every $a\in B^q(X,G)$ and $s\in G_0\times X_{q-1}$, there is a $\theta(s,a)\in\RZ$ such that
 \begin{equation}
  U(s)\confstate a
  =\e^{2\pi\ii\theta(s,a)}\confstate{a+\dd s}.
  \label{eq:configuration-axiom-general}
 \end{equation}
 \item \textbf{Locality axiom.}
 If the $(q-1)$-simplices corresponding to $s_1,\ldots,s_k$ do not lie in any common $d$-simplex, then
 \begin{equation}
  [U(s_k),[\cdots,[U(s_2),U(s_1)]]]=1.
  \label{eq:locality-axiom-general}
 \end{equation}
\end{enumerate}

We illustrate this correspondence using the toric-code $m$ anyons from the introduction, now denoting the original cell complex by $X$. Place one qubit on each original edge. The resulting many-body Hilbert space is spanned by the $Z$-basis states $\ket{v}$, with $v\in C^1(X,\ZZ_2)$ interpreted as a matter field. Here $\ket v$ is a matter basis state; a configuration state $\confstate a$ is generally a superposition of such states. An $X$-string follows the dual lattice, and the original edges it crosses define $s\in C^1(X,\ZZ_2)$. Its coboundary $\dd s$ records $m$ anyons on the original faces containing its endpoints. Thus, closed strings correspond to $z\in Z^1(X,\ZZ_2)$ and general strings to $s\in C^1(X,\ZZ_2)$; their symmetry transformations and hopping operators are
\begin{equation*}
  \mathcal S(z)\ket v=\ket{v+z},\qquad
  U(s)\ket v=\ket{v+s}.
\end{equation*}

For $\ZZ_2$ point particles on a one-dimensional circle, take $X$ to be an $n$-gon with $n\ge3$. Configurations are described by $a\in B^1(X,\ZZ_2)$, with point excitations at the midpoints of the original edges. Denote the elementary $0$-cochain at vertex $i$ by $s_i$. Hopping sends $a\mapsto a+\dd s_i$, changing the excitations on the two adjacent edges. Hopping operators at nonadjacent vertices must commute. This agrees exactly with the dual-lattice chain picture in which particles sit at vertices and move along edges.

The two axioms impose specific linear equations on the phases $\{\theta(s,a)\}$. Denote their solution space by $\mathcal M^q(X,G)$, treating each $\theta(s,a)\in\RZ$ as a continuous variable. This solution space is a finite-dimensional manifold, and we define statistics as its connected components. Their classification can be shown to be
\begin{equation}\label{eq:definition statistics}
	\pi_0\mathcal{M}^q(X,G)\simeq H^{d+2}(B^{q}G,\RZ).
\end{equation}
Strictly speaking, $\mathcal{M}^q(X,G)$ depends on the generating set $G_0$, but different choices give the same connected components, so we suppress this dependence in the notation.

Definitions of topological excitations in the literature can differ substantially. Entanglement bootstrap provides a particularly distinct framework: it also seeks to define topological excitations from lattice models, but its only input is a ground state in a many-body Hilbert space with a tensor-product structure. Quantum-information methods then determine the anyons and their fusion rules, ultimately yielding a unitary modular tensor category~\cite{ShiKatoKim2020,Shi2020Verlinde,KimRanard2024}. By contrast, our framework takes both excitation states and hopping operators as input. Locality does not depend on a tensor-product structure but is expressed purely through commutation relations among hopping operators, and the braided fusion category describing anyon statistics need not be modular. Our notion of topological excitation is therefore naturally suited to subspaces selected by low-energy or symmetry constraints in lattice models, or to boundaries of topological orders in one higher dimension, rather than to topological orders in the same dimension. The degree $d+2$ of $H^{d+2}(B^{q}G,\RZ)$ is consistent with this viewpoint.

\subsection{Choosing and assembling hopping operators}
\label{subsec:hopping-assembly}

Two important constructions concern hopping operators: choosing elementary hopping operators, and assembling them into extended hopping operators such as $X$-strings. Ref.~\cite{Xue2026Statistics} makes the former a central motivation for the axioms of topological excitations but does not clearly define the latter. We need a precise definition of extended hopping operators to compare them with symmetry patch operators.

Our axiomatic framework takes configuration states and hopping operators as given data. This naturally raises the question of how these states and operators relate to the Hamiltonian of a lattice model. For exactly solvable Hamiltonians such as the toric code on a sphere, each eigenspace specified by $A_v=\pm 1,B_p=\pm 1$ is one-dimensional, so these eigenstates naturally serve as configuration states. The definition of hopping operators, however, is less evident. For example, one often says
\begin{itemize}
 \item An $X$-string is an $m$-anyon hopping operator.
\end{itemize}
This does not answer what a hopping operator actually is. The central question is whether it is unique. There are two possible answers:
\begin{enumerate}
 \item The definition uniquely determines the hopping operator; in the toric code, an $m$-anyon hopping operator must be an $X$-string.
 \item The definition allows multiple choices of hopping operators; an $X$-string is just one example.
\end{enumerate}
We adopt the latter viewpoint. On the one hand, a Hamiltonian does not seem to determine hopping operators uniquely in general. On the other, Refs.~\cite{LevinWen2003,KawagoeLevin2020,Fidkowski_2022} carefully examine the stability of statistical invariants under local perturbations of hopping operators. Accordingly, Ref.~\cite{Xue2026Statistics} emphasizes that the configuration and locality axioms are the only conditions imposed on hopping operators: each $U(s)$ may be reselected within the freedom allowed by these axioms.

We now explain precisely what it means to reselect a hopping operator within the axioms. Start with data $\{\confstate{a},U(s)\}$ satisfying the axioms of topological excitations. Choose a particular elementary cochain $t$ and replace $U(t)$ by $U'(t)$, keeping all configuration states $\confstate{a}$ and all other hopping operators $U(s)$, with $s\ne t$, fixed. Since both $U'(t), U(t)$ act on configurations as $a\mapsto a+\dd t$, their difference on the span of configuration states is described by a diagonal operator $\Delta_t$ preserving each state line:

\begin{equation}
  U'(t)=U(t)\Delta_t,
  \qquad
  \Delta_t\confstate a=\e^{2\pi\ii\delta_t(a)}\confstate a.
  \label{eq:hopping-redefinition}
\end{equation}

The locality condition on $U'(t)$ can also be characterized by a locality condition on $\Delta_t$. If the $(q-1)$-simplices corresponding to $t,s_1,\ldots,s_k\in G_0\times X_{q-1}$ do not lie in any common $d$-simplex, then
\begin{equation}
	[U(s_k),[\cdots,[U(s_1),\Delta_t]]]=1.
	\label{eq:locality-axiom-redefinition}
\end{equation}
We call a replacement satisfying Eqs.~\eqref{eq:hopping-redefinition} and~\eqref{eq:locality-axiom-redefinition} an admissible redefinition of elementary hopping. The \textbf{Operator Independence Theorem} (Ref.~\cite{Xue2026Statistics}, Theorem VI.6) states that such a redefinition preserves the connected component in $\mathcal M^q(X,G)$, hence preserves the statistics. Its proof is not immediate; it is what allows hopping operators to be chosen freely within the axioms while retaining well-defined statistics.

We next discuss assembling hopping operators. The axioms take only hopping operators corresponding to elementary cochains labeled by $G_0\times X_{q-1}$ as basic data; they do not include extended hopping operators for general cochains in $C^{q-1}(X,G)$. This reflects the formulation of locality: support is mathematically natural for a simplex or elementary cochain, but not for a general cochain. Ultimately, an extended hopping operator is defined by its assembly from elementary hopping operators, so only the latter can be basic data. On the other hand, the field-theoretic approach of Section~\ref{sec:what-is-hopping} systematically constructs an extended hopping operator $U(s)$ for every $s\in C^{q-1}(X,G)$. We must explain in what sense these extended operators are consistent with the elementary hopping operators in the axioms.

Consider any $s\in C^{q-1}(X,G)$. Select the $(q-1)$-simplices on which $s$ is nonzero, and express each coefficient as a sum of generators in $G_0$, obtaining a decomposition $s=s_1+\cdots+s_n$ into elementary cochains. Simplices with zero coefficients do not participate. Choosing a multiplication order defines the corresponding extended hopping operator
\begin{equation}
  U(s):=U(s_n)\cdots U(s_1).
  \label{eq:extended-hopping-product}
\end{equation}

This product automatically satisfies the configuration axiom: each step adds $\dd s_i$ to the configuration, and the total action is $a\mapsto a+\dd s$.

The operators $U(s_i)$ need not commute, so different multiplication orders may yield different $U(s)$. These differences can be absorbed into admissible redefinitions of elementary hopping operators. For example, consider the operators $U(s_1),U(s_2)$ for two adjacent elementary cochains $s_1,s_2$. We may define the extended operator for their sum by
\begin{equation}
	U(s_1+s_2)=U(s_2)U(s_1).
\end{equation}
Alternatively, we may choose $U(s_1)U(s_2)$. A direct calculation gives
\begin{equation}
	U(s_1)U(s_2)=U(s_2)U(s_1)[U(s_1),U(s_2)]=U(s_2)U'(s_1),
\end{equation}

where
\begin{equation}
	U'(s_1)=U(s_1)[U(s_1),U(s_2)]
\end{equation}
satisfies Eqs.~\eqref{eq:hopping-redefinition} and~\eqref{eq:locality-axiom-redefinition} and is therefore an admissible redefinition of $U(s_1)$. For a general cochain $s$, differences arising from the decomposition within each simplex and from the multiplication order can likewise be absorbed into admissible redefinitions of the corresponding elementary hopping operators.

Thus, given elementary hopping operators for a topological excitation, we call an operator for $s\in C^{q-1}(X,G)$ an admissible extended hopping operator compatible with them if it can be written as
\begin{equation}
  U(s)=U'(s_n)\cdots U'(s_1),
  \qquad s=\sum_i s_i,
  \label{eq:extended-hopping-consistency}
\end{equation}
where the $s_i$ arise from the above decomposition performed separately on each nonzero simplex, and each $U'(s_i)$ is an admissible redefinition of the corresponding elementary hopping operator. If $s$ is a cocycle, we call such an operator a closed hopping operator. Since $\dd s=0$, $U(s)$ preserves every configuration state line, although its phases may differ between state lines.

The extended hopping operators systematically constructed by the field-theoretic approach of Section~\ref{sec:what-is-hopping} are indeed all admissible; see Eq.~\eqref{eq:boundary-hopping-composition}.

\subsection{Statistics and statistical processes}
\label{subsec:statistical-processes}

We now ask how to detect the statistics of a topological excitation. Abstractly, we seek a way to identify the connected component of $\mathcal M^q(X,G)$ containing any given data $\{\confstate{a},U(s)\}$.

Fix an initial configuration $a_0$ and a finite sequence of elementary hopping parameters $s_1,\ldots,s_n$ with signs $\varepsilon_i=\pm1$. Require the final configuration to equal the initial one, so $\sum_i\varepsilon_i\dd s_i=0$. For any data $\{\confstate a,U(s)\}$ satisfying the axioms of topological excitations, the configuration axiom gives the Berry phase
\begin{equation}
 {}^{\mathrm{conf}}\!\bra{a_0}\,
 U(s_n)^{\varepsilon_n}\cdots U(s_1)^{\varepsilon_1}\confstate{a_0}
 =\exp\!\left(2\pi\ii\sum_{s,a}c(s,a)\theta(s,a)\right).
 \label{eq:statistical-berry-phase}
\end{equation}
The integer $c(s,a)\in\ZZ$ counts the net occurrences of $\theta(s,a)$ in the process. A forward step $U(s)$ starting from $a$ contributes $+\theta(s,a)$; an inverse step $U(s)^{-1}$ returning from $a+\dd s$ to $a$ contributes $-\theta(s,a)$. Thus forward steps use their initial configurations and inverse steps use their final configurations. We use this convention throughout. These counts depend only on the initial configuration and parameter sequence, not on the particular hopping operators.

We call the formal sequence
\[
 P=s_n^{\varepsilon_n}\cdots s_1^{\varepsilon_1}
\]
a \textit{process}. For fixed initial configuration $a_0$, a process $P$ determines an evaluation map
\[
 \begin{aligned}
 \operatorname{ev}_P:\mathcal M^q(X,G)&\longrightarrow\RZ,\\
 \{\confstate a,U(s)\}&\longmapsto\sum_{s,a}c(s,a)\theta(s,a).
 \end{aligned}
\]
Here the data $\{\confstate a,U(s)\}$ denote the point of the moduli space specified by their phases $\{\theta(s,a)\}$, and we identify $\RZ$ with $\U(1)$ by $t\mapsto\e^{2\pi\ii t}$. Thus $\operatorname{ev}_P$ records precisely the Berry phase in Eq.~\eqref{eq:statistical-berry-phase}. The states and operators supply the input, whereas the process $P$ determines the map itself. The map $\operatorname{ev}_P$ is a continuous group homomorphism on the moduli space. For finite $X,G$, Ref.~\cite{Xue2026Statistics} proves the equivalence of the following three conditions:
\begin{enumerate}[label=\textbf{(\arabic*)}]
 \item \textbf{Local constancy.} The map $\operatorname{ev}_P$ is constant on every connected component of $\mathcal M^q(X,G)$.
 \item \textbf{Finite order.} There is a common positive integer $m$ such that $m\operatorname{ev}_P(\theta)=0\in\RZ$ for every $\theta$ satisfying the axioms of topological excitations.
 \item \textbf{Local cancellation.} Viewed locally at any point $x$, each hopping $s$ whose support contains $x$ cancels against its inverse $s^{-1}$.
\end{enumerate}

The first condition is conceptually simplest, whereas the third is convenient for checking concrete processes. Examples below illustrate the meaning of local cancellation.

We call a process $P$ satisfying these equivalent conditions a \textit{statistical process}. Two statistical processes $P,P'$ are \textit{equivalent} if they define the same evaluation map on the entire moduli space $\mathcal M^q(X,G)$, namely $\operatorname{ev}_P=\operatorname{ev}_{P'}$. Their equivalence classes are classified by
\begin{equation}
 \Hom(\pi_0\mathcal{M}^q(X,G),\U(1))\simeq H_{d+2}(B^qG,\ZZ)
 \label{eq:statistical-process-classification}
\end{equation}
Equation~\eqref{eq:definition statistics} classifies the statistics of topological excitations, while Eq.~\eqref{eq:statistical-process-classification} classifies the processes detecting them. The idea of statistical processes originates in Refs.~\cite{LevinWen2003,Fidkowski_2022,KawagoeLevin2020}; recent work has also constructed many geometric processes for higher-dimensional excitations~\cite{Feng_2026,feng2026paulistabilizerformalismtopological}. Operator independence from the preceding subsection and initial-state independence below further establish that statistical phases do not depend on the initial state or the choice of hopping operators.

The axioms take configuration states and hopping operators together as input. Interestingly, knowing all the hopping operators already suffices to determine the statistics. The \textbf{Initial-State Independence Theorem} (Ref.~\cite{Xue2026Statistics}, Theorem VI.4) states that $\operatorname{ev}_P$ for a statistical process $P$ does not depend on the initial configuration $a_0$. The corresponding operator product is therefore the same pure phase throughout the Hilbert space spanned by configuration states.\footnote{For the construction in Eq.~\eqref{eq:hopping operator formula}, this conclusion holds on the entire $\mathcal H_{\mathrm{matter}}$; see Remark~\ref{rem:configuration-symmetry} for the proof.} For this reason, we sometimes simply say that a family of hopping operators satisfies the axioms of topological excitations and discuss its statistics, implicitly assuming the existence of compatible configuration states.

The next three figures show anyon exchange, mutual braiding, and one-dimensional $\ZZ_2$ statistics. We depict particles and their paths in the chain picture; operators along long paths are assembled from elementary hopping operators as in Section~\ref{subsec:hopping-assembly}. We first check local cancellation and return to the choice of assembly at the end of this subsection.

For two-dimensional Abelian anyons, we describe the next two processes in the chain picture, choosing the three arms of Fig.~\ref{fig:anyon-T-junction} in a cell decomposition $Y$ of the sphere. The operator $U_{0j}(g)$ moves $g$ along $0\to j$. When an arm contains several cells, or $g$ is not in the chosen generating set, we use an admissible extended hopping operator as in Section~\ref{subsec:hopping-assembly}.
\begin{figure}[H]
 \centering
 \resizebox{\textwidth}{!}{\embeddedTjunction}
 \caption{The anyon T-junction process. Colored dots track the positions of two identical particles; each arrow denotes hopping along one arm. The panels show the six actions in Eq.~\eqref{eq:anyon-T-junction}, read from right to left.}
 \label{fig:anyon-T-junction}
\end{figure}
Place one $g$ particle at each of $1,2$, and compensating excitations of total type $-2g$ far away, so that the total configuration lies in $B_0(Y,G)$. The following six-step process exchanges the identical particles:
\begin{equation}
 U_{02}(g)U_{03}(g)^{-1}U_{01}(g)
 U_{02}(g)^{-1}U_{03}(g)U_{01}(g)^{-1}
 =\e^{2\pi\ii Q(g)}.
 \label{eq:anyon-T-junction}
\end{equation}
Operators act from right to left. The first particle moves from $1$ through $0$ to $3$, the second from $2$ through $0$ to $1$, and finally the first from $3$ through $0$ to $2$. The distant compensating excitations remain fixed. We denote the statistical phase for the depicted exchange orientation by $Q(g)\in\RZ$; the opposite orientation gives its negative. The Initial-State Independence Theorem makes Eq.~\eqref{eq:anyon-T-junction} an operator identity on the span of configuration states~\cite{LevinWen2003,KawagoeLevin2020,Xue2026Statistics}.

The T-junction illustrates local cancellation. Each arm appears exactly once in each hopping direction. Write $ij$ for the configuration with particles at $i,j$. For an observation point at the center $0$, the configurations used in counting are $01,02,03$. All have one particle at $0$, while the other can move along the three outer edges without crossing $0$. They therefore represent the same local configuration class, and the forward and inverse steps on each arm cancel. At the outer endpoint $1$, only $U_{01}$ is retained: its inverse step ends at $02$, and its forward step starts at $03$. The outer path $2\to3$ avoids $1$, so these steps also cancel. The other endpoints and observation points inside the arms give the same check. Although the outer paths do not explicitly appear in Eq.~\eqref{eq:anyon-T-junction}, they still ensure local cancellation.

Anyons of two different types $g,h$ also have mutual braiding statistics. Choose the perpendicular paths $A,B$ in Fig.~\ref{fig:ordinary-braiding-geometry}, with $A$ directed left to right and $B$ bottom to top. Initially $g$ is at the left endpoint of $A$ and $h$ at the bottom endpoint of $B$; distant compensating excitations remain fixed. Four hopping steps move $g$ rightward, then $h$ upward, and then return them in the same order. Their relative position winds once.

\begin{figure}[H]
 \centering
 \begin{tikzpicture}[x=1cm,y=1cm,>=latex,every node/.style={font=\scriptsize}]
  \begin{scope}[shift={(0.00,0)}]
   \draw[red!65!black] (-.62,0)--(.62,0);
   \draw[blue!65!black] (0,-.65)--(0,.65);
   \filldraw[draw=red!75!black,fill=red!75!black] (-.62,0) circle (2.2pt);
   \node[below left,inner sep=2pt] at (-.62,0) {$g$};
   \filldraw[draw=red!75!black,fill=white] (.62,0) circle (2.2pt);
   \filldraw[draw=blue!70!black,fill=blue!70!black] (0,-.65) circle (2.2pt);
   \node[below,inner sep=2pt] at (0,-.65) {$h$};
   \filldraw[draw=blue!70!black,fill=white] (0,.65) circle (2.2pt);
  \end{scope}
  \draw[->,thick] (0.96,0)--node[above] {$U_A$} (2.39,0);
  \begin{scope}[shift={(3.35,0)}]
   \draw[red!65!black] (-.62,0)--(.62,0);
   \draw[blue!65!black] (0,-.65)--(0,.65);
   \filldraw[draw=red!75!black,fill=white] (-.62,0) circle (2.2pt);
   \filldraw[draw=red!75!black,fill=red!75!black] (.62,0) circle (2.2pt);
   \node[below right,inner sep=2pt] at (.62,0) {$g$};
   \filldraw[draw=blue!70!black,fill=blue!70!black] (0,-.65) circle (2.2pt);
   \node[below,inner sep=2pt] at (0,-.65) {$h$};
   \filldraw[draw=blue!70!black,fill=white] (0,.65) circle (2.2pt);
  \end{scope}
  \draw[->,thick] (4.31,0)--node[above] {$U_B$} (5.74,0);
  \begin{scope}[shift={(6.70,0)}]
   \draw[red!65!black] (-.62,0)--(.62,0);
   \draw[blue!65!black] (0,-.65)--(0,.65);
   \filldraw[draw=red!75!black,fill=white] (-.62,0) circle (2.2pt);
   \filldraw[draw=red!75!black,fill=red!75!black] (.62,0) circle (2.2pt);
   \node[below right,inner sep=2pt] at (.62,0) {$g$};
   \filldraw[draw=blue!70!black,fill=white] (0,-.65) circle (2.2pt);
   \filldraw[draw=blue!70!black,fill=blue!70!black] (0,.65) circle (2.2pt);
   \node[above,inner sep=2pt] at (0,.65) {$h$};
  \end{scope}
  \draw[->,thick] (7.66,0)--node[above] {$U_A^{-1}$} (9.09,0);
  \begin{scope}[shift={(10.05,0)}]
   \draw[red!65!black] (-.62,0)--(.62,0);
   \draw[blue!65!black] (0,-.65)--(0,.65);
   \filldraw[draw=red!75!black,fill=red!75!black] (-.62,0) circle (2.2pt);
   \node[below left,inner sep=2pt] at (-.62,0) {$g$};
   \filldraw[draw=red!75!black,fill=white] (.62,0) circle (2.2pt);
   \filldraw[draw=blue!70!black,fill=white] (0,-.65) circle (2.2pt);
   \filldraw[draw=blue!70!black,fill=blue!70!black] (0,.65) circle (2.2pt);
   \node[above,inner sep=2pt] at (0,.65) {$h$};
  \end{scope}
  \draw[->,thick] (11.01,0)--node[above] {$U_B^{-1}$} (12.44,0);
  \begin{scope}[shift={(13.40,0)}]
   \draw[red!65!black] (-.62,0)--(.62,0);
   \draw[blue!65!black] (0,-.65)--(0,.65);
   \filldraw[draw=red!75!black,fill=red!75!black] (-.62,0) circle (2.2pt);
   \node[below left,inner sep=2pt] at (-.62,0) {$g$};
   \filldraw[draw=red!75!black,fill=white] (.62,0) circle (2.2pt);
   \filldraw[draw=blue!70!black,fill=blue!70!black] (0,-.65) circle (2.2pt);
   \node[below,inner sep=2pt] at (0,-.65) {$h$};
   \filldraw[draw=blue!70!black,fill=white] (0,.65) circle (2.2pt);
  \end{scope}
 \end{tikzpicture}
 \caption{The two-particle braiding process. Red and blue dots denote $g,h$, respectively. After four steps, both particles return to their original positions, while their relative position winds once clockwise.}
 \label{fig:ordinary-braiding-geometry}
\end{figure}
The paths cross once, and the depicted process gives the mutual braiding phase
\begin{equation}
 [U_B,U_A]=U_B^{-1}U_A^{-1}U_BU_A
 =\e^{2\pi\ii b_Q(g,h)},\qquad
 b_Q(g,h)=Q(g+h)-Q(g)-Q(h).
 \label{eq:anyon-mutual-braiding}
\end{equation}
The orientation agrees with the exchange orientation in Fig.~\ref{fig:anyon-T-junction}. The reversed commutator satisfies
\begin{equation}
 [U_A,U_B]=\e^{-2\pi\ii b_Q(g,h)}.
 \label{eq:ordinary-braiding-process}
\end{equation}
The function $Q$ is quadratic, and $b_Q$ is additive in each argument. The two processes detect self-statistics and mutual braiding statistics, respectively.

Fix any observation point on $A$. The starting configuration of the forward $U_A$ step and the ending configuration of its inverse differ only in the position of $h$. Since $h$ can move from bottom to top along a path avoiding the observation point, these hopping steps cancel on locally identical configurations. The same applies to the two $U_B$ steps, including at the intersection of the paths.

For one-dimensional $\ZZ_2$ point particles, denote elementary hopping on two adjacent edges by $U_1,U_2$. The corresponding statistical process is
\begin{equation}
 \lambda=[U_2^2,U_1]=\pm1.
 \label{eq:one-dimensional-z2-statistics}
\end{equation}
Figure~\ref{fig:one-dimensional-z2-process} expands this process into six steps. The line segments represent two adjacent edges of a finite circle, with the left and right edges corresponding to $U_1,U_2$.

\begin{figure}[H]
 \centering
 \begin{tikzpicture}[x=1cm,y=1cm,>=latex,every node/.style={font=\small}]
  \foreach \name/\x/\y/\leftfill/\middlefill/\rightfill in {
   q0/-4/0/white/white/white,
   q1/-2/1.8/black/black/white,
   q2/2/1.8/black/white/black,
   q3/4/0/black/black/white,
   q4/2/-1.8/white/white/white,
   q5/-2/-1.8/white/black/black}{
   \node[minimum width=1.65cm,minimum height=.5cm,inner sep=0pt] (\name) at (\x,\y) {};
   \begin{scope}[shift={(\x,\y)}]
    \draw[red!70!black,thick] (-.62,0)--(0,0);
    \draw[blue!70!black,thick] (0,0)--(.62,0);
    \filldraw[fill=\leftfill] (-.62,0) circle (2.2pt);
    \filldraw[fill=\middlefill] (0,0) circle (2.2pt);
    \filldraw[fill=\rightfill] (.62,0) circle (2.2pt);
   \end{scope}
  }
  \draw[->] (q0)--coordinate[midway] (e0) node[midway,above left,text=red!75!black] {$U_1$} (q1);
  \draw[->] (q1)--coordinate[midway] (e1) node[midway,above,text=blue!70!black] {$U_2$} (q2);
  \draw[->] (q2)--coordinate[midway] (e2) node[midway,above right,text=blue!70!black] {$U_2$} (q3);
  \draw[->] (q3)--coordinate[midway] (e3) node[midway,below right,text=red!75!black] {$U_1^{-1}$} (q4);
  \draw[->] (q4)--coordinate[midway] (e4) node[midway,below,text=blue!70!black] {$U_2^{-1}$} (q5);
  \draw[->] (q5)--coordinate[midway] (e5) node[midway,below left,text=blue!70!black] {$U_2^{-1}$} (q0);
  \draw[gray!75,dashed] (e0)--(e3);
  \draw[gray!75,dashed] (e1)--(e4);
  \draw[gray!75,dashed] (e2)--(e5);
 \end{tikzpicture}
 \caption{The one-dimensional $\ZZ_2$ statistical process starts at the leftmost vertex of the hexagon, proceeds clockwise through six steps, and returns to its starting point. Filled dots denote particles. The three dashed lines pair forward and inverse steps that cancel locally at the central point; forward steps use their initial configurations and inverse steps their final configurations.}
 \label{fig:one-dimensional-z2-process}
\end{figure}

Each paired set of steps has the same configuration at the center, and differences far away can be removed along the other side of the circle. The two values in Eq.~\eqref{eq:one-dimensional-z2-statistics} correspond to the two statistics in $H^3(B\ZZ_2,\RZ)\cong\ZZ_2$. We cannot require $U_2^2=1$, since that would exclude $\lambda=-1$.

Finally, return to the choice of extended hopping. Figure~\ref{fig:braiding-assembly-comparison} depicts the same crossing process in two forms: the left uses hopping along complete paths, while the right splits each path into two elementary hops. First take $g,h\in G_0$ and let each of the four arms be an elementary cell.
\begin{figure}[H]
 \centering
 \begin{tikzpicture}[x=1cm,y=1cm,>=latex,every node/.style={font=\small}]
  \begin{scope}[shift={(0,0)}]
   \draw[red!75!black,thick,->] (-1.65,0)--(1.65,0);
   \draw[blue!70!black,thick,->] (0,-1.05)--(0,1.05);
   \node[red!75!black,above] at (.85,0) {$U_A$};
   \node[blue!70!black,left] at (0,.65) {$U_B$};
   \node at (0,-1.5) {Complete paths};
  \end{scope}
  \node at (3.25,0) {$\longleftrightarrow$};
  \begin{scope}[shift={(6.5,0)}]
   \draw[red!75!black,thick,->] (-1.65,0)--(-.04,0);
   \draw[red!75!black,thick,->] (.04,0)--(1.65,0);
   \draw[blue!70!black,thick,->] (0,-1.05)--(0,-.04);
   \draw[blue!70!black,thick,->] (0,.04)--(0,1.05);
   \fill (0,0) circle (1.8pt);
   \node[red!75!black,above] at (-.85,0) {$U_{A,1}$};
   \node[red!75!black,above] at (.85,0) {$U_{A,2}$};
   \node[blue!70!black,left] at (0,-.65) {$U_{B,1}$};
   \node[blue!70!black,left] at (0,.65) {$U_{B,2}$};
   \node at (0,-1.5) {Elementary decomposition};
  \end{scope}
 \end{tikzpicture}
 \caption{Extended and elementary descriptions of the same braiding process. The two admissible multiplication orders give the same statistical phase.}
 \label{fig:braiding-assembly-comparison}
\end{figure}
We may choose either of the following two orders:
\begin{equation}
 U_A=U_{A,2}U_{A,1},\qquad
 U'_A=U_{A,1}U_{A,2},\qquad
 U_B=U_{B,2}U_{B,1}.
 \label{eq:braiding-elementary-assembly}
\end{equation}
Section~\ref{subsec:hopping-assembly} shows that changing the order can be absorbed into admissible redefinitions of elementary hopping. The Operator Independence Theorem therefore gives
\begin{equation}
 [U_B,U_A]=[U_B,U'_A]=\e^{2\pi\ii b_Q(g,h)}
 \label{eq:braiding-order-independence}
\end{equation}
on configuration states. The order on $B$ can be changed similarly. For longer paths or general $g,h$, further decomposition as in Section~\ref{subsec:hopping-assembly} gives the same conclusion. We may therefore use admissible extended hopping directly in the geometric pictures above.

\subsection{One-dimensional type-III statistics}
\label{subsec:typeiii-statistical-process}

We now construct a process detecting one-dimensional type-III statistics. Let $G$ be a finite Abelian group,
let $\{\confstate a,U(s)\}$ describe topological excitations on a circle $X\simeq S^1$,
and write $\mathcal H_{\mathrm{conf}}=\operatorname{span}\{\confstate a\}$.
We use additive notation for the group law. For any $g\in G$ and interval $I$,
let $U_g(I)$ denote an admissible extended hopping operator selected from the same elementary hopping family
as in Section~\ref{subsec:hopping-assembly},
with parameter $g\chi_I$, where $\chi_I$ is the indicator $0$-cochain of the interval.
The operator $U_g(S^1)$ is an admissible closed hopping operator with constant parameter $g$.

Choose four well-separated endpoints $x_0,x_1,x_2,x_3$ in positive order around the circle and set
\[
 P=[x_1,x_3],\qquad Q=[x_0,x_2].
\]
Abbreviate the three operators as
$P(g):=U_g(P)$, $S(h):=U_h(S^1)$, and $Q(k):=U_k(Q)$, and define the process
\begin{equation}
 \mathcal P(g,h,k):=[Q(k),[P(g),S(h)]],
 \qquad [A,B]=A^{-1}B^{-1}AB.
 \label{eq:three-point-theta}
\end{equation}
Figure~\ref{fig:typeiii-four-configurations} fixes the relative positions of the intervals and the signs of the particle types.

\begin{figure}[H]
 \centering
 \begin{tikzpicture}[x=1cm,y=1cm,>=latex,every node/.style={font=\small},
     closedcol/.style={green!40!black},
     pcol/.style={blue!65!black},qcol/.style={orange!85!black}]
  \draw[->,gray!70] (-4.1,5.25)--(4.1,5.25);
  \node[above] at (0,5.25) {Positive orientation};
  \foreach \x/\lab in {-3.3/x_0,-1.1/x_1,1.1/x_2,3.3/x_3}{
   \draw[gray!40,densely dashed] (\x,2.75)--(\x,4.78);
   \node[above] at (\x,4.78) {$\lab$};
  }
  \draw[closedcol,line width=1.8pt] (-4.2,4.45)--(4.2,4.45);
  \draw[closedcol] (-4.2,4.35)--(-4.2,4.55) (4.2,4.35)--(4.2,4.55);
  \node[closedcol,left] at (-4.45,4.45) {$S(h)$};
  \draw[pcol,line width=1.8pt] (-1.1,3.7)--(3.3,3.7);
  \fill[blue!65!black] (-1.1,3.7) circle (2pt);
  \fill[blue!65!black] (3.3,3.7) circle (2pt);
  \node[pcol,left] at (-4.45,3.7) {$P(g)$};
  \node[pcol,below] at (-1.1,3.66) {$+g$};
  \node[pcol,below] at (3.3,3.66) {$-g$};
  \draw[qcol,line width=1.8pt] (-3.3,2.95)--(1.1,2.95);
  \fill[orange!85!black] (-3.3,2.95) circle (2pt);
  \fill[orange!85!black] (1.1,2.95) circle (2pt);
  \node[qcol,left] at (-4.45,2.95) {$Q(k)$};
  \node[qcol,below] at (-3.3,2.91) {$+k$};
  \node[qcol,below] at (1.1,2.91) {$-k$};
  \foreach \xx/\yy/\pf/\qf/\phase in {
      -2.5/1.2/0/0/{S(h)},
       2.5/1.2/1/0/{S(h)^{-1}},
      -2.5/-.7/0/1/{S(h)^{-1}},
       2.5/-.7/1/1/{S(h)}}{
   \begin{scope}[shift={(\xx,\yy)}]
    \draw[black!60,thick] (-1.9,0)--(1.9,0);
    \foreach \pos in {-1.5,-.5,.5,1.5}{
     \draw[gray!60] (\pos,-.05)--(\pos,.05);
    }
    \ifnum\pf=1\relax
     \fill[blue!65!black] (-.5,0) circle (2.4pt);
     \fill[blue!65!black] (1.5,0) circle (2.4pt);
     \node[pcol,above] at (-.5,.04) {$+g$};
     \node[pcol,above] at (1.5,.04) {$-g$};
    \fi
    \ifnum\qf=1\relax
     \fill[orange!85!black] (-1.5,0) circle (2.4pt);
     \fill[orange!85!black] (.5,0) circle (2.4pt);
     \node[qcol,above] at (-1.5,.04) {$+k$};
     \node[qcol,above] at (.5,.04) {$-k$};
    \fi
    \node[closedcol,below] at (0,-.15) {$\phase$};
   \end{scope}
  }
 \end{tikzpicture}
 \caption{Interleaved intervals and four configurations of the type-III statistical process. The circle is cut outside $P\cup Q$, and its left and right ends are identified.
 The lower panels show the four configurations obtained from $a=0$ by applying $P(g)$ and $Q(k)$; green labels denote the corresponding closed hopping operators or their inverses.}
 \label{fig:typeiii-four-configurations}
\end{figure}

After the phases of open hopping cancel, the four remaining closed hopping operations pair with opposite signs on the local configurations shown.
Thus Eq.~\eqref{eq:three-point-theta} is a statistical process.

Let the statistics be $[\omega]\in H^3(BG,\RZ)$, and write
$\widehat\omega=\exp(2\pi\ii\omega)$. Lemma~\ref{lem:typeiii-statistical-process}
in Appendix~\ref{app:typeiii-statistical-evaluation} gives
\[
 \mathcal P(g,h,k)\big|_{\mathcal H_{\mathrm{conf}}}
 =\eta(g,h,k)\one,
 \qquad \eta(g,h,k)=\Alt\widehat\omega(g,h,k).
\]
By the definition in Section~\ref{subsec:patch-D}, $\eta$ is precisely the type-III part of the statistics $[\omega]$.
For example, for $G=\ZZ_2^3$, the statistics class $\omega=\frac12a_1a_2a_3$
gives $\eta=-1$ on the three standard generators; for $G=\ZZ_2^r$, all values of $\eta$ lie in $\{\pm1\}$.

\begin{remark}
 \label{rem:typeiii-equivalent-process}
The same statistical phase is also given by the following process. Suppressing group labels,
the operators $S,[P,S],[P,Q]$ are all diagonal on $\mathcal H_{\mathrm{conf}}$.
Using $[P^{-1},S^{-1}]=P[P,S]P^{-1}$,
$P^{-1}QP=Q[P,Q]^{-1}$, and the commutativity of $[P,Q]$ with $[P,S]$, we obtain
\begin{equation}
 \begin{aligned}
 [[P^{-1},S^{-1}],Q^{-1}]
 &=(PQ)[Q,[P,S]](PQ)^{-1}\\
 &=\eta(g,h,k)\one
 \qquad\text{on }\mathcal H_{\mathrm{conf}}.
 \end{aligned}
 \label{eq:typeiii-equivalent-process}
\end{equation}
\end{remark}

\section{The factor of \texorpdfstring{$-2$}{-2} between type-III anomaly and statistics}
\label{sec:typeiii-configuration}

Section V of Ref.~\cite{Kobayashi_2026} attempts to construct topological excitations directly from symmetry and thereby identify statistics with anomaly. Near its Eq.~(104), the authors write

\begin{quote}
 We assume that the generators of the symmetry $g\in G$ (on a closed surface $\Sigma_{p+1}$ of a $d$-dimensional hypercubic lattice) are expressed as a product of local unitaries:

\begin{equation}
 U_g(\Sigma_{p+1})
 =\prod_{\Delta_{p+1}\in\Sigma_{p+1}}U_{\Delta_{p+1},g}.
 \label{eq:criticized-local-product}
\end{equation}
where $U_{\Delta_{p+1},g}$ is a unitary supported on a $(p+1)$-cube $\Delta_{p+1}$. We do not require the operators $\{U_{g}(\Delta_{p+1})\}$ to commute with each other.
\end{quote}
The paper subsequently denotes these $U_{\Delta_{p+1},g}$ by $U(s)$. The authors append gauge-field translations to $U(s)$ to obtain $\tilde{U}(s)$, start from a $G$-symmetric state $\ket{0}$ with zero gauge field, and generate new states by operator sequences, labeling them by the gauge-field configuration $a$. They then assert that these $\tilde{U}(s)$ and the resulting states are hopping operators and configuration states, respectively. Finally, their Eq.~(108) computes the statistics of $\tilde{U}(s)$ as
\begin{equation}
	e^{i\Theta}=\langle0|\tilde{U}(s)^{\pm}|0\rangle=\langle0|U(s)^{\pm}|0\rangle.
	\label{eq:criticized-statistical-phase}
\end{equation}
and claims that
\begin{quote}
	$\Theta$ generally gives a
	microscopic definition of the ’t Hooft anomaly.
\end{quote}
However, this argument does not verify the configuration axiom, so the constructed operators need not satisfy the authors' definition of topological excitations. If this step fails, the phase in Eq.~\eqref{eq:criticized-statistical-phase} is ill-defined and cannot be called statistics. Here the requirements of symmetry and topological excitation diverge. From the excitation viewpoint, if $U_{\Delta_{p+1},g}$ is a hopping operator, then $U_g(\Sigma_{p+1})$ is a closed hopping operator, but it need not obey a group law. From the symmetry viewpoint, if $U_g(\Sigma_{p+1})$ is a global symmetry, its local truncation must implement that symmetry in its interior, but need not satisfy the configuration axiom. Ref.~\cite{Kobayashi_2026} conflates these two requirements.

We now rigorously disprove this claim using type-III anomaly. Assume first that the same family of local operators constructs both symmetry patches and admissible hopping operators. Combining Theorem~\ref{thm:endpoint-exchange} with Section~\ref{subsec:typeiii-statistical-process}, we derive an inverse-square relation between their anomaly and statistics. This relation directly excludes anomalies whose type-III part is not divisible by $2$.

For each $g\in G$, fix a local circuit decomposition
of $\Sg(g)$ whose elementary factors belong to the same hopping family, and require its interval truncations to be both symmetry patches and admissible extended hopping operators.
At shared endpoints, use the same local factors in the same multiplication order.
Write $\mathcal H_{\mathrm{conf}}$ for the span of configuration states.
The closed hopping $S(g)$ of Section~\ref{subsec:typeiii-statistical-process} is then $\Sg(g)$ here;
we retain that subsection's geometric orientation and statistical phase $\eta(g,h,k)$.

\begin{theorem}
 \label{thm:typeiii-inverse-square}
 Let $G$ be a finite Abelian group, and let the anomaly of a one-dimensional $G$ symmetry $\Sg$ be
 $[\nu]\in H^3(BG,\U(1))$.
 If $\Sg$ is also a family of admissible closed hopping operators for some $G$-type topological excitations
 and satisfies the truncation conditions above, then for all $g,h,k\in G$,
 \begin{equation}
  \Alt\nu(g,h,k)=\eta(g,h,k)^{-2}.
  \label{eq:typeiii-inverse-square}
 \end{equation}
 where $\Alt\nu$ and $\eta$ are the type-III parts of the anomaly and statistics, respectively.
\end{theorem}

\begin{proof}
Multiplying the equation in Theorem~\ref{thm:endpoint-exchange} on the left by
$D_L(g,h)^{-1}D_L(g,k)^{-1}$ gives
\begin{equation}
 \begin{aligned}
 &[D_L(g,h),D_L(g,k)]\,[D_L(g,k),\Sg(h)^{-1}]\\
 &\qquad=\Alt\nu(g,h,k)\,[D_L(g,h),\Sg(k)^{-1}].
 \end{aligned}
 \label{eq:typeiii-exchange-commutators}
\end{equation}
We show that the first commutator is $1$ on $\mathcal H_{\mathrm{conf}}$,
while the other two give the corresponding type-III statistical phases.

Take four well-separated points $x_0,x_1,x_2,x_3$ in order around the circle
and choose three truncations of the symmetry transformations:
\[
 P(t):[x_1,x_3],\qquad Q(t):[x_0,x_2],\qquad R(t):[x_1,x_2].
\]
Thus $P,R$ share the actual truncation at $x_1$, and $Q,R$ share that at $x_2$.
All four adjacent gaps around the circle are chosen large enough that the endpoint supports remain disjoint after the conjugations below.

\begin{figure}[H]
 \centering
 \begin{tikzpicture}[x=1cm,y=1cm,>=latex,every node/.style={font=\small}]
  \foreach \x/\lab in {0/x_0,3.1/x_1,6.2/x_2,9.3/x_3} {
   \draw[gray!55,densely dashed] (\x,-.4)--(\x,2.65);
   \node[above] at (\x,2.7) {$\lab$};
  }
  \foreach \y in {0,1.05,2.1} {
   \draw[gray!30] (0,\y)--(9.3,\y);
  }
  \draw[blue!65!black,very thick] (3.1,2.1)--(9.3,2.1);
  \draw[orange!80!black,very thick] (0,1.05)--(6.2,1.05);
  \draw[violet!75!black,very thick] (3.1,0)--(6.2,0);
  \node[left,blue!65!black] at (-.25,2.1) {$P(t)$};
  \node[left,orange!80!black] at (-.25,1.05) {$Q(t)$};
  \node[left,violet!75!black] at (-.25,0) {$R(t)$};
  \foreach \x/\y in {3.1/2.1,9.3/2.1,0/1.05,6.2/1.05,3.1/0,6.2/0} {
   \fill (\x,\y) circle (1.6pt);
  }
  \node[align=center,below] at (3.1,-.45) {$P,R$ share\\left truncation};
  \node[align=center,below] at (6.2,-.45) {$Q,R$ share\\right truncation};
 \end{tikzpicture}
 \caption{Interleaved intervals and their intersection. The operators $P,R$ share their truncation at $x_1$;
 conjugation by $Q$ agrees with the global symmetry near $x_1$ and acts trivially near $x_3$.}
 \label{fig:typeiii-interleaved-patches}
\end{figure}

The common truncation and Eq.~\eqref{eq:end} give
\begin{align}
 W_P(u,v)&:=[P(u)^{-1},\Sg(v)^{-1}]
 =D_L(u,v)D_{x_3}(u,v),
 \label{eq:typeiii-two-completions}
 \\
 W_R(u,v)&:=[R(u)^{-1},\Sg(v)^{-1}]
 =D_L(u,v)D_{x_2}(u,v).
 \label{eq:typeiii-R-completion}
\end{align}
Absorbing phase differences into the respective right endpoints makes $D_L$ the same in both equations.
The total configuration changes of $W_P,W_R$ vanish,
so each acts by a phase on every configuration state line.
The disjoint supports of the right-end factors $D_{x_3},D_{x_2}$ give
\begin{equation}
 [D_L(u,v),D_L(r,s)]=[W_P(u,v),W_R(r,s)].
 \label{eq:typeiii-endpoint-commutator}
\end{equation}
The right-hand side is $1$ on $\mathcal H_{\mathrm{conf}}$.

Next, conjugation by $Q(t)$ agrees with that by $\Sg(t)$ near $x_1$,
and it commutes with $D_{x_3}(u,v)$, so
\begin{equation}
 \begin{aligned}
 [D_L(u,v),\Sg(t)^{-1}]
 &=[W_P(u,v),Q(t)^{-1}]\\
 &=[[P(u)^{-1},\Sg(v)^{-1}],Q(t)^{-1}].
 \end{aligned}
 \label{eq:typeiii-local-response}
\end{equation}
The equivalent statistical process in Eq.~\eqref{eq:typeiii-equivalent-process} then gives
\begin{equation}
 [D_L(u,v),\Sg(t)^{-1}]\big|_{\mathcal H_{\mathrm{conf}}}
 =\eta(u,v,t)\one.
 \label{eq:typeiii-response-phase}
\end{equation}

Substituting the actions of these complete hopping products on $\mathcal H_{\mathrm{conf}}$
into Eq.~\eqref{eq:typeiii-exchange-commutators} yields
\begin{equation}
 \eta(g,k,h)=\Alt\nu(g,h,k)\eta(g,h,k).
 \label{eq:typeiii-response-exchange}
\end{equation}
The alternating property in Lemma~\ref{lem:typeiii-statistical-process} gives
$\eta(g,k,h)=\eta(g,h,k)^{-1}$,
and hence $\Alt\nu(g,h,k)=\eta(g,h,k)^{-2}$.
\end{proof}

\begin{corollary}
 \label{cor:typeiii-configuration-obstruction}
 Let $G$ be a finite Abelian group. If the type-III part of the anomaly of a one-dimensional $G$ symmetry $\Sg$
 is not divisible by $2$, meaning that it cannot be written as the pointwise square of an alternating function,
 then, under the truncation conditions above, $\Sg$ cannot provide any $G$-type topological excitations with
 a family of admissible closed hopping operators.

 In a decomposition $G=\prod_i\ZZ_{N_i}$, this condition is equivalent to the existence
 of $i<j<\ell$ in Eq.~\eqref{eq:abelian-cocycle-representative}
 such that $\gcd(N_i,N_j,N_\ell)$ is even and $p_{ij\ell}$ is odd.
 In particular, the $\ZZ_2^3$ anomaly $\frac12a_1a_2a_3$ has this obstruction.
\end{corollary}
\begin{proof}
Theorem~\ref{thm:typeiii-inverse-square} gives $\Alt\nu=(\eta^{-1})^2$.
Since $\eta^{-1}$ is still an alternating function, the type-III part of the anomaly must be divisible by $2$.
\end{proof}

\section{The relation between topological excitations and symmetry: a one-dimensional \texorpdfstring{$\ZZ_2$}{Z2} example}\label{sec:relation}

In this section, we seek the actual relation between symmetry and topological excitations. Since the axiomatic definition of topological excitations does not itself involve any symmetry, our task is to construct, under suitable assumptions, configuration states and hopping operators satisfying the excitation axioms naturally from a global symmetry. We consider the simplest case of a $1$-dimensional $\ZZ_2$ symmetry and argue that hopping operators should be regarded as symmetric operators commuting with the symmetry transformation $\mathcal{S}$, while
the subspace spanned by the configuration states $\{\confstate{a}\}$ corresponds to the symmetric subspace $\mathcal{S}=1$.

More concretely, we place one qubit at each vertex of an $n$-gon $X$, so the full many-body Hilbert space $\mathcal{H}$ has $\dim \mathcal{H}=2^n$, with basis
$\ket{v}$ labeled by $v\in C^0(X,\ZZ_2)$. Consider a $\ZZ_2$ symmetry $\mathcal S$ of the general form
\begin{equation}\label{eq:S-v}
	\mathcal S\ket{v}=(-1)^{f(v)}\ket{v+\one}.
\end{equation}
where $f(v)$ is a function depending locally on $v$, and $\one\in Z^0(X,\ZZ_2)$ is the cocycle with constant value
$1$.
The condition $\mathcal S^2=1$ means that
\begin{equation}
	f(v)+f(v+\one)=0\pmod 2.
\end{equation}

We now construct $\ZZ_2$ topological excitations. The first step is to construct the $2^{n-1}$ configuration states corresponding to $B^1(X,\ZZ_2)$. Each configuration $a\in B^1(X,\ZZ_2)$ can be written as a coboundary in two ways:
\begin{equation}
	a=\dd v=\dd(v+\one),\quad v\in C^0(X,\ZZ_2).
\end{equation}
The special form of $\mathcal{S}$ in Eq.~\eqref{eq:S-v} preserves the two-dimensional space $\mathcal{H}_a$ spanned by $\ket{v},\ket{v+\one}$. Its eigenvectors with eigenvalues $\pm1$ are $\ket{v}\pm(-1)^{f(v)}\ket{v+\one}$, respectively. We choose the normalized eigenvector of $\mathcal{S}$ with eigenvalue $+1$ as the configuration state:
\begin{equation}
	\confstate{a}
	\propto\ket{v}+(-1)^{f(v)}\ket{v+\one}.
	\label{eq:configuration-state}
\end{equation}
Replacing the representative $v$ by $v+\one$ multiplies the right-hand side only by the overall phase $(-1)^{f(v)}$, so this state line depends only on $a=\dd v$.

The configuration states thus span a $2^{n-1}$-dimensional subspace of $\mathcal{H}$, precisely the superselection sector $\mathcal{S}=1$.

The next step is to construct hopping operators. We assume the form
\begin{equation}
	U(s)\ket{v}
	=\exp \left(2\pi\ii\int_X L(s,v)\right)\ket{v+s},
	\label{eq:hopping-ansatz}
\end{equation}
where $s,v\in C^0(X,\ZZ_2)$ and $L(s,v)\in C^1(X,\RZ)$ is a natural cochain operation satisfying
$L(0,v)=0$. A hopping operator of this form automatically satisfies the locality axiom and maps each
configuration state $\confstate{a}$ into the two-dimensional subspace $\mathcal{H}_{a+\dd s}$. The remaining configuration axiom requires $U(s)\confstate{a}$ to lie in the $\mathcal{S}=1$ eigenspace, which is equivalent to commutation of $\mathcal{S}$ and $U(s)$:
\begin{equation}
	U(s)\mathcal S=\mathcal S U(s).
\end{equation}

Together with the Operator Independence Theorem, these observations give the following conclusion.

Fix a symmetry $\mathcal{S}$ of the form in Eq.~\eqref{eq:S-v} and the configuration states in Eq.~\eqref{eq:configuration-state}. An elementary $U(s)$ of the form in Eq.~\eqref{eq:hopping-ansatz} is a valid hopping operator for the topological excitations if and only if it \textbf{commutes with $\mathcal S$}; every choice satisfying this condition gives the same statistics.

By comparison, for a given symmetry $\mathcal{S}$, $P_I$ is a valid symmetry patch operator if and only if it \textbf{implements $\mathcal{S}$ in the interior of $I$}; every valid choice gives the same Else--Nayak index.

Thus, hopping operators and symmetry patch operators are both closely related to symmetry, but in two distinct ways.

As a concrete example, consider the anomalous $\ZZ_2$ symmetry
$f(v)=\int_Xv\smile\dd v$. Substituting Eqs.~\eqref{eq:S-v}
and~\eqref{eq:hopping-ansatz} into $U(s)\mathcal S=\mathcal S U(s)$ gives the equation determining
$L$:
\begin{align}
  \int_X\!\bigl[L(s,v+\one)-L(s,v)\bigr]
  &=\frac12\bigl[f(v+s)-f(v)\bigr]\pmod1\notag\\
  &=\frac12\int_X
    \bigl(\dd v\smile s+s\smile\dd v+s\smile\dd s\bigr)\pmod1.
  \label{eq:G}
\end{align}

Using $\dd(v+\one)=\dd v$, we obtain a solution
\begin{equation}
  L(s,v)
  =\frac12\,v\smile
   \bigl(\dd v\smile s+s\smile\dd v+s\smile\dd s\bigr).
  \label{eq:L0}
\end{equation}
The corresponding hopping operator is therefore
\begin{equation}
  U(s)\ket{v}
  =(-1)^{\int_Xv\smile
    (\dd v\smile s+s\smile\dd v+s\smile\dd s)}\ket{v+s}.
  \label{eq:U-s-action}
\end{equation}

We compute the statistical phase of this hopping family using the one-dimensional $\ZZ_2$ process in Eq.~\eqref{eq:one-dimensional-z2-statistics} and Fig.~\ref{fig:one-dimensional-z2-process}. Let $s$ be the elementary cochain equal to $1$ only at vertex $i$, and let $t$ equal $1$ only at the neighboring vertex $i+1$ or $i-1$. Equation~\eqref{eq:U-s-action} directly gives
\[
 U(s)^2\ket v=(-1)^{v_{i-1}+v_{i+1}}\ket v.
\]
Since $U(t)$ changes the value at exactly one of these two neighboring vertices,
\begin{equation}
  [U(s)^2,U(t)]=-1.
\end{equation}
This family of topological excitations therefore has nontrivial statistics.

When we set the hopping parameter $s$ to $\one$ on the entire space, we obtain
\begin{equation}
  U(\one)\ket{v}=\ket{v+\one}.
\end{equation}
Although this is a $\ZZ_2$ symmetry, it differs from $\mathcal{S}$ and is anomaly-free. This example shows that a hopping operator can sometimes be regarded as a symmetry patch operator for a global symmetry, while the anomaly of that symmetry can differ from the statistics of the hopping operator. The global symmetry actually corresponding to the statistics must be sought through $\mathcal{S}U(s)=U(s)\mathcal{S}$. We further discuss this relation in Section~\ref{sec:what-is-hopping}.

Another argument makes this mismatch more striking. Define the operator
\begin{equation}
	\Phi(s)\ket{v}
	=\ii^{\int_X\overline{\dd v}\smile\overline{s}}\ket{v}.
\end{equation}
Here $\overline{x}$ denotes the standard integer lift of $x\in\ZZ_N$ into $0,\ldots,N-1$; cochains are lifted simplex by simplex. The exponent uses the integer cup product of $\overline{\dd v}$ and $\overline{s}$.
It is easy to see that $\Phi(s)$ commutes with every symmetry $\mathcal{S}$ of the form in Eq.~\eqref{eq:S-v}. Hence
\begin{equation}
	U(s)\mapsto U(s)\Phi(s)
\end{equation}
is a valid redefinition of the kind discussed in Section~\ref{subsec:hopping-assembly} and, by the Operator Independence Theorem, leaves the statistics unchanged. On the other hand, the identity $\int_X\overline{\dd v}\equiv2\int_Xv\smile\dd v\pmod4$ gives
\begin{equation}
	\Phi(\one)\ket{v}=(-1)^{\int_X v\smile\dd v}\ket{v},
\end{equation}
so $U(\one)$ and $U(\one)\Phi(\one)$ have different Else--Nayak indices when regarded as $\ZZ_2$ global symmetries. The table below shows that the Else--Nayak index and statistics can not only differ but can take all possible combinations. Here the hopping operator $U(s)$ is given by Eq.~\eqref{eq:U-s-action}, whereas $X(s)$ simply shifts the field values:
\begin{equation}
	X(s)\ket{v}=\ket{s+v}.
\end{equation}

The first two rows and the last two rows of the table are, respectively, valid redefinitions on the same family of configuration states.

\begin{table}[H]
  \centering
\caption{Statistics of different $\ZZ_2$ hopping operators and the Else--Nayak indices of their closed actions.}
	\label{tab:z2-statistics-anomaly}
	\begin{tabularx}{0.94\textwidth}{@{}>{\bfseries}p{0.3\textwidth}
			>{\raggedright\arraybackslash}X>{\raggedright\arraybackslash}X@{}}
		\toprule Hopping operator
		& Else--Nayak index & Statistics \\
		\midrule
		$X(s)$ & $1$
		& $1$ \\
		$X(s)\Phi(s)$ & $-1$
		& $1$ \\
		\midrule
		$U(s)$ & $1$
		& $-1$ \\
		$U(s)\Phi(s)$ & $-1$
		& $-1$ \\
		\bottomrule
	\end{tabularx}
\end{table}

Table~\ref{tab:z2-statistics-anomaly} and Theorem~\ref{thm:typeiii-inverse-square} illustrate two difficulties with identifying symmetry patches and hopping operators. Theorem~\ref{thm:typeiii-inverse-square} gives a factor of $-2$ between the type-III parts of anomaly and statistics, whereas the table shows that anomaly and statistics can vary independently in the $\ZZ_2$ example. Section V of Ref.~\cite{Kobayashi_2026} claims that the statistics of symmetry patch operators can reveal whether the original symmetry can be gauged; Table~\ref{tab:z2-statistics-anomaly} provides a direct counterexample.

For a one-dimensional $\ZZ_2$ symmetry $\mathcal{S}$, there is a possible remedy: when truncating $\mathcal{S}$ to obtain a symmetry patch operator $P_I$, require not only the topological-excitation axioms but also a symmetric truncation, namely commutation with $\mathcal{S}$. In Table~\ref{tab:z2-statistics-anomaly}, this condition indeed excludes the two cases in which the Else--Nayak index and statistics differ. However, Corollary~\ref{cor:no-symmetric-truncation} already shows that some anomalous symmetries admit no symmetric truncation. This remedy is therefore not universal.

\section{Comparing topological excitations and symmetry through gauge theory}
\label{sec:what-is-hopping}

The one-dimensional example in Section~\ref{sec:relation} established the relation between hopping operators $U$ and a global symmetry $\mathcal{S}$: to identify the statistics of $U$ with the anomaly of $\mathcal{S}$, the crucial property of $U$ is that it commutes with $\mathcal{S}$, rather than implements the symmetry. This section extends this statement to a $(q-1)$-form symmetry in $d$ spatial dimensions with anomaly $[\omega]\in H^{d+2}(B^qG,\RZ)$, where $1\leq q\leq d$, and explains the physical meanings of hopping and symmetry and their relation through gauging.

We use the following cochain ansatz. The classical matter field is $v\in C^{q-1}(X,G)$, and global symmetry acts on the classical variables by shifts $v\mapsto v+z$ with closed cochains $z\in Z^{q-1}(X,G)$.
The quantum theory has matter basis states $\ket v_X$. A quantum symmetry may supplement the closed-cochain shift with a phase depending locally on $v$, while satisfying the group law.
On the same matter space, we also consider shifts by arbitrary $s\in C^{q-1}(X,G)$ and adopt the following cochain ansatz for candidate hopping operators:
\begin{equation}\label{eq: ansatz}
  U(s)\ket v_X
  =\exp\left(2\pi\ii\int_XL(s,v)\right)\ket{v+s}_X,
  \qquad s,v\in C^{q-1}(X,G),
\end{equation}
The local phase is a sum of functions on $d$-simplices, each depending only on the fields and shift parameters within that simplex; different simplices may use different functions. To construct explicit solutions uniformly from $[\omega]$, we choose a normalized natural local cochain operation
$L:C^{q-1}(-,G)\times C^{q-1}(-,G)\longrightarrow C^d(-,\RZ)$ satisfying $L(0,v)=0$, which amounts to using the same compatible formula on all simplices. Naturality provides a construction method here; it is not a necessary condition for the valid redefinitions of Section~\ref{subsec:hopping-assembly}.
A closed parameter $z$ preserves $\dd v$, whereas a general parameter $s$ gives $\dd v\mapsto\dd v+\dd s$.
Even when $s=z$, $\mathcal S(z)$ and $U(z)$ merely shift the same basis labels; their quantum phases must still be determined separately.
We construct these two kinds of phases from the same $[\omega]$ and explain how the symmetry group law and hopping axioms constrain them.

The spatial manifold is a triangulated sphere $X\simeq S^d$, bounding a bulk $M\simeq D^{d+1}$ with $X=\partial M$. We first start from DW theory with twist $[\omega]$ and impose only the interior Gauss constraints. The surviving boundary degrees of freedom are described by the restriction $a=A_X$ of the gauge field. Operators $T(t)$ commuting with the interior Gauss action directly give physical hopping on these degrees of freedom.

We then introduce independent boundary matter and couple it to the bulk gauge field. The gauge-invariant variable of the coupled system is $a=A_X+\dd v$: fixing $v=0$ recovers the DW boundary, whereas fixing the gauge field $A$ gives the matter-field description. In the latter representation, Gauss transformations preserving the background give the global symmetry, the remaining Gauss constraints select the symmetric subspace, and gauge-invariant hopping acts within this subspace. The statistics of the hopping constructed this way equal the anomaly of the symmetry. Within the ansatz of Eq.~\eqref{eq: ansatz}, the configuration axiom on the corresponding configuration states is also equivalent to the symmetric condition. Equation~\eqref{eq:gauging-hopping-diagram} summarizes these three spaces and the correspondence between their hopping operators.

\subsection{Gauge transformations in higher-form Dijkgraaf--Witten theory}
\label{subsec:dw-gauge-transformations}

Let $G$ be a finite Abelian group and $q\geq1$, and denote the bulk spacetime dimension by
\begin{equation}
  D:=d+2.
\end{equation}

In Dijkgraaf--Witten theory, we consider only flat gauge fields,
$A\in Z^q(N,G)$, where $N$ is a triangulated $D$-dimensional spacetime. Classically, a gauge transformation is specified by
$t\in C^{q-1}(N,G)$ and acts as
\begin{equation}
	t:A\longrightarrow A+\dd t,
	\qquad A\in Z^q(N,G).
	\label{eq:classical-gauge-transformation}
\end{equation}
Applying $t_1$ and $t_2$ successively changes the field along
\begin{equation}
	A\xrightarrow{\ t_1\ }A+\dd t_1
	\xrightarrow{\ t_2\ }A+\dd(t_1+t_2)
	\label{eq:two-classical-gauge-transformations}
\end{equation}
which is equivalent to the direct gauge transformation labeled by $t_1+t_2$. The two paths form a commutative triangle
\[
\gaugetriangle{A}{A+\dd t_1}{A+\dd(t_1+t_2)}
  {t_1}{t_2}{t_1+t_2}
\]

At the quantum level, the gauge field $A$ is accompanied by a cocycle $\omega(A)\in Z^D(N,\RZ)$, where
\begin{equation}
  \omega:Z^q(-,G)\longrightarrow Z^D(-,\RZ).
  \label{eq:twist-cocycle-operation}
\end{equation}
is a cocycle operation. Integrating it over $D$-dimensional spacetime gives the action phase
\begin{equation}
	\exp \left(2\pi\ii\int_N\omega(A)\right).
	\label{eq:twisted-gauge-action}
\end{equation}
A gauge transformation of $A$ induces a corresponding transformation of $\omega(A)$. By the properties of the cocycle operation, $\omega(A+\dd t)$ and $\omega(A)$ represent the same cohomology class, so they differ by a coboundary, or equivalently a homotopy:
\begin{equation}
	\omega(A+\dd t)-\omega(A)=\dd \Theta^1(t;A).
  \label{eq:theta-1-descent}
\end{equation}
Here $\Theta^1$ is a natural map $C^{q-1}(-,G)\times Z^q(-,G)\to C^{D-1}(-,\RZ)$. Corresponding to $A\xrightarrow{t}A+\dd t$, we have
\begin{equation}
	\omega(A)\xrightarrow{\Theta^1(t;A)}\omega(A+\dd t).
\end{equation}
For successive transformations $t_1,t_2$, the sum of the phases on the two edges generally differs from that on the direct edge; their difference is the next coboundary:
\begin{equation}
	\dd \Theta^2(t_2,t_1;A)
	=\Theta^1(t_1;A)+\Theta^1(t_2;A+\dd t_1)-\Theta^1(t_1+t_2;A).
  \label{eq:theta-2-descent}
\end{equation}
This coboundary is given by $\Theta^2(t_2,t_1;A)$ in the interior of the triangle:
\[
\begin{tikzpicture}[baseline=(current bounding box.center),
    every node/.style={font=\small,inner sep=2pt},>=latex]
  \node[anchor=east] (g0) at (-3.0,0) {$\omega(A)$};
  \node[anchor=south] (g1) at (0,1.85) {$\omega(A+\dd t_1)$};
  \node[anchor=west] (g2) at (3.0,0) {$\omega(A+\dd(t_1+t_2))$};
  \draw[->,shorten >=2pt] (g0.east) --
    node[above left,font=\footnotesize] {$\Theta^1(t_1;A)$} (g1.south);
  \draw[->,shorten <=2pt,shorten >=2pt] (g1.south) --
    node[above right,font=\footnotesize] {$\Theta^1(t_2;A+\dd t_1)$} (g2.west);
  \draw[->,shorten >=2pt] (g0.east) --
    node[below,font=\footnotesize] {$\Theta^1(t_1+t_2;A)$} (g2.west);
  \node at (0,1.0) {$\Big\Uparrow$};
  \node[font=\footnotesize] at (0,0.43) {$\Theta^2(t_2,t_1;A)$};
\end{tikzpicture}
\]
For three successive gauge transformations, four $\Theta^2$ terms form the four faces of a tetrahedron. We use
$\Theta^3(t_3,t_2,t_1;A)\in C^{D-3}(N,\RZ)$ to describe their oriented sum, with descent relation
\begin{align}
  \dd\Theta^3(t_3,t_2,t_1;A)
  &=\Theta^2(t_3,t_2;A+\dd t_1)
    -\Theta^2(t_3,t_1+t_2;A)\notag\\
  &\quad+\Theta^2(t_2+t_3,t_1;A)
    -\Theta^2(t_2,t_1;A).
  \label{eq:theta-3-descent}
\end{align}

Appendix~\ref{app:descent-construction} constructs all these phases from the same $\omega$ using prism integrals. In the operator constructions below, $\Theta^1$ gives the phase of the gauge transformation $T(t)$, $\Theta^2$ gives the phases of boundary matter transformations and hopping, and $\Theta^3$ ensures their composition relations.

We use a normalized cocycle operation $\omega$, whose descent phases satisfy
\begin{equation}\label{eq:normalized}
	\Theta^k(t_k,\cdots,t_{i+1},0,t_{i-1},\cdots,t_1;A)=0,\forall i.
\end{equation}

\subsection{Gauss constraints and gauge averaging}
\label{subsec:gauge-invariant-time-slice}

We now construct the gauge-invariant Hilbert space of Dijkgraaf--Witten theory on a spatial manifold $M$. Let
$\ket{A}$ denote the basis state of a flat gauge field $A\in Z^q(M,G)$. These states span the kinematic gauge-field space
\begin{equation}
  \mathcal H_{\mathrm{DW}}:=\operatorname{span}\{\ket A:A\in Z^q(M,G)\}.
  \label{eq:DW-kinematic-Hilbert-space}
\end{equation}
The Gauss constraints select a subspace of this space. A gauge parameter
$t\in C^{q-1}(M,G)$ acts on the gauge field $A$ as an arrow $A\xrightarrow{t}A+\dd t$. Upon quantization, it becomes the operator $T(t)$ acting on basis states:
\begin{equation}
  T(t)\ket{A}=
  \exp \left(2\pi\ii\int_M\Theta^1(t;A)\right)
  \ket{A+\dd t}.
  \label{eq:quantum-gauge-transformation}
\end{equation}
The phase is the integral of the twist action along the time direction of a spacetime cylinder. Compare $T(t_2)T(t_1)\ket{A}$ and $T(t_2+t_1)\ket{A}$. Their gauge fields agree, while their phases correspond to the three edges of the triangle
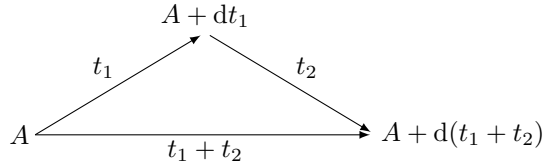
\[
\gaugetriangle{A}{A+\dd t_1}{A+\dd(t_1+t_2)}
  {t_1}{t_2}{t_1+t_2}
\]
By Eq.~\eqref{eq:theta-2-descent}, these phases differ by the $2$-homotopy $\dd \Theta^2(t_2,t_1;A)$. Hence

\begin{equation}
  T(t_2)T(t_1)\ket{A}
  =\exp \left(2\pi\ii\int_{\partial M}\Theta^2(s_2,s_1;A_{\partial M})\right)
   T(t_1+t_2)\ket{A}.
  \label{eq:spatial-gauge-composition-operator}
\end{equation}
Here $s_i:=(t_i)_{\partial M}$ is the boundary restriction of the bulk parameter. By Eq.~\eqref{eq:normalized}, if at least one of $t_1,t_2$ vanishes on the boundary,
the boundary phase in Eq.~\eqref{eq:spatial-gauge-composition-operator} vanishes:
\begin{equation}
  T(t_2)T(t_1)=T(t_1+t_2),
  \qquad
  (t_1)_{\partial M}=0\quad\text{or}\quad(t_2)_{\partial M}=0.
  \label{eq:gauge-composition-trivial-boundary-leg}
\end{equation}
Below we use this multiplication property whenever one gauge parameter vanishes on the boundary.

When $M$ is closed, the boundary phase vanishes, and $T(t)$ is a strict representation of the additive group
$C^{q-1}(M,G)$. Define the gauge average
\begin{equation}
  \Pi_{\mathrm{DW}}
  :=\frac{1}{|C^{q-1}(M,G)|}
    \sum_{t\in C^{q-1}(M,G)}T(t).
  \label{eq:DW-closed-projector}
\end{equation}
The strict group law and unitarity make this an orthogonal projector onto the space satisfying all Gauss constraints:
\begin{equation}
  \operatorname{im}\Pi_{\mathrm{DW}}
  =\{\ket\psi\in\mathcal H_{\mathrm{DW}}:
      T(t)\ket\psi=\ket\psi\ \text{for all }t\in C^{q-1}(M,G)\}.
  \label{eq:DW-closed-invariant-space}
\end{equation}

States in this space can be constructed from a flat field $A\in Z^q(M,G)$ by computing $\Pi_{\mathrm{DW}}\ket{A}$. The result is a superposition of basis states in $[A]\in H^q(M,G)$ reached by gauge transformations. Parameters preserving $A$ are cocycles
$t\in Z^{q-1}(M,G)$, and the corresponding gauge transformations act as
\begin{equation}\label{eq:closed-gauge-stabilizer}
  T(t)\ket A=\exp\left(2\pi\ii\int_M\Theta^1(t;A)\right)\ket A.
\end{equation}
If all these stabilizer phases equal $1$, the orbit contributes a one-dimensional gauge-invariant subspace; if any phase is nontrivial, the gauge average cancels and the orbit contributes zero.
Thus, the dimension of $\operatorname{im}\Pi_{\mathrm{DW}}$ depends on $[\omega]$ and the global topology of $M$. This is precisely the Hilbert space $Z(M)$ assigned by Dijkgraaf--Witten topological field theory, or the ground-state subspace when the theory is regarded as a topological order.

In our range $1\leq q\leq d$, all flat fields on a sphere $M\simeq S^{d+1}$ belong to one gauge orbit. Take a disk bounded by $M$ and extend $A$ and an arbitrary closed parameter $t$ into it as closed cochains. Such extensions exist because $A$ is a coboundary; for $q>1$, $t$ is also a coboundary, while for $q=1$, $t$ is constant. Since the extended $t$ remains closed, Eq.~\eqref{eq:theta-1-descent} gives $\dd\Theta^1(t;A)=0$. Stokes' theorem then yields
\begin{equation}
  \int_M\Theta^1(t;A)=0,
  \qquad T(t)\ket A=\ket A.
  \label{eq:sphere-gauge-stabilizer-trivial}
\end{equation}
All stabilizer phases are therefore trivial, and the unique gauge orbit contributes a one-dimensional invariant space: $\dim Z(S^{d+1})=1$.

\subsection{The charge-condensation boundary of DW theory}
\label{subsec:boundary-excitation-statistics}

Now consider $M\simeq D^{d+1}$ and $X:=\partial M\simeq S^d$. The nontrivial phase in Eq.~\eqref{eq:spatial-gauge-composition-operator} prevents us from treating $T(t)$ for every $t\in C^{q-1}(M,G)$ as a valid gauge transformation or Gauss-law operator. Instead, we treat only $T(u)$ for interior parameters $u\in C^{q-1}(M,G)$ satisfying $u_X=0$ as gauge transformations. Throughout this subsection, $u$ denotes a $(q-1)$-cochain with $u_X=0$.

Classically, interior gauge transformations send a flat gauge field $A\in Z^q(M,G)$ to $A+\dd u$ while preserving its boundary value $a:=A_X$. Since $M$ is a disk, $a$ ranges over $B^q(X,G)$, and two bulk extensions of the same $a$ always differ by $\dd u$. Thus $a$ is a complete gauge-invariant description of the physical degrees of freedom under interior gauge transformations.

There is no matter field in this subsection, so gauge invariance means the Gauss constraints on the gauge field. Since only interior Gauss constraints are imposed, charges can terminate on $\partial M=X$, giving a charge-condensation boundary. The invariant $a$ can be interpreted as boundary gauge flux: imagine extending $M\simeq D^{d+1}$ outward beyond its boundary and extending $A$ by zero outside. Although $A$ is flat inside $M$, a nonzero curvature appears at its boundary with the exterior $A=0$, and is described precisely by $a=A_X$. For a general cochain $t\in C^{q-1}(M,G)$, the change $A\mapsto A+\dd t$ is no longer a gauge transformation; it changes the boundary flux configuration as $a\mapsto a+\dd s$, where $s=t_X$.

We now quantize this picture. Averaging over interior gauge transformations $T(u)$ gives
\begin{equation}
  \Pi_{\mathrm{DW},\partial}
  :=\frac{1}{|\{u\in C^{q-1}(M,G):u_X=0\}|}
    \sum_{\substack{u\in C^{q-1}(M,G)\\u_X=0}}T(u).
  \label{eq:DW-boundary-projector}
\end{equation}
By Eq.~\eqref{eq:gauge-composition-trivial-boundary-leg}, this is an orthogonal projector onto the gauge-invariant subspace
\begin{equation}
  \begin{aligned}
  \mathcal H_{\mathrm{DW},\partial}
  &:=\operatorname{im}\Pi_{\mathrm{DW},\partial}\\
  &=\{\ket\phi\in\mathcal H_{\mathrm{DW}}:
      T(u)\ket\phi=\ket\phi\text{ for all }u_X=0\}.
  \end{aligned}
  \label{eq:DW-boundary-Hilbert-space}
\end{equation}
The boundary state associated with $A$ is
\begin{equation}
  \ket A_\partial:=\Pi_{\mathrm{DW},\partial}\ket A.
  \label{eq:boundary-state-from-bulk-field}
\end{equation}
This averaged state is nonzero by the interior-gauge-transformation case of Lemma~\ref{lem:boundary-gauss-action} below.

Since $T(u)$ commutes with $\Pi_{\mathrm{DW},\partial}$, we obtain
\begin{equation}
  \ket{A+\dd u}_\partial
  =\exp\left(-2\pi\ii\int_M\Theta^1(u;A)\right)\ket A_\partial.
  \label{eq:boundary-state-change-of-extension}
\end{equation}
Thus different $\ket{A}_\partial$ with the same boundary flux $a=A_X$ are collinear, whereas states with different $a$ are orthogonal. Choose a normalized representative $\confstate a$ of each line. Then
\begin{equation}
  \mathcal H_{\mathrm{DW},\partial}
  =\bigoplus_{a\in B^q(X,G)}\mathbb C\confstate a.
  \label{eq:boundary-configuration-lines}
\end{equation}
Next we construct hopping operators. By Eq.~\eqref{eq:gauge-composition-trivial-boundary-leg}, a general $T(t)$ commutes with all $T(u)$, so
\begin{equation}
  T(t)\Pi_{\mathrm{DW},\partial}
  =\Pi_{\mathrm{DW},\partial}T(t).
  \label{eq:boundary-hopping-commutes-with-projector}
\end{equation}
Thus $T(t)$ acts unitarily on the boundary state space $\mathcal H_{\mathrm{DW},\partial}$. Two bulk extensions of the same boundary parameter $s$ differ only by an interior parameter $u$. Since $T(t+u)=T(t)T(u)$ and $T(u)$ acts trivially on this space, the action of $T(t)$ depends only on $s$. We therefore define
\begin{equation}
  T_\partial(s):=T(t)\big|_{\mathcal H_{\mathrm{DW},\partial}},
  \qquad t_X=s\in C^{q-1}(X,G),
  \label{eq:boundary-hopping-operator}
\end{equation}
By definition,
\begin{equation}
  T_\partial(s)\ket A_\partial
  =\exp\left(2\pi\ii\int_M\Theta^1(t;A)\right)
    \ket{A+\dd t}_\partial.
  \label{eq:boundary-hopping-on-bulk-seed}
\end{equation}
This implements the classical transition $a\mapsto a+\dd s$ described above. Although its formula comes from a gauge transformation, its actual role is to commute with the interior Gauss action and act on physical boundary degrees of freedom. It is therefore a gauge-invariant hopping operator rather than a gauge transformation, paralleling the distinction between a symmetry transformation and a symmetric operator.

Equations~\eqref{eq:boundary-state-change-of-extension} and \eqref{eq:boundary-hopping-on-bulk-seed} are precisely the configuration axiom for $\confstate{a}$ and $T_\partial(s)$ as topological excitations. The other requirement, the locality axiom, follows from the locality of the phase in Eq.~\eqref{eq:spatial-gauge-composition-operator}.
\begin{equation}
  T_\partial(s_2)T_\partial(s_1)\confstate a
  =\exp\left(2\pi\ii\int_X\Theta^2(s_2,s_1;a)\right)
    T_\partial(s_1+s_2)\confstate a.
  \label{eq:boundary-hopping-composition}
\end{equation}
This composition formula also gives the assembly relation in Section~\ref{subsec:hopping-assembly}: after decomposing the parameters into elementary cochains, the local phases produced by assembly can be absorbed into valid redefinitions of the corresponding elementary hopping operators.
With the group commutator convention $[U,V]=U^{-1}V^{-1}UV$,
\begin{equation}
  [T_\partial(s_2),T_\partial(s_1)]\confstate a
  =\exp\left(2\pi\ii\int_X
    [\Theta^2(s_2,s_1;a)-\Theta^2(s_1,s_2;a)]\right)\confstate a.
  \label{eq:boundary-hopping-commutator}
\end{equation}
For elementary cochains $s_1,s_2$, locality of $\Theta^2$ and its vanishing when either parameter is zero imply that only $d$-simplices containing both $(q-1)$-simplices contribute. Taking a further commutator with $T_\partial(s)$ gives the difference of this local phase at $a$ and $a+\dd s$, which is nonzero only on the common support of all three operators. Iterating, the nested commutator is trivial whenever the elementary cochains $s_1,\ldots,s_k$ do not all lie in any one $d$-simplex. This is precisely the locality axiom. The description in terms of boundary flux in DW theory thus gives a natural interpretation of the topological-excitation axioms proposed in Ref.~\cite{Xue2026Statistics}. The statistics class is $[\omega]\in H^{d+2}(B^qG,\RZ)$; see Ref.~\cite{XueWen2026Holographic}.

\subsection{Coupling boundary matter to the bulk gauge field}
\label{subsec:gauging-perspective}

The previous subsection interpreted topological excitations as boundary gauge flux on a charge-condensation boundary of DW theory in one higher spatial dimension. We now introduce an independent boundary matter field $v\in C^{q-1}(X,G)$ and its gauge transformations, and couple it to the DW theory to obtain the full gauge theory. This perspective identifies topological excitations with the physical degrees of freedom remaining after gauging, and hopping with gauge-invariant operators. The gauge choice $v=0$ recovers the preceding construction.

Classically, the combined gauge transformation of the gauge and matter fields is
\begin{equation}
  A\longmapsto A-\dd t,\qquad
  v\longmapsto v+t_X,\qquad t\in C^{q-1}(M,G),
  \label{eq:coupled-classical-gauge-transformation}
\end{equation}
whose invariant is precisely the configuration of the topological excitations:
\begin{equation}
  a:=A_X+\dd v\in B^q(X,G).
  \label{eq:gauge-invariant-configuration}
\end{equation}
Different pairs $(A,v)$ describe the same $a$ and are related by gauge transformations. Here $v$ can be regarded as a path from $A_X$ to $a$; a gauge transformation moves its starting point while preserving its endpoint $a$.

In the quantum theory, we add boundary matter to the gauge-field space $\mathcal H_{\mathrm{DW}}$, with matter Hilbert space
\begin{equation}
  \mathcal H_{\mathrm{matter}}
  =\operatorname{span}\{\ket v_X:v\in C^{q-1}(X,G)\}.
  \label{eq:gauging-source-spaces}
\end{equation}
This gives the combined space
\begin{equation}
	\mathcal H_{\mathrm{total}}
	:=\mathcal H_{\mathrm{DW}}\otimes\mathcal H_{\mathrm{matter}}
	=\operatorname{span}\{\ket{A,v}:=\ket A_M\otimes\ket v_X\}.
	\label{eq:gauging-kinematic-space}
\end{equation}

We implement the gauge-field transformation $A\mapsto A-\dd t$ by $T^{-1}(t)$. For fixed background $A$, define the gauge transformation of boundary matter by
\begin{equation}
  R(s;A_X)\ket v_X
  :=\exp\left(2\pi\ii\int_X\Theta^2(v,s;A_X-\dd s)\right)\ket{v+s}_X
  \label{eq:boundary-compensating-operator}
\end{equation}
where $s=t_X$. Since $s$ is composed before $v$, the phase is $\Theta^2(v,s;A_X-\dd s)$, with starting point $A_X-\dd s$.

The gauge and matter fields together give the full gauge transformation
\begin{equation}
  G(t)\bigl(\ket A_M\otimes\ket v_X\bigr)
  =T(t)^{-1}\ket A_M\otimes R(t_X;A_X)\ket v_X.
  \label{eq:coupled-gauss-factorization}
\end{equation}
Equation~\eqref{eq:theta-2-descent} gives the boundary phase from composing two $T^{-1}$ operators, while Eq.~\eqref{eq:theta-3-descent} gives the opposite phase from composing two $R$ operators. Hence
\begin{equation}
  G(t_2)G(t_1)=G(t_1+t_2).
  \label{eq:coupled-gauss-group-law}
\end{equation}

Thus $G(t)$ is a family of unitary gauge transformations satisfying the strict group law. On the same combined space, define
\begin{equation}
  \mathcal U(s)\ket{A,v}
  :=\exp\left(2\pi\ii\int_X\Theta^2(s,v;A_X)\right)\ket{A,v+s},
  \qquad s\in C^{q-1}(X,G).
  \label{eq:gauged-hopping}
\end{equation}
The operator $\mathcal U(s)$ gives $a\mapsto a+\dd s$, corresponding to appending a segment $s$ at the endpoint of the path $v:A_X\to a$. The following diagrams show the compositions associated with $R(s;A_X)$ and $\mathcal U(s)$, respectively:
\begin{center}
\begin{minipage}[t]{0.48\textwidth}
  \centering
  \small Change the starting point, preserve $a$\par\smallskip
  \begin{tikzpicture}[>=latex,every node/.style={font=\small,inner sep=2pt}]
    \coordinate (g0) at (-1.45,0);
    \coordinate (g1) at (0,1.1);
    \coordinate (g2) at (1.45,0);
    \node[below] at (g0) {$A_X-\dd s$};
    \node[above] at (g1) {$A_X$};
    \node[below] at (g2) {$a=A_X+\dd v$};
    \draw[->,shorten >=2pt] (g0)--node[above left] {$s$} (g1);
    \draw[->,shorten <=2pt,shorten >=2pt] (g1)--node[above right] {$v$} (g2);
    \draw[->,shorten >=2pt] (g0)--(g2);
    \node at (0,-0.48) {$v+s$};
  \end{tikzpicture}
\end{minipage}\hfill
\begin{minipage}[t]{0.48\textwidth}
  \centering
  \small Preserve the starting point, change $a$\par\smallskip
  \begin{tikzpicture}[>=latex,every node/.style={font=\small,inner sep=2pt}]
    \coordinate (g0) at (-1.45,0);
    \coordinate (g1) at (0,1.1);
    \coordinate (g2) at (1.45,0);
    \node[below] at (g0) {$A_X$};
    \node[above] at (g1) {$a=A_X+\dd v$};
    \node[below] at (g2) {$a+\dd s$};
    \draw[->,shorten >=2pt] (g0)--node[above left] {$v$} (g1);
    \draw[->,shorten <=2pt,shorten >=2pt] (g1)--node[above right] {$s$} (g2);
    \draw[->,shorten >=2pt] (g0)--(g2);
    \node at (0,-0.48) {$v+s$};
  \end{tikzpicture}
\end{minipage}
\end{center}
Substituting $(s,v,t_X)$ into Eq.~\eqref{eq:theta-3-descent} at $A_X-\dd t_X$ and integrating over the closed space $X$ yields

\begin{equation}
  \mathcal U(s)G(t)=G(t)\mathcal U(s).
  \label{eq:gauged-hopping-invariance}
\end{equation}
Averaging over all gauge transformations gives the projector onto the gauge-invariant subspace $G(t)=1$:
\begin{equation}
  \Pi_{\mathrm{gauged}}
  :=\frac1{|C^{q-1}(M,G)|}\sum_{t\in C^{q-1}(M,G)}G(t),
  \qquad
  \mathcal H_{\mathrm{gauged}}:=\operatorname{im}\Pi_{\mathrm{gauged}}.
  \label{eq:gauged-projector}
\end{equation}
The following lemma shows that the quantum phases from different parameters do not cancel the gauge average.

\begin{lemma}
\label{lem:boundary-gauss-action}
For $M\simeq D^{d+1}$, $X=\partial M$, and $1\leq q\leq d$, any $t,t'\in C^{q-1}(M,G)$ and basis state $\ket{A,v}$ satisfy
\begin{equation}
  G(t)\ket{A,v}\propto G(t')\ket{A,v}
  \quad\Longleftrightarrow\quad
  G(t)\ket{A,v}=G(t')\ket{A,v}.
  \label{eq:coupled-gauss-ray-equality}
\end{equation}
\end{lemma}

\begin{proof}
Proportional basis states have the same labels, so $u:=t-t'$ satisfies $\dd u=0$ and $u_X=0$. Glue $M$ to an oppositely oriented copy along $X$ to form a sphere, using the same gauge field $A$ and zero parameter on the copy. The resulting gauge field and parameter are closed, so Eq.~\eqref{eq:sphere-gauge-stabilizer-trivial} makes the phase on the sphere trivial. The copy contributes zero by normalization, so the phase on the original $M$ also vanishes. Thus $T(u)\ket A=\ket A$; since $u_X=0$, we also have $G(u)\ket{A,v}=\ket{A,v}$. The strict group law of $G$ now gives the result.
\end{proof}

Therefore $\Pi_{\mathrm{gauged}}\ket{A,v}$ is nonzero. Averaged states with the same $a=A_X+\dd v$ differ only by a phase, and those with different $a$ are orthogonal. Every orbit has stabilizers given by $\dd u=0$, $u_X=0$, so these averaged states have the same norm. Suppressing the common normalization, denote them by $\ket a_{\mathrm{gauged}}$. We obtain the direct-sum decomposition
\begin{equation}
	\mathcal H_{\mathrm{gauged}}
	=\bigoplus_{a\in B^q(X,G)}\mathbb C\ket a_{\mathrm{gauged}},
	\qquad\dim\mathcal H_{\mathrm{gauged}}=|B^q(X,G)|.
	\label{eq:gauged-configuration-decomposition}
\end{equation}

Equation~\eqref{eq:gauged-hopping-invariance} ensures that $\mathcal U(s)$ preserves $\mathcal H_{\mathrm{gauged}}$. It maps the ray of $\ket a_{\mathrm{gauged}}$ to that of $\ket {a+\dd s}_{\mathrm{gauged}}$, precisely as required by the configuration axiom. The local integral in Eq.~\eqref{eq:gauged-hopping} immediately implies the locality axiom, completing the definition of the topological excitations.

The correspondence between the two constructions is seen by fixing $v=0$. Any matter field $v$ can be set to zero by a gauge transformation with $t_X=-v$, leaving the interior transformations $G(u)=T(u)^{-1}$. Thus, for a gauge-invariant state $\ket\Psi=\sum_{A,v}\Psi(A,v)\ket{A,v}$, fixing $v=0$ gives
\begin{equation}
  \operatorname{Fix}_{v=0}\ket\Psi:=\ket\phi
  =\sum_A\Psi(A,0)\ket A_M
  \in\mathcal H_{\mathrm{DW},\partial}.
  \label{eq:gauging-vzero-components}
\end{equation}
The remaining Gauss constraints are precisely $T(u)\ket\phi=\ket\phi$. Conversely, $\Pi_{\mathrm{gauged}}(\ket\phi\otimes\ket0_X)$ reconstructs the original state: its $v=0$ component receives contributions only from $t_X=0$ and is proportional to $\Pi_{\mathrm{DW},\partial}\ket\phi=\ket\phi$, while all other components are fixed by the Gauss constraints. Gauge fixing $\operatorname{Fix}_{v=0}$ and gauge averaging $\Pi_{\mathrm{gauged}}$ establish an isomorphism $\mathcal H_{\mathrm{DW},\partial}\simeq \mathcal H_{\mathrm{gauged}}$, mapping the configuration states one-to-one up to a constant normalization.

Hopping operators also correspond one-to-one. For $t_X=s$, $\mathcal U(s)$ sends the representative $(A,0)$ to $(A,s)$, and $G(-t)$ then sends it to $(A+\dd t,0)$, restoring $v=0$. On gauge-averaged states, $G(-t)$ acts as the identity. Normalization and Eq.~\eqref{eq:theta-2-descent} give the total phase of these two steps as
\begin{equation}
  -\int_M\Theta^1(-t;A+\dd t)
  +\int_X\Theta^2(s,-s;A_X+\dd s)
  =\int_M\Theta^1(t;A).
  \label{eq:gauging-restoration-phase}
\end{equation}
The right-hand side is precisely the phase of $T(t)$, so on $\mathcal H_{\mathrm{gauged}}$,
\begin{equation}
  \operatorname{Fix}_{v=0}\mathcal U(s)
  =T_\partial(s)\operatorname{Fix}_{v=0}.
  \label{eq:gauging-vzero-intertwining}
\end{equation}
The two constructions therefore have the same statistics $[\omega]$.

\subsection{Gauge fixing and anomalous symmetry}
\label{subsec:gauge-fixing-symmetry}

The previous subsection used the gauge choice $v=0$ to identify the physical space and hopping of the coupled system with those of the DW boundary. We now fix a flat gauge field $A\in Z^q(M,G)$ and describe the same physical space in the matter-field representation. Here $a=A_X+\dd v$. Gauge transformations preserving $A$ still act on the matter space and give the global symmetry.

First consider the operators on the full space $\ket A_M\otimes\mathcal H_{\mathrm{matter}}$. For $z\in Z^{q-1}(X,G)$, let $\widetilde z\in Z^{q-1}(M,G)$ be a closed bulk extension. Such an extension exists in our sphere--disk setting with $1\leq q\leq d$: for $q=1$, $z$ is constant; for $q>1$, $z$ is exact and its potential can be extended into $M$. The background-preserving $G(\widetilde z)$ and physical hopping act as
\begin{equation}
  \begin{aligned}
  G(\widetilde z)\bigl(\ket A_M\otimes\ket\psi_X\bigr)
    &=\ket A_M\otimes\mathcal S_A(z)\ket\psi_X,\\
  \mathcal U(s)\bigl(\ket A_M\otimes\ket\psi_X\bigr)
    &=\ket A_M\otimes U_A(s)\ket\psi_X.
  \end{aligned}
  \label{eq:gauss-residual-symmetry}
\end{equation}
Since $G,\mathcal{U}$ preserve the background gauge field, these equations reduce them to operators $\mathcal{S}_A,U_A$ acting only on the matter space $\mathcal{H}_{\mathrm{matter}}$, with formulas
\begin{align}
  \mathcal S_A(z)\ket v_X
  &:=\exp\left(2\pi\ii\left[
    \int_X\Theta^2(v,z;A_X)-\int_M\Theta^1(\widetilde z;A)
  \right]\right)\ket{v+z}_X,
  \label{eq:SPT anomalous symmetry}\\
  U_A(s)\ket v_X
  &:=\exp\left(2\pi\ii\int_X\Theta^2(s,v;A_X)\right)\ket{v+s}_X.
  \label{eq:hopping operator formula}
\end{align}
By Lemma~\ref{lem:boundary-gauss-action}, $\mathcal S_A(z)$ is independent of the closed extension. The strict group law of $G$ and Eq.~\eqref{eq:gauged-hopping-invariance} further give
\begin{equation}
  \mathcal S_A(z')\mathcal S_A(z)=\mathcal S_A(z'+z),
  \qquad U_A(s)\mathcal S_A(z)=\mathcal S_A(z)U_A(s).
  \label{eq:boundary-hopping-symmetry-inherited}
\end{equation}

Thus $\mathcal S_A$ defines a global symmetry. Appendix~\ref{subsec:spt-anomalous-boundary-symmetry} identifies it with the boundary symmetry of the SPT phase constructed from $[\omega]$, so its anomaly is $[\omega]\in H^{d+2}(B^qG,\RZ)$~\cite{ChenGuLiuWen2013,GaiottoEtAl2015GeneralizedSymmetries}. The operator $U_A$ is symmetric under this symmetry.

Comparing Eqs.~\eqref{eq:boundary-compensating-operator} and \eqref{eq:SPT anomalous symmetry}, for a closed parameter $z$ we have
\[
  \mathcal S_A(z)
  =\exp\left(-2\pi\ii\int_M\Theta^1(\widetilde z;A)\right)R(z;A_X).
\]
The two differ only by an overall phase independent of $v$. When $s$ is the restriction of $z$ to a patch, $\dd s=0$ away from its boundary. The operator $R(s;A_X)$ then has the same local action as the global symmetry inside the patch and acts trivially outside, away from the boundary. It therefore defines a symmetry patch operator.

Although the symmetry patch operator $R(s,A_X)$ and the hopping operator $U_A(s)$ both act on classical variables as $v\mapsto v+s$, their quantum phases use $\Theta^2(v,s;A_X-\dd s)$ and $\Theta^2(s,v;A_X)$, respectively. This distinction is fundamental: the former composes at the starting point of the path $A_X\xrightarrow{v} a$ and changes the background, while the latter composes at the endpoint and changes the configuration.

From this perspective, Eq.~\eqref{eq:boundary-hopping-symmetry-inherited} may be understood as associativity rather than commutativity; see Appendix~\ref{app:tetrahedron-coherence}. It can be written as
\begin{equation}\label{eq:bimodule}
  U(s)\bigg(\ket\psi\mathcal S(z)\bigg)
  =\bigg(U(s)\ket\psi\bigg)\mathcal S(z).
\end{equation}

Fixing the background $A$ does not yet impose all Gauss constraints; the remaining constraints require matter states to satisfy $\mathcal S_A(z)=1$. Thus the gauge-invariant subspace $\mathcal{H}_{\mathrm{gauged}}$ corresponds to the symmetric subspace $\mathcal H_{\mathrm{sym}}$, with projector
\begin{equation}
  \Pi_{\mathrm{sym}}
  :=\frac{1}{|Z^{q-1}(X,G)|}\sum_z\mathcal S_A(z),
  \qquad
  \mathcal H_{\mathrm{sym}}=\operatorname{im}\Pi_{\mathrm{sym}}.
  \label{eq:boundary-symmetry-projector}
\end{equation}
Again by Lemma~\ref{lem:boundary-gauss-action}, $\Pi_{\mathrm{sym}}\ket v_X$ has nonzero coefficients throughout the orbit $v+Z^{q-1}(X,G)$. Projected states with the same $a=A_X+\dd v$ differ only by a phase, whereas distinct $a$ label orthogonal state lines. The operator $U_A(s)$ commutes with the projector and sends these lines to those labeled by $a+\dd s$.

Conversely, adjoining the fixed gauge-field state $\ket A_M$ to a symmetric state and applying $\Pi_{\mathrm{gauged}}$ recovers the full gauge-invariant state. As with the choice $v=0$, gauge fixing and gauge averaging are inverse after suitable normalization. The two gauge choices give different representations of the same physical space and hopping:
\begin{equation}
  \begin{tikzcd}[column sep=3.6em,row sep=3.0em]
  \mathcal H_{\mathrm{sym}}\arrow[d,"U_A(s)"']
  &\mathcal H_{\mathrm{gauged}}
    \arrow[l,"\operatorname{Fix}_A"',"\cong"]
    \arrow[r,"\operatorname{Fix}_{v=0}","\cong"']
    \arrow[d,"\mathcal U(s)"]
  &\mathcal H_{\mathrm{DW},\partial}\arrow[d,"T_\partial(s)"]\\
  \mathcal H_{\mathrm{sym}}
  &\mathcal H_{\mathrm{gauged}}
    \arrow[l,"\operatorname{Fix}_A"]\arrow[r,"\operatorname{Fix}_{v=0}"']
  &\mathcal H_{\mathrm{DW},\partial}
  \end{tikzcd}
  \label{eq:gauging-hopping-diagram}
\end{equation}
The intertwining relation on the left follows directly from Eq.~\eqref{eq:gauss-residual-symmetry}; the one on the right was established in Eq.~\eqref{eq:gauging-vzero-intertwining}.

For $a=A_X+\dd v$, the state $\ket a_{\mathrm{gauged}}$ in Eq.~\eqref{eq:gauged-configuration-decomposition} corresponds here to $\Pi_{\mathrm{sym}}\ket v_X$. By the preceding subsection, $U_A$ therefore satisfies the topological-excitation axioms on $\mathcal H_{\mathrm{sym}}$ and has the same statistics $[\omega]$ as $T_\partial$. On the gauge-invariant space, $G$ acts trivially, while physical hopping still acts on configurations and their quantum phases. After fixing the background, this distinction becomes the residual symmetric condition and the action of $U_A$ preserving it. For the configuration states in this section, an operator $U(s)$ of the form in Eq.~\eqref{eq: ansatz} satisfies the configuration axiom if and only if it commutes with $\mathcal S_A$. Remark~\ref{rem:configuration-symmetry} proves the converse implication.

This correspondence is independent of the choice of fixed background $A$. The operators $U_A$ and $\mathcal S_A$ generally depend on the background, but gauge transformations relate them. Equation~\eqref{eq:gauged-hopping-invariance} and the group law of $G$ give the following identities on $\mathcal{H}_{\mathrm{matter}}$:
\begin{equation}
  \begin{aligned}
  U_{A'}(s)
    &=R(t_X;A_X)U_A(s)R(t_X;A_X)^{-1},\\
  \mathcal S_{A'}(z)
    &=R(t_X;A_X)\mathcal S_A(z)R(t_X;A_X)^{-1},
  \end{aligned}
  \qquad A'=A-\dd t.
  \label{eq:background-operator-conjugation}
\end{equation}

This covariance explains why the symmetry patch operator $R(t_X;A_X)$ generally commutes with neither hopping operators nor the global symmetry: it also changes the background, turning $U_A$ and $\mathcal S_A$ into the operators for the new background. Operators at different backgrounds are related by conjugation, so their statistics and anomaly agree; explicit calculations may therefore use $A=0$.

\begin{remark}
\label{rem:configuration-symmetry}
For the configuration states of this section and the cochain ansatz in Eq.~\eqref{eq: ansatz}, the configuration axiom constrains only the symmetric subspace but implies that $U(s)$ commutes with $\mathcal S_A$ on the entire matter space. The proof is as follows.

Suppose $U(s)$ has this form. The configuration axiom ensures that it and its inverse preserve $\mathcal H_{\mathrm{sym}}$. Since $\mathcal S_A(z)$ acts as the identity on this subspace,
\[
  [U(s),\mathcal S_A(z)]\Pi_{\mathrm{sym}}\ket v_X
  =\Pi_{\mathrm{sym}}\ket v_X.
\]
The operators $U(s)$ and $\mathcal S_A(z)$ shift basis states by $s$ and $z$, respectively, so their group commutator is diagonal. Since $\Pi_{\mathrm{sym}}\ket v_X$ has nonzero coefficients throughout the closed-shift orbit, the equation requires the commutator to act as the identity on every basis state in that orbit. As $v$ ranges over all cochains, these orbits cover the entire matter space. Hence $[U(s),\mathcal S_A(z)]=1$, so $U(s)$ is symmetric.

For hopping constructed by Eq.~\eqref{eq:hopping operator formula}, the same orbit argument also applies to statistical processes. A statistical process is invariant under independent constant rephasings of each elementary hopping operator, so the net number of forward minus inverse occurrences of every parameter is zero, and the operator product is diagonal in the $\ket v_X$ basis. The Initial-State Independence Theorem cited in Section~\ref{subsec:statistical-processes} (Ref.~\cite{Xue2026Statistics}, Theorem VI.4) ensures that this product gives the same phase on all configuration states. The orbit argument then shows that it equals this phase times the identity on the entire $\mathcal H_{\mathrm{matter}}$.
\end{remark}

\section{Anyon statistics and \texorpdfstring{$1$}{1}-form anomaly}
\label{sec:anyon-statistics-anomaly}
\label{subsec:anyon-legal-table}

Readers familiar with two-dimensional topological order may know that, when a family of closed operators serves both as closed anyon hopping and as a $1$-form symmetry, its statistics and anomaly are closely related~\cite{FengEtAl2025HigherFormAnomalies}. Sections~\ref{sec:typeiii-configuration} and~\ref{sec:relation} show that this phenomenon does not extend to general topological excitations. We must nevertheless explain why it appears so natural in the familiar two-dimensional examples. We do so here by examining the geometry of braiding.

Fix a $1$-form symmetry and compare its different truncations to the same open path $A$. Truncating the symmetry operator on a closed path to $A$ gives $U_A$. A second closed path $B$ crosses $A$ once, with both endpoints of $A$ far from $B$, as in Fig.~\ref{fig:closed-braiding-B}. Let $U_B$ denote the action of the given symmetry on $B$.

\begin{figure}[H]
 \centering
 \begin{tikzpicture}[x=1cm,y=1cm,>=latex,every node/.style={font=\small}]
  \fill[gray!12] (-2.1,0) circle (.35cm);
  \fill[gray!12] (1.35,0) circle (.35cm);
  \draw[gray,densely dotted] (-2.1,0) circle (.35cm);
  \draw[gray,densely dotted] (1.35,0) circle (.35cm);
  \draw[red!75!black,line width=1.2pt,->] (-2.1,0)--(1.35,0);
  \fill[red!75!black] (-2.1,0) circle (2pt);
  \fill[red!75!black] (1.35,0) circle (2pt);
  \node[red!75!black,above] at (-1.0,0) {$U_A$};
  \node[below] at (-2.1,-.38) {$V_L$};
  \node[below] at (1.35,-.38) {$V_R$};
  \draw[blue!70!black,line width=1.2pt,->] (0,-1.05)--(0,1.05);
  \draw[blue!70!black,line width=1.2pt,dashed,->] (0,1.05) arc[start angle=90,end angle=-3,x radius=2.3cm,y radius=1.05cm];
  \draw[blue!70!black,line width=1.2pt,dashed] ({2.3*cos(-3)},{1.05*sin(-3)}) arc[start angle=-3,end angle=-90,x radius=2.3cm,y radius=1.05cm];
  \node[blue!70!black,left] at (0,.6) {$B$};
  \node[blue!70!black,right] at (2.4,.1) {$U_B$};
 \end{tikzpicture}
 \caption{The closed path $B$ encircles the right endpoint of the open path $A$ and crosses $A$ once. The gray regions are the supports of the truncation corrections $V_L,V_R$, both disjoint from $B$. The dashed portion of $B$ does not intersect $A$ and therefore does not affect the braiding commutator.}
 \label{fig:closed-braiding-B}
\end{figure}

Another truncation $\widetilde U_A$ differs only near the endpoints of $A$, so we can write
\begin{equation}
 \widetilde U_A=V_LV_RU_A.
\end{equation}
Since the supports of $V_L,V_R$ are disjoint from that of $U_B$, these corrections do not change the commutator. The braiding statistics are therefore the same:
\begin{equation}
	[\widetilde U_A,U_B]=[U_A,U_B].
	\label{eq:anyon-truncation-braiding}
\end{equation}

This argument that anomaly determines statistics relies on the special geometry of braiding. One operator remains closed, while all truncation freedom of the other is confined to its two distant endpoints. The commutator measures braiding, so once the symmetry is fixed, truncation freedom cannot change this statistical phase. The argument fails in more general settings, such as Section~\ref{sec:relation}. This also illustrates what is special about anyon statistics: in the general theory of generalized statistics, the configuration axiom is essential to the definition of statistics and its classification by $H^{d+2}(B^qG,\RZ)$, whereas it seems not to play a central role in braiding.

We have shown that, for operators serving both as closed anyon hopping and as a $1$-form symmetry, the anomaly uniquely determines the statistics. The subtlety is that the two differ by a sign. Consider the braiding of two $\ZZ_N$ particle species.

We use the matter-field representation of Section~\ref{sec:what-is-hopping}, with $G=\ZZ_N\oplus\ZZ_N$ and $v=(v_A,v_B)\in C^1(X,G)$ on the sphere $X\simeq S^2$. The descent construction gives
\begin{equation}
 \omega(A_A,A_B)=\frac1N A_A\smile A_B,
 \qquad L(s,v)=\frac1N v_A\smile s_B.
 \label{eq:ordinary-mutual-descendant}
\end{equation}
The cup products are first evaluated in $\ZZ_N$ and then mapped to $\RZ$ by the coefficient homomorphism $\ZZ_N\to\RZ$, $1\mapsto1/N$. The phase $L(s,v)$ defines hopping with statistics $[\omega]$. We close this hopping and compare it with the canonical symmetry constructed from the same $\omega$:
\begin{align}
 U(z)\ket v
 &=\exp\left(\frac{2\pi\ii}{N}\int_Xv_A\smile z_B\right)\ket{v+z},
 \label{eq:ordinary-closed-hopping}\\
 \mathcal S_\omega(z)\ket v
 &=\exp\left(\frac{2\pi\ii}{N}\int_Xz_A\smile v_B\right)\ket{v+z}.
 \label{eq:ordinary-systematic-symmetry}
\end{align}
These formulas have almost the same form, differing only by the positions of the matter field and the closed parameter. On the sphere, closed parameters $z,w$ are exact, so $\int_Xw_A\smile z_B=0$. Hence both families of closed operators satisfy the group law. We already know that $\mathcal S_\omega$ has anomaly $[\omega]$; it remains to put $U(z)$ in the same form to compare their anomalies.

Apply the following local change of basis, defined triangle by triangle:
\begin{equation}
 \Phi\ket v=\exp\left(\frac{2\pi\ii}{N}\int_Xv_A\smile v_B\right)\ket v.
 \label{eq:ordinary-mutual-basis-change}
\end{equation}
Expanding the phase and using $\int_Xz_A\smile z_B=0$ gives
\begin{equation}
 \Phi^{-1}U(z)\Phi\ket v
 =\exp\left(-\frac{2\pi\ii}{N}\int_Xz_A\smile v_B\right)\ket{v+z}
 =\mathcal S_{-\omega}(z)\ket v.
 \label{eq:ordinary-mutual-opposite-anomaly}
\end{equation}
This change of basis exchanges $v_A\smile z_B$ for $z_A\smile v_B$ and introduces a minus sign. Since an FDQC preserves the anomaly, the closed operator $U(z)$ has anomaly $-[\omega]$.

The minus sign has physical content. In Eqs.~\eqref{eq:SPT anomalous symmetry} and~\eqref{eq:hopping operator formula}, the same bulk DW theory constructs commuting operators $\mathcal S_\omega$ and $U$, with the anomaly of the former and the statistics of the latter both conventionally set to $[\omega]$. When $2[\omega]\ne0$, the symmetry obtained by directly closing $U$ has anomaly $-[\omega]$ and is inequivalent to $\mathcal S_\omega$.

This opposite-sign relation extends to general Abelian anyons, including both braiding and half-braiding. For the hopping constructed in Section~\ref{sec:what-is-hopping}, one can choose closed hopping satisfying the group law and carrying $1$-form anomaly $-[\omega]$. Under the conditions in Appendix~\ref{subsec:anyon-lifts}, this anomaly is also independent of the choice of valid assembly. The precise statement and proof are given there.

This result parallels the inverse-square relation in Theorem~\ref{thm:typeiii-inverse-square}: when the same family of operators serves as both hopping and symmetry, a fixed coefficient relates anomaly and statistics, namely $-2$ for one-dimensional type-III and $-1$ for two-dimensional anyons. However, we do not yet have tools as direct as the one-dimensional Else--Nayak index and its endpoint relations, so the argument for general anyons is not yet as rigorous as its one-dimensional counterpart.

\section{Equality of statistics and anomaly versus the symmetric condition}
\label{sec:statistical-obstruction}
\label{sec:double-role-operators}

In Section~\ref{sec:relation}, we saw that, for topological excitations with statistics $[\omega]$, even when valid closed hopping operators $U(z)$ satisfy the group law and thus define a symmetry, their anomaly $[\omega']$ generally differs from $[\omega]$. For the one-dimensional $\ZZ_2$ point excitations considered there, an additional symmetric condition repairs this mismatch: requiring every closed hopping operator to commute with every elementary hopping operator ensures $[\omega']=[\omega]$, namely
\begin{equation}
	U(s)U(z)=U(z)U(s),\qquad \dd z=0.
	\label{eq:symmetric-extended-hopping}
\end{equation}
For a fixed closed action $U(z)$, this condition is invariant under reselection of the elementary hopping operators.

This repair ensures $[\omega']=[\omega]$ because, as shown in Section~\ref{sec:what-is-hopping}, the correct relation between hopping $U(s)$ and symmetry $\mathcal S(z)$ is $U(s)\mathcal S(z)=\mathcal S(z)U(s)$. Equation~\eqref{eq:symmetric-extended-hopping} requires the closed hopping $U(z)$ to play the role of $\mathcal S(z)$ in that construction. This requirement cannot be met for general topological excitations, and we now analyze the obstructions. Even in many cases where $[\omega']=[\omega]$, Eq.~\eqref{eq:symmetric-extended-hopping} is impossible. The obstructions discussed here arise from the statistics $[\omega]$ alone: the arguments neither assume a group law for $U(z)$ nor depend on $[\omega']$.

First consider two-dimensional anyons. As in Section~\ref{subsec:statistical-processes}, represent $[\omega]\in H^4(B^2G,\RZ)$ by a quadratic function $Q$, with $Q(ng)=n^2Q(g)$. Under the uniqueness hypotheses of Theorem~\ref{thm:anyon-closed-rigidity}, the closed anomaly is $-Q$. Equality of statistics and anomaly therefore requires $Q=-Q$, or $2Q=0$.

The symmetric condition, however, requires trivial braiding statistics, which is stronger than $2Q=0$. Use the geometry of Fig.~\ref{fig:closed-braiding-B}, with $U_B$ taken from the given family of closed operators. Equation~\eqref{eq:symmetric-extended-hopping} requires $U_A$ to commute with this closed $U_B$, so
\begin{equation}
 1=U_A^{-1}U_B^{-1}U_AU_B
 =\exp\!\left(-2\pi\ii b_Q(g,h)\right),
 \label{eq:closed-braiding-obstruction}
\end{equation}
Here $b_Q(g,h)=Q(g+h)-Q(g)-Q(h)$, so Eq.~\eqref{eq:closed-braiding-obstruction} requires $b_Q=0$. Quadraticity gives
\begin{equation*}
 b_Q(g,g)=Q(2g)-2Q(g)=2Q(g),
\end{equation*}
Thus $b_Q=0$ implies $2Q=0$, but the converse fails. For example, the toric-code anyons $1,e,m,f$ satisfy
\begin{equation*}
 Q(1)=Q(e)=Q(m)=0,\qquad Q(f)=\frac12,
\end{equation*}
Hence $2Q=0$, but $e$ and $m$ fuse to $f$, giving $b_Q(e,m)=Q(f)-Q(e)-Q(m)=\frac12\ne0$.

Therefore $b_Q=0$ is necessary and sufficient for the symmetric condition to be realizable (a construction proving sufficiency appears in Appendix~\ref{app:anyon-symmetric-sufficiency}), whereas equality of statistics and anomaly requires only $2Q=0$.

For one-dimensional point excitations, the statistical process of Section~\ref{subsec:typeiii-statistical-process} gives a similar constraint. Let $G$ be any finite Abelian fusion group.

If Eq.~\eqref{eq:symmetric-extended-hopping} holds, the closed operator $U_h(S^1)$ commutes with every elementary hopping operator and hence with the assembled $U_g(P)$. The inner commutator in Eq.~\eqref{eq:three-point-theta} is then $1$. Lemma~\ref{lem:typeiii-statistical-process} gives
\begin{equation}
 \eta(g,h,k)=\Alt\widehat\omega(g,h,k)=1,
 \qquad g,h,k\in G.
 \label{eq:typeiii-symmetric-statistics-obstruction}
\end{equation}
Nontrivial type-III statistics therefore prevent all hopping operators from commuting with their own closed actions. In the setting of Theorem~\ref{thm:typeiii-inverse-square}, equality of the type-III parts of statistics and anomaly requires only $\eta(g,h,k)^3=1$, whereas the symmetric condition still requires $\eta(g,h,k)=1$.

Next consider $\ZZ_2^2$ loop statistics on the two-dimensional sphere $S^2$~\cite{Xue2026Statistics}. Both statistics and anomaly are classified by $H^4(B\ZZ_2^2,\RZ)=\ZZ_2^2$. Partition $S^2$ into four patches $R_1,R_2,R_3,R_4$ with the incidence pattern of the faces of a tetrahedron. Write $U_A(R),U_B(R)$ for the hopping operators of the two loop species, where $R\in\{R_1,R_2,R_3,R_4,S^2\}$.

The two independent loop statistics are detected by the following nested commutators:
\begin{align}
	&\left[U_A(R_3),
	\left[U_A(R_2),
	\left[U_A(R_1),U_B(S^2)\right]\right]\right]=\pm1,\notag\\
	&\left[U_B(R_3),
	\left[U_B(R_2),
	\left[U_B(R_1),U_A(S^2)\right]\right]\right]=\pm1.
	\label{eq:2d-loop-invariants}
\end{align}

If Eq.~\eqref{eq:symmetric-extended-hopping} holds, the innermost commutator in each expression is already the identity, so both statistical phases must be $+1$. Thus no hopping family with nontrivial loop statistics can satisfy Eq.~\eqref{eq:symmetric-extended-hopping}.

Appendix~\ref{app:loop-full-pattern} shows that, for every loop statistics class $[\omega]$, closed hopping satisfying the group law exists, and its anomaly $[\omega']$ can be chosen arbitrarily in $H^4(B\ZZ_2^2,\RZ)$, as in Section~\ref{sec:relation}. On the other hand, when the statistics are nonzero, solutions with $[\omega]=[\omega']$ exist but solutions of Eq.~\eqref{eq:symmetric-extended-hopping} do not, in contrast to that example.

\section{Conclusion}
\label{sec:conclusion}

Starting from the respective definitions of symmetry and topological excitations, we have compared anomaly and statistics and constructed their connection through gauging. Coupling boundary matter to a DW gauge field in one higher dimension leaves gauge-invariant degrees of freedom described by $a=A_X+\dd v$, and physical hopping changes this configuration. After fixing the background gauge field, the Gauss transformations that preserve it become the global symmetry, the remaining Gauss conditions select the symmetric subspace, and physical hopping commutes with this symmetry. The same DW twist thus appears as hopping statistics and as symmetry anomaly.

We then examined the problems with identifying symmetry patches and hopping operators. For finite Abelian groups in one dimension, Section~\ref{subsec:typeiii-statistical-process} constructs a type-III statistical process, and Appendix~\ref{app:typeiii-statistical-evaluation} proves that its phase $\eta$ is alternating. Combining this with the endpoint exchange relation in Theorem~\ref{thm:endpoint-exchange}, we obtain $\Alt\nu=\eta^{-2}$ for compatible microscopic truncations. A type-III part that is not divisible by $2$ is therefore incompatible with such a hopping realization, including the $\ZZ_2^3$ anomaly $\frac12a_1a_2a_3$. Consequently, the local-product construction in Section~V of Ref.~\cite{Kobayashi_2026} does not automatically guarantee the configuration axiom.

A second question is whether statistics determine the anomaly when closed hopping itself defines a symmetry. The one-dimensional $\ZZ_2$ and two-dimensional loop examples show that these cohomology classes can vary independently. In the state families and local representations considered in Section~\ref{sec:anyon-statistics-anomaly}, two-dimensional anyons instead exhibit a special correspondence: valid closed hopping satisfying the group law has anomaly equal to the negative of its statistics.

A third question is whether equality of anomaly and statistics implies the symmetric condition. Any two-dimensional loop statistics can coexist with the same anomaly, but nonzero statistics prevent hopping from commuting with its own closed actions. Two-dimensional anyons likewise show that, even when the classes coincide because they have order two, nontrivial mutual braiding prevents the symmetric condition. Equality of the two cohomology classes is therefore insufficient to ensure the corresponding operator commutation relations.

These results show that identifying symmetry patches with hopping operators is not the correct way to relate symmetry and topological excitations.

\section*{Acknowledgments}
\addcontentsline{toc}{section}{Acknowledgments}

We thank Yitao Feng, Shenghan Jiang, Bowen Shi, Xiao-Gang Wen, Yunlong Zang, and Carolyn Zhang for helpful discussions.

\appendix
\numberwithin{equation}{section}

\section{Type-III statistics and anomaly in one dimension}
\label{app:typeiii-statistics-anomaly}

This appendix first characterizes the type-III part through the self-statistics of the seven particle species of $\ZZ_2^3$, then evaluates the phase of the statistical process in Subsection~\ref{subsec:typeiii-statistical-process}, and finally computes a type-III anomaly and its endpoint factors in a lattice model.

\subsection{The Fano plane and type-III statistics}
\label{app:typeiii-fano-plane}

For $V=\ZZ_2^3$, the type-III invariant can be expressed through the self-statistics of seven particle species.
We use the coordinate $1$-cocycles $a_1,a_2,a_3$ from Subsection~\ref{subsec:patch-D}
and the seven generators of $H^3(BV,\RZ)\simeq\ZZ_2^7$ listed there.

Associate each element of $V^\times:=V\setminus\{0\}$ with a nontrivial point excitation whose self-statistics phase lies in $\{\pm1\}$. Divide the circle $S^1$ into three intervals $I_1,I_2,I_3$, and let $U_g(I_1)$ denote the hopping operator for a particle of type $g$ on $I_1$. The process in Eq.~\eqref{eq:one-dimensional-z2-statistics} gives
\begin{equation}
 \lambda(g)=[U_g(I_1)^2,U_g(I_2)]\in\{\pm1\}.
 \label{eq:typeiii-species-self-statistics}
\end{equation}

\begin{figure}[H]
	\centering
	\begin{tikzpicture}[x=1.15cm,y=1.15cm,every node/.style={font=\small}]
		\coordinate (a) at (0,1.8);
		\coordinate (b) at (-1.56,-0.9);
		\coordinate (c) at (1.56,-0.9);
		\coordinate (ab) at (-0.78,0.45);
		\coordinate (ac) at (0.78,0.45);
		\coordinate (bc) at (0,-0.9);
		\coordinate (o) at (0,0);
		\draw[thick,linkblue] (a)--(b)--(c)--cycle;
		\draw[thick,linkblue] (a)--(bc);
		\draw[thick,linkblue] (b)--(ac);
		\draw[thick,linkblue] (c)--(ab);
		\draw[thick,linkblue] (o) circle (0.9);
		\foreach \p in {a,b,c,ab,ac,bc,o}{\fill[red!75!black] (\p) circle (2.1pt);}
		\node[above] at (a) {$100$};
		\node[below left] at (b) {$010$};
		\node[below right] at (c) {$001$};
		\node[left] at (ab) {$110$};
		\node[right] at (ac) {$101$};
		\node[below] at (bc) {$011$};
		\node[above right,fill=white,inner sep=1pt,xshift=2pt,yshift=1pt] at (o) {$111$};
	\end{tikzpicture}
\caption{The geometry of $V^\times$ is the Fano plane. Each point represents a nontrivial excitation whose self-statistics invariant can independently take the value $\pm1$. The particles at the three points on each line fuse to zero. A change of basis of $V$ simply permutes these seven points.}
	\label{fig:fano-statistical-process}
\end{figure}

\begin{theorem}\label{thm: type-III}
Point excitations with fusion group $V$ have nontrivial type-III statistics if and only if an odd number of elements $g\in V^\times$ satisfy $\lambda(g)=-1$. Equivalently, this occurs if and only if
	\begin{equation}
		\prod_{g\in V^\times}\lambda(g)=-1.
		\label{eq:fano-phase-product}
	\end{equation}
\end{theorem}
\begin{proof}
We compute these self-statistics by restriction to subgroups. For each nontrivial $g\in V^\times$, restrict the statistical class to $\{0,g\}\simeq\ZZ_2$; the resulting class corresponds to $\lambda(g)$. For $g=(x_1,x_2,x_3)\in\ZZ_2^3$, this replaces the $1$-cocycle $a_i$ in the seven generators of Subsection~\ref{subsec:patch-D} by $x_i a$, where $a$ is the canonical $1$-cocycle of $B\ZZ_2$. As $g$ runs over the seven nonzero vectors, $x_i$, $x_ix_j$ ($i\ne j$), and $x_1x_2x_3$ take the value $1$ respectively $4$, $2$, and $1$ times. Thus the sum of the seven restricted classes retains only the type-III coefficient, and the product of their statistical phases gives Eq.~\eqref{eq:fano-phase-product}.
\end{proof}

By Lemma~\ref{lem:typeiii-statistical-process} in the next subsection, for any basis $(g_1,g_2,g_3)$, Eqs.~\eqref{eq:three-point-theta} and~\eqref{eq:fano-phase-product} give the same phase: $\eta(g_1,g_2,g_3)=\prod_{g\in V^\times}\lambda(g)$. The former probes the type-III component through an interleaved process involving three independent particle species, while the latter characterizes it through a permutation-invariant combination of the self-statistics of all seven species.

\subsection{The type-III statistical phase}
\label{app:typeiii-statistical-evaluation}

We prove that the phase of the statistical process in Subsection~\ref{subsec:typeiii-statistical-process}
equals the type-III part of the statistics, using the notation introduced there.

\begin{lemma}
 \label{lem:typeiii-statistical-process}
Let the statistics of the above topological excitations be $[\omega]\in H^3(BG,\RZ)$,
and choose a normalized cocycle representative. Write
$\widehat\omega(g,h,k):=\exp(2\pi\ii\omega(g,h,k))$ for the corresponding multiplicative cocycle.
For any $g,h,k\in G$, Eq.~\eqref{eq:three-point-theta} is a statistical process, and
 \begin{equation}
  \mathcal P(g,h,k)\big|_{\mathcal H_{\mathrm{conf}}}
  =\eta(g,h,k)\one,
  \qquad \eta(g,h,k)=\Alt\widehat\omega(g,h,k).
  \label{eq:typeiii-statistical-phase}
 \end{equation}
Thus $\eta$ is the type-III part of $[\omega]$, and is therefore an alternating function. In particular,
 \begin{equation}
  \eta(g,k,h)=\eta(g,h,k)^{-1}.
  \label{eq:typeiii-alternating-character}
 \end{equation}
\end{lemma}

\begin{proof}
By the classification of statistical classes and independence of the initial state and the choice of operators, we evaluate the process in
the DW boundary model of Subsection~\ref{subsec:boundary-excitation-statistics},
taking $U(s)=T_\partial(s)$. Equation~\eqref{eq:boundary-hopping-commutator}
gives the logarithmic phase of the inner commutator as
\begin{equation}
 \int_{S^1}\bigl[\Theta^2(g\chi_P,h;a)-\Theta^2(h,g\chi_P;a)\bigr]\in\RZ.
 \label{eq:typeiii-inner-descent-phase}
\end{equation}
Here $h$ is a constant $0$-cochain.
The outer commutator with $Q(k)$ takes the difference of Eq.~\eqref{eq:typeiii-inner-descent-phase}
between $a$ and $a+\dd(k\chi_Q)$.

Choose the initial configuration $a=-\dd(k\chi_Q)$, so that the outer hopping leads to the zero configuration.
The two configurations differ only at the endpoints of $Q$. At $x_0$, which lies outside $P$, we have $g\chi_P=0$,
so the descent phase vanishes by normalization. The configurations at both endpoints of $P$ are unchanged, so their contributions cancel.
Only a positively oriented edge $e$ at $x_2$ remains, where $g\chi_P=g$ and the configuration changes from $k$ to $0$.

By the prism formula in Appendix~\ref{app:descent-construction},
Eq.~\eqref{eq:theta-k-evaluation}, the three tetrahedra of $e\times\Delta^2$ give
\begin{equation}
 \int_e\Theta^2(g,h;a)
 =-\omega(k,h,g)+\omega(h,k,g)-\omega(h,g,k),
 \qquad a|_e=k.
 \label{eq:typeiii-one-edge-descent}
\end{equation}
The three staircase tetrahedra have orientations $+,-,+$, and the overall coefficient of the second descent is $-1$.
Subtracting the expression with $g,h$ exchanged, and using the vanishing contribution of the final zero configuration on this edge, we obtain
\begin{equation}
 \begin{aligned}
 \eta(g,h,k)
 &=\exp\!\bigl(2\pi\ii[\omega(g,h,k)+\omega(h,k,g)+\omega(k,g,h)\\[-.25em]
 &\hspace{34mm}-\omega(g,k,h)-\omega(k,h,g)-\omega(h,g,k)]\bigr)\\
 &=\Alt\widehat\omega(g,h,k).
 \end{aligned}
 \label{eq:typeiii-descent-evaluation}
\end{equation}
Exchanging two variables swaps the numerator and denominator of the six-term ratio, giving antisymmetry.
The remaining properties follow from the three relations satisfied by $\Alt$ in Subsection~\ref{subsec:patch-D}.
\end{proof}

\subsection{A lattice example of an anomalous symmetry}
\label{app:typeiii-constructions}

We use the $\ZZ_2^3$ model of Ref.~\cite{LessaChengWang}, Eq.~(41), compute $P,\Omega,C,D$ in the order of Section~\ref{sec:else-nayak}, and directly verify Theorem~\ref{thm:endpoint-exchange}. We then test the symmetric condition and the configuration axiom for these explicit operators.

On a periodic qubit chain of length $2N$, define
\begin{equation}
 \mathsf O=\prod_{j\,\mathrm{odd}}X_j,
 \qquad \mathsf E=\prod_{j\,\mathrm{even}}X_j,
 \qquad \mathsf C=\prod_{j=1}^{2N}\mathrm{CZ}_{j,j+1}.
 \label{eq:lcw-global-generators}
\end{equation}
Site indices are identified modulo $2N$; in particular, $0\equiv2N$. The conjugation formula
\begin{equation}
 \mathsf C X_j\mathsf C^{-1}=Z_{j-1}X_jZ_{j+1}
 \label{eq:lcw-C-conjugates-X}
\end{equation}
makes every $Z$ appear twice in the product over odd or even sites. The three generators therefore commute pairwise and square to $1$. The global symmetry is
\begin{equation}
 \mathcal S(g)=\mathsf O^{g_1}\mathsf E^{g_2}\mathsf C^{g_3},
 \qquad g=(g_1,g_2,g_3)\in\ZZ_2^3.
 \label{eq:lcw-global-symmetry}
\end{equation}

Choose a long interval $I=[1,2n+1]$ that does not cross the periodic cut, and define
\begin{equation}
 \mathsf O_I=\prod_{j=1}^{n}X_{2j-1},\qquad
 \mathsf E_I=\prod_{j=1}^{n}X_{2j},\qquad
 \mathsf C_I=\prod_{j=1}^{2n}\mathrm{CZ}_{j,j+1},
 \qquad P(g)=\mathsf O_I^{g_1}\mathsf E_I^{g_2}\mathsf C_I^{g_3}.
 \label{eq:lcw-literal-patch}
\end{equation}
The operator $\mathsf O_I$ includes the odd site at the left endpoint but excludes the odd site at the right endpoint. Concatenating these intervals around the circle recovers $\mathsf O,\mathsf E,\mathsf C$, respectively. The interior $Z$ factors cancel in pairs, leaving
\begin{equation}
 \begin{aligned}
 \mathsf C_I\mathsf O_I\mathsf C_I^{-1}&=Z_{2n}\mathsf O_I,
 &\mathsf C_I\mathsf E_I\mathsf C_I^{-1}&=Z_1Z_{2n+1}\mathsf E_I,\\
 \mathsf C\mathsf O_I\mathsf C^{-1}&=Z_0Z_{2n}\mathsf O_I,
 &\mathsf C\mathsf E_I\mathsf C^{-1}&=Z_1Z_{2n+1}\mathsf E_I.
 \end{aligned}
 \label{eq:lcw-patch-conjugations}
\end{equation}

We first compute $P(g)P(h)P(g+h)^{-1}$. Assigning the endpoint factors as in Eq.~\eqref{eq:patch-fusion}, we may take
\begin{equation}
 \begin{aligned}
 \Omega_L(g,h)&=Z_1^{g_3h_2},\\
 \Omega_R(g,h)&=(-1)^{g_3(h_1h_2+g_2h_1+g_1h_2)}
 Z_{2n}^{g_3h_1}Z_{2n+1}^{g_3h_2}.
 \end{aligned}
 \label{eq:lcw-omega-factors}
\end{equation}
All exponents are evaluated in $\ZZ_2$. Since $P(g)Z_1P(g)^{-1}=(-1)^{g_1}Z_1$, the left-endpoint associativity relation~\eqref{eq:EN-left} gives
\begin{equation}
 \nu(g,h,k)=(-1)^{g_1h_3k_2}.
 \label{eq:lcw-explicit-nu}
\end{equation}
This is a nontrivial type-III cocycle. For the exchange factors of the same patch, Eq.~\eqref{eq:C-definition} gives
\[
 C_L(g,h)=Z_1^{g_3h_2+h_3g_2},\qquad
 C_R(g,h)=\Omega_R(g,h)\Omega_R(h,g)^{-1}.
\]
In particular, the full commutators of the three generators are
\begin{equation}
 [\mathsf O_I,\mathsf E_I]=1,\qquad
 [\mathsf O_I,\mathsf C_I]=Z_{2n},\qquad
 [\mathsf E_I,\mathsf C_I]=Z_1Z_{2n+1}.
 \label{eq:lcw-patch-commutators}
\end{equation}

Next, compute $P(g)\Sg(h)P(g)^{-1}\Sg(h)^{-1}$. One choice of endpoint phases compatible with Eq.~\eqref{eq:product} is
\begin{equation}
 \begin{aligned}
 D_L(g,h)&=(-1)^{g_1h_2h_3}
 Z_0^{g_1h_3}Z_1^{g_2h_3+g_3h_2},\\
 D_R(g,h)&=(-1)^{g_1(g_2h_3+g_3h_2+h_2h_3)}
 Z_{2n}^{g_1h_3}Z_{2n+1}^{g_2h_3+g_3h_2}.
 \end{aligned}
 \label{eq:lcw-D-factors}
\end{equation}
Conjugation by $\Sg(h)$ multiplies $Z_0,Z_1$ by $(-1)^{h_2},(-1)^{h_1}$, respectively. Multiplication then gives
\begin{equation}
 \begin{aligned}
 D_L(g,h)\,\conj{\Sg(h)}{D_L(g,k)}
 &=(-1)^{g_1h_3k_2+h_1k_3g_2+h_1g_3k_2}D_L(g,h+k)\\
 &=\frac{\nu(g,h,k)\nu(h,k,g)}{\nu(h,g,k)}D_L(g,h+k).
 \end{aligned}
 \label{eq:lcw-D-multiplication}
\end{equation}
This is precisely Eq.~\eqref{eq:product}. Exchanging $h,k$ yields Eq.~\eqref{eq:exchange} of Theorem~\ref{thm:endpoint-exchange}, with phase
\begin{equation}
 \Alt\nu(g,h,k)
 =(-1)^{g_1h_3k_2+h_1k_3g_2+k_1g_3h_2
              +g_1k_3h_2+k_1h_3g_2+h_1g_3k_2}.
 \label{eq:lcw-explicit-alt}
\end{equation}
The simplest check uses the three generators $g_O=(1,0,0)$, $g_E=(0,1,0)$, and $g_C=(0,0,1)$. In this case,
\[
 D_L(g_O,g_E)=1,\qquad D_L(g_O,g_C)=Z_0,\qquad
 \mathsf E Z_0\mathsf E^{-1}=-Z_0,
\]
so the two sides of Eq.~\eqref{eq:exchange} indeed differ by $\Alt\nu(g_O,g_E,g_C)=-1$.

We now test the symmetric condition. Equation~\eqref{eq:lcw-patch-conjugations} directly gives
\[
 \mathsf O_I\mathsf C=Z_0Z_{2n}\mathsf C\mathsf O_I,\qquad
 \mathsf E_I\mathsf C=Z_1Z_{2n+1}\mathsf C\mathsf E_I,\qquad
 \mathsf C_I\mathsf E=Z_1Z_{2n+1}\mathsf E\mathsf C_I.
\]
Thus none of these three types of patches commutes with all global symmetries: they are not symmetric.

The contradiction with the configuration axiom is equally explicit. Write
\begin{equation}
 W_I=[\mathsf O_I,\mathsf C_I]=Z_{2n}.
 \label{eq:typeiii-zero-configuration-word}
\end{equation}
If these patches were admissible hopping operators, $W_I$ would have zero net configuration change. The operator $\mathsf E$, obtained by concatenating all $\mathsf E_I$ into a closed operator, would also preserve the configuration. The configuration axiom would therefore require both operators to act by phases on every configuration state line. However, the actual operators satisfy
\begin{equation}
 \mathsf E W_I=-W_I\mathsf E.
 \label{eq:lcw-diagonal-anticommutation}
\end{equation}
Two anticommuting invertible operators have no common eigenvector, so no nonzero configuration state can satisfy both requirements. Attaching the independent Abelian gauge-field shifts described in Ref.~\cite{Kobayashi_2026} does not remove this contradiction: the shifts cancel in the commutator, leaving $Z_{2n}\otimes1$.

\clearpage
\section{Algebraic-topological constructions in Dijkgraaf--Witten theory}
\label{app:tetrahedron-coherence}

We first construct the descent phases of Subsection~\ref{subsec:dw-gauge-transformations} from $\omega$. We then compare the boundary phases to show how the group law of $\mathcal S$ and the commutativity of $U$ with $\mathcal S$ arise by assembling the four faces of a tetrahedron.

\subsection{Construction of the descent phases}
\label{app:descent-construction}

Starting from $A\in Z^q(N,G)$, consider the composition of $k$ successive gauge transformations $t_i\in C^{q-1}(N,G)$:
\begin{equation}
  A\xrightarrow{t_1} A+\dd t_1\xrightarrow{t_2}\ldots\xrightarrow{t_k}A+\dd(t_1+\cdots+t_k).
  \label{eq:gauge-simplex-vertices}
\end{equation}
Geometrically, place these data on a simplex $\Delta^k$, labeling vertex $i$ by $A+\dd(t_1+\cdots+t_i)$ and edge $[i-1,i]$ by $t_i$. Since $A,t_i$ are themselves cochains on $N$, we construct a cocycle on $N\times \Delta^k$. Let
$\iota_i\in C^0(\Delta^k,\ZZ)$ take the value $0$ at vertices $0,\ldots,i-1$ and $1$ at vertices
$i,\ldots,k$. The integer-valued $\iota_i$ acts by scalar multiplication on the $G$-valued cochain
$t_i$. Define
\begin{equation}
  \operatorname{N}(t_k,\cdots,t_1;A)
  :=A+
  \dd
  \left(
    \sum_{i=1}^k \iota_i\smile t_i
  \right)
  \in Z^q(N\times\Delta^k,G).
  \label{eq:interpolating-gauge-field}
\end{equation}
Here the cochains $A,t_i$ on $N$ and $\iota_i$ on $\Delta^k$ are implicitly pulled back to $N\times \Delta^k$. The restriction of $\operatorname{N}(t_k,\ldots,t_1;A)$ to vertex $j$ of $\Delta^k$ is readily seen to be
$A+\dd(t_1+\cdots+t_j)$. This construction, called the nerve of $(t_k,\cdots,t_1;A)$, geometrically records the entire history of the gauge transformations. If
\(
\delta_j\operatorname{N}(t_k,\cdots,t_1;A)
\in Z^q(N\times\Delta^{k-1},G)
\)
denotes its restriction to $N\times \partial_j\Delta^k$, then
\begin{equation}
\delta_j\operatorname{N}(t_k,\ldots,t_1;A)
=
\begin{cases}
\operatorname{N}(t_k,\ldots,t_2;A+\dd t_1),
  &j=0,\\
\operatorname{N}(t_k,\ldots,t_{j+2},t_j+t_{j+1},t_{j-1},\ldots,t_1;A),
  &1\leq j\leq k-1,\\
\operatorname{N}(t_{k-1},\ldots,t_1;A),
  &j=k.
\end{cases}
\label{eq:gauge-nerve-faces}
\end{equation}

Applying the twist to this field gives
\begin{equation}
\omega\big(\operatorname{N}(t_k,\cdots,t_1;A)\big)\in Z^D(N\times \Delta^k,\RZ).
  \label{eq:twist-on-gauge-groupoid}
\end{equation}
This cocycle encodes all information about the gauge transformation of $\omega(A)$. Its $k$th level of coherence data is obtained by prism integration along $\Delta^k$.
Using the product orientation, write
\begin{equation}
  \varepsilon_k:=(-1)^{\,kD-k(k+1)/2}.
\end{equation}
Integration over the parameter simplex defines
\begin{equation}
  \Theta^k(t_k,\ldots,t_1;A)
  :=\varepsilon_k\int_{\Delta^k}\omega\big(\operatorname{N}(t_k,\cdots,t_1;A)\big)
  \in C^{D-k}(N,\RZ),
  \qquad 0\leq k\leq D.
  \label{eq:theta-k-definition}
\end{equation}
That is, for any $(D-k)$-simplex $\sigma$ in $N$,
\begin{equation}
  \int_\sigma\Theta^k(t_k,\ldots,t_1;A)
  =\varepsilon_k
  \int_{\sigma\times\Delta^k}\omega\big(\operatorname{N}(t_k,\cdots,t_1;A)\big).
  \label{eq:theta-k-evaluation}
\end{equation}
In particular, $\Theta^0(A)=\omega(A)$.
Stokes' theorem gives the full hierarchy of descent relations
\begin{equation}
  \dd\Theta^k=\delta\Theta^{k-1},
  \qquad 1\leq k\leq D,
  \label{eq:theta-descent-hierarchy}
\end{equation}
where $\delta=\sum_{i=0}^k(-1)^i\delta_i$ is the alternating sum over the faces of the gauge simplex. Equations~\eqref{eq:theta-1-descent}, \eqref{eq:theta-2-descent}, and~\eqref{eq:theta-3-descent} in the main text correspond to $k=1,2,3$, respectively.

In general, $\Theta^k$ records the interior phase of the $k$-simplex spanned by $k$ composable gauge transformations, and Eq.~\eqref{eq:theta-descent-hierarchy} ensures that the phases on its faces fit together. In the next two subsections we take $M\simeq D^{d+1}$ and $X=\partial M$, and label the left, top, right, and bottom vertices of each tetrahedron by $0,1,2,3$, respectively.

\subsection{The symmetry group law}
\label{app:symmetry-group-law}

For $z,z'\in Z^{q-1}(X,G)$, denote their closed bulk extensions by $\widetilde z,\widetilde z'$, and use $\widetilde z+\widetilde z'$ as the closed extension of $z+z'$. A single symmetry action corresponds to a closed gauge path. It changes the gauge path $A\xrightarrow{r}A+\dd r$ to $A\xrightarrow{r+\widetilde z}A+\dd r$, corresponding to right composition $r\mapsto r\circ\widetilde z$:
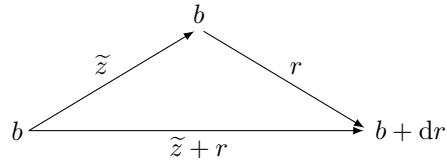
\[
\gaugetriangle{A}{A}{A+\dd r}
  {\widetilde z}{r}{\widetilde z+r}
\]
The phase in Eq.~\eqref{eq:SPT anomalous symmetry} is the sum of the reversed phase of the path $A\xrightarrow{\widetilde z}A$ and the integral of the triangle phase over $X$. When comparing $\mathcal S(z')\mathcal S(z)\ket v_X$ and $\mathcal S(z'+z)\ket v_X$, the state labels agree, while the phase difference involves three such triangles and three bulk gauge paths. They form the following tetrahedron, with $r_X=v$:
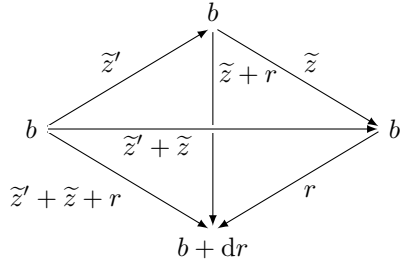
\[
\begin{gaugetetrahedron}{A}{A}{A}{A+\dd r}
  \gaugetetrahedronedges{\widetilde z'}{\widetilde z}{r}
    {\widetilde z'+\widetilde z}{\widetilde z+r}{\widetilde z'+\widetilde z+r}
\end{gaugetetrahedron}
\]
First, the three $\Theta^1$ terms correspond to the three edges of face $[012]$. By Eq.~\eqref{eq:theta-2-descent},
\begin{equation}
  \int_M\!\left[
    \Theta^1(\widetilde z';A)+\Theta^1(\widetilde z;A)
    -\Theta^1(\widetilde z'+\widetilde z;A)\right]
  =\int_X\Theta^2(z,z';A_X).
\end{equation}
The phase difference between the two operator actions, divided by $2\pi$, is therefore precisely the oriented sum of the other three faces and face $[012]$:
\begin{align}
  &\int_X\!\left[
    \Theta^2(v,z;A_X)-\Theta^2(v,z'+z;A_X)
    +\Theta^2(v+z,z';A_X)-\Theta^2(z,z';A_X)
  \right]\notag\\
  &\hspace{25mm}=\int_X\dd\Theta^3(v,z,z';A_X)=0.
\end{align}
The last equality uses $\partial X=\varnothing$. Hence
\begin{equation}
  \mathcal S(z')\mathcal S(z)=\mathcal S(z'+z).
\end{equation}
This also shows that the bulk phase $-\int_M\Theta^1(\widetilde z;A)$ in $\mathcal S$, which is independent of $v$, supplies the fourth face, allowing the three boundary phases to form a closed tetrahedron.

\clearpage
\subsection{Commutativity of hopping and symmetry}
\label{app:hopping-symmetry-coherence}

For $s\in C^{q-1}(X,G)$ and $z\in Z^{q-1}(X,G)$, $U(s)$ corresponds to left composition of gauge paths, while $\mathcal S(z)$ corresponds to right composition. The two orders are related by compatibility of composition:
\begin{equation}
  (s\circ v)\circ z\simeq s\circ(v\circ z).
\end{equation}
More explicitly, both sides of Eq.~\eqref{eq:bimodule} are phase multiples of $\ket{v+s+z}_X$. The bulk phases of $\mathcal S$ cancel in the comparison, and the remaining boundary phases form the four faces of the following tetrahedron:
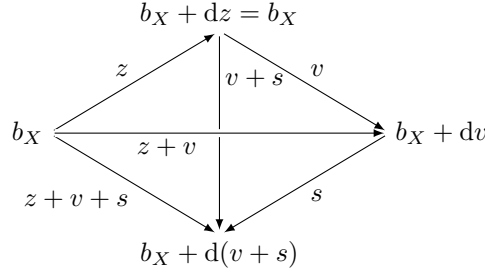
\[
\begin{gaugetetrahedron}{A_X}{A_X+\dd z=A_X}{A_X+\dd v}{A_X+\dd(v+s)}
  \gaugetetrahedronedges{z}{v}{s}{z+v}{v+s}{z+v+s}
\end{gaugetetrahedron}
\]
Its four oriented triangles give, in order,
\begin{equation}
  \begin{array}{c@{\quad\longleftrightarrow\quad}l}
    \lbrack123\rbrack & \Theta^2(s,v;A_X),\\
    \lbrack023\rbrack & \Theta^2(s,v+z;A_X),\\
    \lbrack013\rbrack & \Theta^2(v+s,z;A_X),\\
    \lbrack012\rbrack & \Theta^2(v,z;A_X).
  \end{array}
  \label{eq:theta3-face-identification}
\end{equation}
By Eq.~\eqref{eq:theta-3-descent}, the total phase difference is
\begin{align}
  \dd\Theta^3(s,v,z;A_X)
  &={}\Theta^2(s,v;A_X)-\Theta^2(s,v+z;A_X)\notag\\
  &\quad+\Theta^2(v+s,z;A_X)-\Theta^2(v,z;A_X).
  \label{eq:boundary-theta3-coherence}
\end{align}
Its integral over the closed boundary $X$ vanishes. Thus
\begin{equation}
  U(s)\mathcal S(z)=\mathcal S(z)U(s),
  \qquad
  U(s)\Pi_{\mathrm{sym}}=\Pi_{\mathrm{sym}}U(s).
\end{equation}
This is the geometric expression of the operator argument in the main text: physical hopping commutes with the full Gauss action in the coupled space. In the matter representation at fixed background, the same relation becomes the compatibility of the phases carried by the four faces of a tetrahedron.

\clearpage
\section{Boundary states of SPT and DW theories}
\label{app:spt-boundary-correspondence}

Subsection~\ref{subsec:spt-anomalous-boundary-symmetry} identifies $\mathcal S_A$ in the main text with the effective boundary symmetry of the SPT phase constructed from $\omega$, thereby determining its anomaly. Subsections~\ref{subsec:spt-dw-projection} and~\ref{subsec:spt-hopping-lift} use SPT wavefunctions to express the correspondence of states and hopping operators, providing another description of the gauge-fixing construction. We retain the sphere--disk setting of the main text and suppress overall normalization factors independent of the state labels.

\subsection{SPT wavefunctions and boundary symmetry}
\label{subsec:spt-anomalous-boundary-symmetry}

On a closed sphere $M\simeq S^{d+1}$, let $r\in C^{q-1}(M,G)$ label independent matter degrees of freedom, and let $A\in Z^q(M,G)$ be a fixed flat gauge field. The SPT wavefunction constructed from $\omega$ is
\begin{equation}
  \ket{\Psi}_{\mathrm{SPT}}
  :=\sum_{r\in C^{q-1}(M,G)}
  \exp\left(2\pi\ii\int_M\Theta^1(r;A)\right)\ket r.
  \label{eq:spt-wavefunction}
\end{equation}
For $M\simeq D^{d+1}$, fixing $r_X=v$ gives the boundary state
\begin{equation}
  \ket{\Psi_v}_{\mathrm{SPT}}
  :=\sum_{r_X=v}
  \exp\left(2\pi\ii\int_M\Theta^1(r;A)\right)\ket r.
  \label{eq:spt-boundary}
\end{equation}
States with different $v$ have disjoint support in the basis and equal norms. After normalization, they therefore form an orthonormal basis of the boundary matter space:
\begin{equation}
  \ket v_X\quad\longleftrightarrow\quad\ket{\Psi_v}_{\mathrm{SPT}}.
  \label{eq:spt-boundary-matter-correspondence}
\end{equation}

The global SPT symmetry acts by pure shifts on the original matter basis:
\begin{equation}
  \mathcal S_{\mathrm{SPT}}(\widetilde z)\ket r
  =\ket{r+\widetilde z},\qquad
  \widetilde z\in Z^{q-1}(M,G).
  \label{eq:spt-onsite-translation}
\end{equation}
Integrating Eq.~\eqref{eq:theta-2-descent} over $M$ gives
\begin{equation}
  \int_M\Theta^1(r+\widetilde z;A)
  =\int_M\Theta^1(r;A)+\int_M\Theta^1(\widetilde z;A)
    -\int_X\Theta^2(v,z;A_X),
  \qquad r_X=v,\quad \widetilde z_X=z.
  \label{eq:spt-amplitude-closed-translation}
\end{equation}
A pure shift changes the boundary label in Eq.~\eqref{eq:spt-boundary} from $v$ to $v+z$. Comparing the resulting wavefunction with the previously defined $\ket{\Psi_{v+z}}_{\mathrm{SPT}}$, the ratio of their coefficients on the same basis vector gives the phase of the effective boundary action:
\begin{equation}
  \mathcal S_{\mathrm{SPT}}(\widetilde z)\ket{\Psi_v}_{\mathrm{SPT}}
  =\exp\left(2\pi\ii\left[
    \int_X\Theta^2(v,z;A_X)-\int_M\Theta^1(\widetilde z;A)
  \right]\right)\ket{\Psi_{v+z}}_{\mathrm{SPT}},
  \qquad z=\widetilde z_X.
  \label{eq:spt-global-transformation}
\end{equation}
A closed shift with zero boundary value has a trivial interior stabilizer phase, and $\Theta^2(v,0;A_X)=0$, so it acts trivially on boundary states. This also shows that the effective symmetry depends only on $z$, not on its closed bulk extension.

On a closed sphere, Eq.~\eqref{eq:spt-amplitude-closed-translation} has no boundary term, and the stabilizer phase of a closed parameter is also $1$. Hence
\begin{equation}
  \mathcal S_{\mathrm{SPT}}(\widetilde z)\ket{\Psi}_{\mathrm{SPT}}
  =\ket{\Psi}_{\mathrm{SPT}}.
  \label{eq:spt-sphere-symmetry}
\end{equation}
Equation~\eqref{eq:spt-global-transformation} agrees with Eq.~\eqref{eq:SPT anomalous symmetry} in the main text. Under the correspondence~\eqref{eq:spt-boundary-matter-correspondence}, $\mathcal S_A$ is therefore precisely the boundary symmetry of this SPT phase, with anomaly $[\omega]$.

\subsection{Symmetry projection and DW states}
\label{subsec:spt-dw-projection}

The symmetry projector in the SPT matter space is
\begin{equation}
  \Pi_{\mathrm{SPT}}
  :=\frac{1}{|Z^{q-1}(M,G)|}
  \sum_{\widetilde z\in Z^{q-1}(M,G)}\mathcal S_{\mathrm{SPT}}(\widetilde z).
  \label{eq:spt-symmetry-projector}
\end{equation}
Two matter labels $r,r'$ give the same $A+\dd r$ if and only if they differ by a closed cochain. In the present sphere--disk setting, every flat field has the form $A+\dd r$. The projected space is therefore identified, through
\begin{equation}
  \Pi_{\mathrm{SPT}}\ket r
  \quad\longleftrightarrow\quad\ket{A+\dd r}
  \label{eq:spt-gauge-field-identification}
\end{equation}
with the kinematic gauge-field space $\mathcal H_{\mathrm{DW}}$. This identification describes gauging in the present topologically trivial sector; the DW Gauss conditions are further implemented by the wavefunction amplitudes.

At the same fixed gauge field $A$, define auxiliary vectors in the DW boundary state space already constructed, in order to compare the projected SPT boundary states:
\begin{equation}
  \begin{aligned}
  \ket{\Psi_v}
  &:=T_\partial(v)\Pi_{\mathrm{DW},\partial}\ket A\\
  &=\exp\left(2\pi\ii\int_M\Theta^1(\widetilde v;A)\right)
    \Pi_{\mathrm{DW},\partial}\ket{A+\dd\widetilde v},
  \qquad \widetilde v_X=v.
  \end{aligned}
  \label{eq:relative-gauge-orbit-wavefunction}
\end{equation}
Here $\widetilde v$ is any bulk extension; the interior Gauss conditions ensure independence of this choice. This definition selects a vector representative of the physical state line with configuration $A_X+\dd v$; $v$ does not become an independent physical degree of freedom in DW theory. Equation~\eqref{eq:gauge-composition-trivial-boundary-leg} further gives
\begin{equation}
  \ket{\Psi_v}
  \ \propto\
  \sum_{r_X=v}T(r)\ket A
  =\sum_{r_X=v}
    \exp\left(2\pi\ii\int_M\Theta^1(r;A)\right)\ket{A+\dd r}.
  \label{eq:spt-projected-boundary-amplitude}
\end{equation}
Thus Eq.~\eqref{eq:spt-gauge-field-identification} yields, on SPT boundary states,
\begin{equation}
  \Pi_{\mathrm{SPT}}\ket{\Psi_v}_{\mathrm{SPT}}
  \quad\longleftrightarrow\quad\ket{\Psi_v}.
  \label{eq:spt-dw-boundary-projection}
\end{equation}

Different auxiliary vectors for the same configuration remain related by phases. By Eq.~\eqref{eq:spatial-gauge-composition-operator}, for closed $z\in Z^{q-1}(X,G)$ and a closed extension $\widetilde z$,
\begin{equation}
  \ket{\Psi_{v+z}}
  =\exp\left(2\pi\ii\left[
    \int_M\Theta^1(\widetilde z;A)-\int_X\Theta^2(v,z;A_X)
  \right]\right)\ket{\Psi_v}.
  \label{eq:different frame same configuration state}
\end{equation}
This phase is the inverse of the symmetry phase in Eq.~\eqref{eq:spt-global-transformation}. It ensures that $\Pi_{\mathrm{SPT}}\mathcal S_{\mathrm{SPT}}=\Pi_{\mathrm{SPT}}$ is compatible with the projection above. Under Eq.~\eqref{eq:spt-boundary-matter-correspondence}, symmetry projection on the boundary likewise gives
\begin{equation}
  \Pi_{\mathrm{sym}}\ket v_X
  \quad\longleftrightarrow\quad
  \Pi_{\mathrm{SPT}}\ket{\Psi_v}_{\mathrm{SPT}}
  \quad\longleftrightarrow\quad\ket{\Psi_v}.
  \label{eq:spt-symmetric-dw-correspondence}
\end{equation}

\subsection{Hopping in the SPT matter space}
\label{subsec:spt-hopping-lift}

The SPT matter operator corresponding to the DW operator $T(t)$ is
\begin{equation}
  \widetilde T(t)\ket r
  :=\exp\left(2\pi\ii\int_M\Theta^1(t;A+\dd r)\right)\ket{r+t}.
  \label{eq:spt-lifted-hopping}
\end{equation}
Its phase depends only on $\dd r$, so it commutes with closed shifts:
\begin{equation}
  \widetilde T(t)\mathcal S_{\mathrm{SPT}}(\widetilde z)
  =\mathcal S_{\mathrm{SPT}}(\widetilde z)\widetilde T(t).
  \label{eq:spt-lifted-hopping-symmetry}
\end{equation}
It therefore preserves the symmetric subspace and corresponds to $T(t)$ under Eq.~\eqref{eq:spt-gauge-field-identification}. Meanwhile, Eq.~\eqref{eq:theta-2-descent} combines the phase of $\widetilde T(t)$ with the two bulk phases of the boundary wavefunctions to give
\begin{equation}
  \widetilde T(t)\ket{\Psi_v}_{\mathrm{SPT}}
  =\exp\left(2\pi\ii\int_X\Theta^2(s,v;A_X)\right)
    \ket{\Psi_{v+s}}_{\mathrm{SPT}},\qquad s=t_X.
  \label{eq:spt-lifted-boundary-hopping}
\end{equation}
This action depends only on $s$ and gives exactly $U_A(s)$ in the main text under Eq.~\eqref{eq:spt-boundary-matter-correspondence}. Applying $\Pi_{\mathrm{SPT}}$ to the boundary states yields the DW physical hopping in auxiliary coordinates:
\begin{equation}
  T_\partial(s)\ket{\Psi_v}
  =\exp\left(2\pi\ii\int_X\Theta^2(s,v;A_X)\right)
    \ket{\Psi_{v+s}}.
  \label{eq:relative-frame-hopping}
\end{equation}
Thus $U_A(s)$, $\widetilde T(t)$ restricted to the SPT boundary state space, and $T_\partial(s)$ describe the same physical hopping. This correspondence complements the operator relations obtained through the two gauge-fixing choices in the main text.

\section{Anyon hopping and the symmetric condition}
\label{app:anyon-triangle}

We first state and prove Theorem~\ref{thm:anyon-closed-rigidity}, then construct hopping operators satisfying the symmetric condition of Section~\ref{sec:statistical-obstruction} when $b_Q=0$.

\subsection{The anomaly of closed hopping}
\label{subsec:anyon-lifts}

\begin{theorem}
\label{thm:anyon-closed-rigidity}
Let $X$ be a finite simplicial triangulation of $S^2$, and let $G$ be a finite Abelian group. We use the canonical matter space, configuration states, and elementary hopping operators constructed in Section~\ref{sec:what-is-hopping} for the statistical class $[\omega]\in H^4(B^2G,\RZ)$. The elementary hopping operators are given by Eq.~\eqref{eq:hopping operator formula}, with background gauge field $A=0$:
\begin{equation}
 U(s)\ket v
 =\exp\left(2\pi\ii\int_X\Theta^2(s,v;0)\right)\ket{v+s},
 \qquad v\in C^1(X,G).
 \label{eq:anyon-elementary-hopping}
\end{equation}
Denote the span of the configuration states by $\mathcal H_{\mathrm{conf}}$.

\textbf{Existence.} For any such $X$, there is a family of closed hopping operators $\{U_0(z):z\in Z^1(X,G)\}$ obtained by admissible concatenation according to Eq.~\eqref{eq:extended-hopping-consistency}, which strictly satisfies the group law on the full matter space,
\begin{equation}
 U_0(z)U_0(w)=U_0(z+w),
 \label{eq:anyon-closed-group-law}
\end{equation}
and has $1$-form anomaly $-[\omega]$.

\textbf{Uniqueness.} Suppose further that $X$ contains two triangles with disjoint vertex sets. Let $U(z)$ be any other family of closed hopping operators obtained by admissible concatenation according to Eq.~\eqref{eq:extended-hopping-consistency}, satisfying the group law on $\mathcal H_{\mathrm{conf}}$, and whose resulting operators have the following form on the full matter space:
\[
 U(z)\ket v=\e^{2\pi\ii u_z(v)}\ket{v+z}.
\]
Then there exists an FDQC $\Phi$, independent of $z$, such that
\[
 U(z)=\Phi U_0(z)\Phi^{-1}
\]
on the full matter space. Thus, in local realizations that define a $1$-form symmetry, the anomaly of $U(z)$ is also $-[\omega]$.
\end{theorem}

\begin{proof}
Write $L(s,v)=\Theta^2(s,v;0)$. We first construct closed hopping with anomaly $-[\omega]$ and then prove uniqueness.

Section~\ref{sec:anyon-statistics-anomaly} already constructs a reference closed action for mutual braiding. For two distinct cyclic factors $\ZZ_{N_i},\ZZ_{N_j}$, take $N=\gcd(N_i,N_j)$ and reduce both variables in the phase to $\ZZ_N$. The calculation in Eqs.~\eqref{eq:ordinary-mutual-descendant}--\eqref{eq:ordinary-mutual-opposite-anomaly} applies unchanged.

For a single $\ZZ_N$ factor, set $N_*=N\gcd(N,2)$ and take the generator
\[
 \omega_N(A)=\frac{1}{N_*}
 \left(\overline A\smile\overline A+\overline A\smile_1\dd\overline A\right),
\]
where an overline denotes the integer representative in $0,\ldots,N-1$.
The descent construction of Appendix~\ref{app:descent-construction} gives $L_{\omega_N}(s,v)=\Theta^2(s,v;0)$, which can be written using cup and cup-$1$ products as
\begin{align}
 N_*L_{\omega_N}(s,v)
 ={}&\overline v\smile\overline s
 -(\overline s+\overline v-\overline{s+v})\smile\overline{s+v}\notag\\
 &+\overline v\smile_1
 \bigl(\overline{\dd v}+\dd\overline s-\overline{\dd(s+v)}\bigr)
 \quad\pmod{N_*}.
 \label{eq:anyon-prism-numerator}
\end{align}
For a closed parameter $\dd z=0$, we will prove
\begin{equation}
 L_{\omega_N}(z,v)+L_{\omega_N}(v,z)
 -\frac{\overline{\dd v}\smile_1\overline z}{N_*}
 =-\frac{\dd(\overline v\smile_1\overline z)}{N_*}.
 \label{eq:anyon-cyclic-reference}
\end{equation}

We use the cup-$1$ coboundary identity
\[
 \dd(\overline v\smile_1\overline z)
 =\dd\overline v\smile_1\overline z
 -\overline v\smile_1\dd\overline z
 -\overline v\smile\overline z
 -\overline z\smile\overline v.
\]
Since $\overline z+\overline v-\overline{z+v}$ is a multiple of $N$, the second terms of Eq.~\eqref{eq:anyon-prism-numerator} cancel modulo $N_*$ when the two descent phases are added. Combining the remaining terms with the coboundary identity above gives
\begin{align}
 &N_*\bigl[L_{\omega_N}(z,v)+L_{\omega_N}(v,z)\bigr]
 -\overline{\dd v}\smile_1\overline z
 +\dd(\overline v\smile_1\overline z)\notag\\
 &\qquad\equiv
 (\dd\overline v-\overline{\dd v})\smile_1\overline z
 +\overline z\smile_1(\dd\overline v-\overline{\dd v})
 \equiv0\pmod{N_*}.
 \label{eq:anyon-integer-remainder}
\end{align}
The last step uses the cup-$1$ identity in degrees $(2,1)$: the sum of the two terms on the right contains the factors $\dd\overline v-\overline{\dd v}$ and $\dd\overline z$. Both are multiples of $N$, so their product is divisible by $N^2$. Since $N_*\mid N^2$, this proves Eq.~\eqref{eq:anyon-cyclic-reference}.

We therefore add the following term to the local hopping phase:
\[
 -\frac{\overline a\smile_1\overline s}{N_*},\qquad a=\dd v.
\]
Since $X$ is closed, the total coboundary on the right of Eq.~\eqref{eq:anyon-cyclic-reference} integrates to zero, and the modified closed hopping phase is $-\int_XL_{\omega_N}(v,z)$. Following Eq.~\eqref{eq:SPT anomalous symmetry}, we further multiply the closed hopping by the overall phase
\[
 \exp\!\left(2\pi\ii\int_M\Theta^1_{\omega_N}(\widetilde z;0)\right).
\]
The resulting operator is precisely $\mathcal S_{-\omega_N}(z)$, so the group law holds strictly. This overall phase can be assigned to one nonzero elementary factor in the chosen admissible concatenation; for $z=0$, take the identity operator.
The background correction depends only on $a$ within the same triangle, preserves the configuration axiom, and gives admissible extended hopping as in Subsection~\ref{subsec:hopping-assembly}.
For a general finite Abelian group $G=\bigoplus_i\ZZ_{N_i}$, every $[\omega]$ is a linear combination of the self-statistics of the cyclic factors and the mutual braiding between distinct factors. Adding the phases of the self-statistics construction above and the mutual-braiding construction in the main text gives a reference closed action with anomaly $-[\omega]$ for every $[\omega]$.

Denote this closed hopping by $U_0(z)$. For another family $U(z)$, we will construct a diagonal FDQC $\Phi$ such that
\begin{equation}
 U(z)=\Phi U_0(z)\Phi^{-1}.
 \label{eq:anyon-local-conjugacy}
\end{equation}
All phases below take values in $\RZ$. Absorbing orientation signs into $a_\Delta\in G$, configurations on the sphere satisfy
\begin{equation}
 a\in B^2(X,G)\quad\Longleftrightarrow\quad
 \sum_\Delta a_\Delta=0,
 \label{eq:anyon-total-background}
\end{equation}
and every closed parameter can be written as $z=\dd\lambda$. The phase difference between the two actions has the form
\begin{equation}
 U(z)=\exp\!\left(2\pi\ii K(z;a)\right)U_0(z),\qquad
 K(z;a)=\sum_{\Delta:\,z|_\Delta\ne0}k_{\Delta,z}(a_\Delta).
 \label{eq:anyon-relative-phase}
\end{equation}
Here $k_{\Delta,z}$ may depend on all of $z$.

This decomposition follows from locality. For an edge $e$, Eq.~\eqref{eq:locality-axiom-redefinition} forces its redefinition phase to depend only on $a_P,a_Q$ in the adjacent triangles $P,Q$: anyons outside these triangles can move freely in the complement of $P\cup Q$. Choose one other edge in each of $P,Q$. The mixed difference in $a_P,a_Q$ generated by these two hoppings vanishes, because the two chosen edges and $e$ do not lie in a common triangle. Thus the redefinition phase splits into a sum of two single-triangle functions. The concatenation phase in Eq.~\eqref{eq:boundary-hopping-composition} has the same decomposition. Since Eq.~\eqref{eq:extended-hopping-consistency} does not use edges with zero coefficient, the sum has the indicated range.

The two actions have the same shift, and every canonical configuration state has nonzero coefficients throughout the orbit $\dd v=a$. Thus Eq.~\eqref{eq:anyon-relative-phase} holds on the full matter space. The group law gives
\begin{equation}
 K(z+w;a)=K(z;a)+K(w;a).
 \label{eq:anyon-relative-additivity}
\end{equation}
Since $K(z;0)$ is a group homomorphism, the overall phase redefinition
\[
 U(z)\longmapsto\e^{-2\pi\ii K(z;0)}U(z)
\]
preserves the group law. We may therefore assume $K(z;0)=0$ without loss of generality and normalize the local terms so that $k_{\Delta,z}(0)=0$.

For $x\in\Delta$, write the local term of $K(\dd(g\delta_x);a)$ as $k_{\Delta,x}(g;a_\Delta)$, with $k_{\Delta,x}(0;t)=0$. Place an anyon $t$ in $\Delta$ and $-t$ in a triangle $E$ not containing $x$, with all other anyons set to zero. Such an $E$ can be chosen from the two disjoint triangles in the hypothesis. Only $\Delta$ contributes to the phase, so Eq.~\eqref{eq:anyon-relative-additivity} directly implies that $k_{\Delta,x}(g;t)$ is additive in $g$. Decomposing $\lambda$ into single-vertex cochains gives
\begin{equation}
 K(\dd\lambda;a)=\sum_\Delta\sum_{x\in\Delta}
 k_{\Delta,x}(\lambda_x;a_\Delta).
 \label{eq:anyon-vertex-expansion}
\end{equation}

Next, take triangles $D,E$ with disjoint vertex sets and place anyons $t,-t$ on them. Set $\lambda$ to $g$ at the vertices of $D$ and to zero at all other vertices. Since $\dd\lambda|_D=\dd\lambda|_E=0$, Eq.~\eqref{eq:anyon-relative-phase} gives $K=0$, while Eq.~\eqref{eq:anyon-vertex-expansion} gives
\[
 \sum_{x\in D}k_{D,x}(g;t)=0.
\]
Now place anyons $t,-t$ on any $\Delta\ne D$ and on $D$, respectively, and take $\lambda=g$ globally constant. Then $\dd\lambda=0$, and the contribution from $D$ also vanishes, so
\begin{equation}
 \sum_{x\in\Delta}k_{\Delta,x}(g;t)=0.
 \label{eq:anyon-constant-potential}
\end{equation}

For an ordered triangle $\Delta=(012)$, define
\begin{equation}
 \begin{aligned}
 f_\Delta(v)&=k_{\Delta,1}(v_{01};a_\Delta)
              +k_{\Delta,2}(v_{02};a_\Delta),\\
 F(v)&=\sum_\Delta f_\Delta(v).
 \end{aligned}
 \label{eq:anyon-local-rephasing}
\end{equation}
A closed shift preserves $a_\Delta$. By additivity and Eq.~\eqref{eq:anyon-constant-potential},
\[
 \begin{aligned}
 f_\Delta(v+\dd\lambda)-f_\Delta(v)
 &=k_{\Delta,1}(\lambda_1-\lambda_0;a_\Delta)
   +k_{\Delta,2}(\lambda_2-\lambda_0;a_\Delta)\\
 &=\sum_{x\in\Delta}k_{\Delta,x}(\lambda_x;a_\Delta).
 \end{aligned}
\]
Summing gives $F(v+z)-F(v)=K(z;a)$.

Thus
\[
 \ket v\longmapsto\exp\!\left(2\pi\ii F(v)\right)\ket v
\]
is a diagonal FDQC relating the two normalized actions. Composing with the inverse of the initial basis change gives Eq.~\eqref{eq:anyon-local-conjugacy}, so the anomaly of $U(z)$ is $-[\omega]$.
\end{proof}

\subsection{Sufficiency of the symmetric condition}
\label{app:anyon-symmetric-sufficiency}

We now prove that $b_Q=0$ is sufficient for the symmetric condition to be realizable. In this case, $Q=\chi/2$ for a homomorphism $\chi:G\to\ZZ_2$. Take $\omega(A)=\frac12\chi(A)\smile\chi(A)$. The hopping operator from Eq.~\eqref{eq:hopping operator formula} has phase $(-1)^{\int_X\chi(v)\smile\chi(s)}$; this follows directly from Eq.~\eqref{eq:anyon-prism-numerator} with $N=2$ and twice the generator. Multiply this hopping operator by the local phase $(-1)^{\int_X\chi(a)\smile_1\chi(s)}$, where $a=\dd v$. The cup-$1$ identity then puts the resulting operator in the form
\begin{equation}
 U(s)\ket v=(-1)^{\int_X[\chi(s)\smile\chi(v)
 +\chi(v)\smile_1\dd\chi(s)]}\ket{v+s}.
 \label{eq:anyon-self-construction}
\end{equation}
The added phase depends only on the configuration within the same triangle. This is therefore an admissible redefinition in the sense of Subsection~\ref{subsec:hopping-assembly}, preserving both the configuration axiom and the statistics. The exponent is separately additive in $s,v$, so any elementary concatenation differs from this formula only by an overall phase depending on the parameters. The formula thus defines admissible extended hopping. When $s=z$ is closed, the second term in Eq.~\eqref{eq:anyon-self-construction} vanishes, and the closed operator $U(z)$ is precisely the canonical symmetry that commutes with all hopping operators. Hence Eq.~\eqref{eq:symmetric-extended-hopping} holds.

\clearpage
\section{Independence of statistics and anomaly for two-dimensional loops}
\label{app:loop-full-pattern}

This appendix supplements the construction in Section~\ref{sec:statistical-obstruction} and shows how statistics and the anomaly of closed hopping operators can take independent values.

The two generators of $H^4(B\ZZ_2^2,\RZ)=\ZZ_2^2$ can be chosen as
\[
 \omega_A=\frac12a_A^3a_B,\quad
 \omega_B=\frac12a_Aa_B^3.
\]

Each statistical class admits a simple closed action. We first give the corresponding open hopping. For $\omega_A$, the descent construction in Section~\ref{sec:what-is-hopping} gives
\[
 L(s,v)=\frac12v_A\bigl[
 s_A(a_A+\dd s_A)(a_B+\dd s_B)
 +a_As_A(a_B+\dd s_B)+a_A^2s_B\bigr].
\]
When $s=z$ is constant, the first two terms cancel, and the closed action is
\[
 U(z)\ket v
 =(-1)^{z_B\int_Xv_Aa_A^2}\ket{v+z}.
\]
This closed action is the canonical symmetry representing $\frac12a_Ba_A^3$, so its anomaly is also $[\omega_A]$. Exchanging $A,B$ gives the other generator. Adding the phases then realizes any prescribed $[\omega]$ as both the statistics and the closed anomaly.

Admissible background phase corrections can independently change the anomaly while preserving the statistics. Thus, for any $[\omega],[\omega']\in H^4(B\ZZ_2^2,\RZ)$, we can construct hopping with statistics $[\omega]$ whose admissible closed action satisfies the group law and has anomaly $[\omega']$.

To change the anomaly independently, start with the pure shift $\ket v\mapsto\ket{v+s}$, and use overlines for integer lifts. For $\ell_A,\ell_B\in\ZZ_2$, add the following background phase to elementary hopping:
\begin{equation}
 U_\ell(s)\ket v
 =\exp\left(2\pi\ii\int_X
 \frac{\overline{a_A}\overline{a_B}(\overline{\ell_A}\,\overline{s_A}+\overline{\ell_B}\,\overline{s_B})}4\right)\ket{v+s}.
 \label{eq:loop-background-construction}
\end{equation}
Each term depends only on the background and parameters within the same triangle, so this is an admissible redefinition and preserves trivial statistics. Its value for general $s$ can be assembled from elementary factors: the intermediate backgrounds are $a+\dd s_1+\cdots+\dd s_j$, and the differences produced when combining phases remain background functions on the same triangle. They can therefore be absorbed into redefinitions of the corresponding factors. When the parameter vanishes on a triangle, no correction is present there.

A closed parameter $z$ is constant. For any $\ZZ_2$-valued $0$-cochains $x,y$,
\begin{equation}
 \int_X\frac{\overline{\dd x}\,\overline{\dd y}}4
 =\frac12\int_X x(\dd x+\dd y)\dd y\quad\pmod1.
 \label{eq:quarter-domain-wall-product-local-identity}
\end{equation}
This identity follows by expanding $\overline{\dd x}=\dd\overline{x}-2\overline{x}\smile\dd\overline{x}$ and integrating by parts. The right-hand side takes only the values $0,\frac12$, so the closed action in Eq.~\eqref{eq:loop-background-construction} satisfies the $\ZZ_2^2$ group law. For $\ell_A=1,\ell_B=0$, its two generators are
\[
 U_{(1,0)}(\one_A)\ket v
 =(-1)^{\int_Xv_A(\dd v_A+\dd v_B)\dd v_B}\ket{v+\one_A},
 \qquad
 U_{(1,0)}(\one_B)\ket v=\ket{v+\one_B}.
\]
Descent applied to $\frac12a_A^2(a_A+a_B)a_B$ gives exactly this symmetry. Since $\frac12a_A^2a_B^2$ is a coboundary, its anomaly is $[\omega_A]$. Exchanging $A,B$ gives $[\omega_B]$. Thus trivial statistics can already coexist with any closed anomaly $[\omega']\in H^4(B\ZZ_2^2,\RZ)$.

To obtain a prescribed $[\omega']$, choose $\ell_A,\ell_B$ such that
\[
 \ell_A[\omega_A]+\ell_B[\omega_B]=[\omega']-[\omega].
\]
Add the admissible background correction of Eq.~\eqref{eq:loop-background-construction} with these coefficients to the hopping obtained by descent from $[\omega]$. This preserves the statistics $[\omega]$ while changing the closed anomaly to $[\omega']$.



\begin{thebibliography}{10}

\bibitem{Adler1969AxialAnomaly}
Stephen~L. Adler.
\newblock Axial-vector vertex in spinor electrodynamics.
\newblock {\em Physical Review}, 177(5):2426--2438, 1969.

\bibitem{BellJackiw1969PCAC}
J.~S. Bell and R.~Jackiw.
\newblock A {PCAC} puzzle: {$\pi^0\to\gamma\gamma$} in the {$\sigma$}-model.
\newblock {\em Il Nuovo Cimento A}, 60(1):47--61, 1969.

\bibitem{tHooft1980AnomalyMatching}
Gerard {'t Hooft}.
\newblock Naturalness, chiral symmetry, and spontaneous chiral symmetry
  breaking.
\newblock In Gerard {'t Hooft}, C.~Itzykson, A.~Jaffe, H.~Lehmann, P.~K.
  Mitter, I.~M. Singer, and R.~Stora, editors, {\em Recent Developments in
  Gauge Theories}, volume~59 of {\em NATO Advanced Study Institutes Series B:
  Physics}, pages 135--157. Springer US, 1980.

\bibitem{GaiottoEtAl2015GeneralizedSymmetries}
Davide Gaiotto, Anton Kapustin, Nathan Seiberg, and Brian Willett.
\newblock Generalized global symmetries.
\newblock {\em Journal of High Energy Physics}, 2015(2):172, 2015.

\bibitem{Wilczek1982FractionalSpin}
Frank Wilczek.
\newblock Quantum mechanics of fractional-spin particles.
\newblock {\em Physical Review Letters}, 49(14):957--959, 1982.

\bibitem{Kitaev2003ToricCode}
A.~Yu. Kitaev.
\newblock Fault-tolerant quantum computation by anyons.
\newblock {\em Annals of Physics}, 303(1):2--30, 2003.

\bibitem{Kitaev2006Anyons}
Alexei Kitaev.
\newblock Anyons in an exactly solved model and beyond.
\newblock {\em Annals of Physics}, 321(1):2--111, 2006.

\bibitem{WangLevin2014LoopBraiding}
Chenjie Wang and Michael Levin.
\newblock Braiding statistics of loop excitations in three dimensions.
\newblock {\em Physical Review Letters}, 113(8):080403, 2014.

\bibitem{Fidkowski_2022}
Lukasz Fidkowski, Jeongwan Haah, and Matthew~B. Hastings.
\newblock Gravitational anomaly of $(3+1)$-dimensional {$\mathbb{Z}_2$} toric
  code with fermionic charges and fermionic loop self-statistics.
\newblock {\em Physical Review B}, 106(16):165135, October 2022.

\bibitem{Kobayashi_2026}
Ryohei Kobayashi, Yuyang Li, Hanyu Xue, Po-Shen Hsin, and Yu-An Chen.
\newblock Generalized statistics on lattices.
\newblock {\em Physical Review X}, 16(1), January 2026.

\bibitem{Xue2026Statistics}
Hanyu Xue.
\newblock Statistics of {A}belian topological excitations.
\newblock {\em Physical Review B}, 113:045143, 2026.

\bibitem{Feng_2026}
Yitao Feng, Hanyu Xue, Yuyang Li, Meng Cheng, Ryohei Kobayashi, Po-Shen Hsin,
  and Yu-An Chen.
\newblock Anyonic membranes and pontryagin statistics.
\newblock {\em Physical Review Letters}, 136(8), February 2026.

\bibitem{FengEtAl2025HigherFormAnomalies}
Yitao Feng, Ryohei Kobayashi, Yu-An Chen, and Shinsei Ryu.
\newblock Higher-form anomalies on lattices.
\newblock {\em Physical Review Letters}, 136(4):046504, 2026.

\bibitem{HsinChen2026Bockstein}
Po-Shen Hsin and Yu-An Chen.
\newblock Bockstein braiding statistics, 2026.

\bibitem{ElseNayak2014}
Dominic~V. Else and Chetan Nayak.
\newblock Classifying symmetry-protected topological phases through the
  anomalous action of the symmetry on the edge.
\newblock {\em Physical Review B}, 90:235137, 2014.

\bibitem{LevinWen2003}
Michael Levin and Xiao-Gang Wen.
\newblock Fermions, strings, and gauge fields in lattice spin models.
\newblock {\em Physical Review B}, 67(24):245316, 2003.

\bibitem{KawagoeLevin2020}
Kyle Kawagoe and Michael Levin.
\newblock Microscopic definitions of anyon data.
\newblock {\em Physical Review B}, 101(11):115113, 2020.

\bibitem{KawagoeLevin2021Anomalies}
Kyle Kawagoe and Michael Levin.
\newblock Anomalies in bosonic symmetry-protected topological edge theories:
  Connection to {F} symbols and a method of calculation.
\newblock {\em Physical Review B}, 104(11):115156, 2021.
\newblock \href{https://doi.org/10.1103/PhysRevB.104.115156}{doi:10.1103/PhysRevB.104.115156}.

\bibitem{XueWen2026Holographic}
Hanyu Xue and Xiao-Gang Wen.
\newblock Holographic theory of mixed-dimensional statistics and
  conservation-encoding hopping-operator algebras, 2026.

\bibitem{EilenbergMacLane1954}
Samuel Eilenberg and Saunders Mac~Lane.
\newblock On the groups {$H(\Pi,n)$}, {II}: Methods of computation.
\newblock {\em Annals of Mathematics}, 60(1):49--139, 1954.

\bibitem{JoyalStreet1993}
Andr\'e Joyal and Ross Street.
\newblock Braided tensor categories.
\newblock {\em Advances in Mathematics}, 102(1):20--78, 1993.

\bibitem{ChatterjeeWen2023}
Arkya Chatterjee and Xiao-Gang Wen.
\newblock Symmetry as a shadow of topological order and a derivation of
  topological holographic principle.
\newblock {\em Physical Review B}, 107:155136, 2023.

\bibitem{BolsEtAl2026}
Alex Bols, Wojciech De~Roeck, Michiel De~Wilde, and Bruno de~O.~Carvalho.
\newblock Classification of locality preserving symmetries on spin chains.
\newblock {\em Communications in Mathematical Physics}, 407(1):10, 2026.

\bibitem{ShiKatoKim2020}
Bowen Shi, Kohtaro Kato, and Isaac~H. Kim.
\newblock Fusion rules from entanglement.
\newblock {\em Annals of Physics}, 418:168164, 2020.

\bibitem{Shi2020Verlinde}
Bowen Shi.
\newblock Verlinde formula from entanglement.
\newblock {\em Physical Review Research}, 2(2):023132, 2020.

\bibitem{KimRanard2024}
Isaac~H. Kim and Daniel Ranard.
\newblock Classifying 2d topological phases: Mapping ground states to
  string-nets, 2024.

\bibitem{feng2026paulistabilizerformalismtopological}
Yitao Feng, Hanyu Xue, Ryohei Kobayashi, Po-Shen Hsin, and Yu-An Chen.
\newblock Pauli stabilizer formalism for topological quantum field theories and
  generalized statistics, 2026.

\bibitem{ChenGuLiuWen2013}
Xie Chen, Zheng-Cheng Gu, Zheng-Xin Liu, and Xiao-Gang Wen.
\newblock Symmetry protected topological orders and the group cohomology of
  their symmetry group.
\newblock {\em Physical Review B}, 87(15):155114, 2013.

\bibitem{LessaChengWang}
Leonardo~A. Lessa, Meng Cheng, and Chong Wang.
\newblock Mixed-state quantum anomaly and multipartite entanglement.
\newblock {\em Physical Review X}, 15:011069, 2025.

\end{thebibliography}
\end{document}